\newif\ifarxiv
\arxivtrue

\ifarxiv
   \documentclass[manuscript,screen,nonacm,natbib=false]{acmart}
  \setcopyright{none}
  \acmDOI{}
  \acmISBN{}
  \acmConference[arXiv]{arXiv preprint}{2026}{}
  \acmYear{2026}
\else
  \documentclass[manuscript, screen, review]{jair}
  \setcopyright{cc}
  \acmDOI{10.1613/jair.1.xxxxx}

  \JAIRAE{Insert JAIR AE Name}
  \JAIRTrack{} 
  \acmVolume{0}
  \acmArticle{0}
  \acmMonth{0}
  \acmYear{2026}
\fi

\ifarxiv
\fi

\RequirePackage[
  datamodel=acmdatamodel,
  style=acmauthoryear,
  backend=biber,
  giveninits=true,
  uniquename=init
  ]{biblatex}

\usepackage{booktabs}
\usepackage{array}
\usepackage{tabularx}
\usepackage{adjustbox}
\usepackage[ruled]{algorithm2e}
\usepackage{xcolor}
\usepackage{balance}
\usepackage{tensor}
\usepackage{enumitem}
\usepackage{arydshln}
\usepackage{mathtools}

\usepackage{comment}
\includecomment{fullversion}

\usepackage{tikz}
\usepackage{placeins}
\usepackage{cleveref}
\usepackage{pgfplots}
\usepackage[bibliography=common]{apxproof}

\usepgfplotslibrary{groupplots}
\pgfplotsset{compat=1.18}

\usetikzlibrary{positioning,calc,shapes.geometric,arrows.meta}
\usepgfplotslibrary{fillbetween}

\SetAlFnt{\small}
\SetAlCapFnt{\small}
\SetAlCapNameFnt{\small}
\SetAlCapHSkip{0pt}
\IncMargin{-\parindent}

\newcommand{\PROP}{\ensuremath{\mathrm{PROP}}}
\newcommand{\PROPc}{\ensuremath{\mathrm{PROP}{c}}}
\newcommand{\PROPx}{\ensuremath{\mathrm{PROP}_{\times}}}

\newcommand{\EF}{\ensuremath{\mathrm{EF}}}

\DeclareMathOperator{\Envy}{Envy}

\newcommand{\EFax}[1]{\mathrm{EF}_{\times #1}}
\newcommand{\Prop}{\mathrm{Prop}}
\newcommand{\M}{\mathcal M}
\newcommand{\PROPax}[1]{\mathrm{PROP}_{\times #1}}

\usetikzlibrary{decorations.pathreplacing,arrows.meta}

\renewcommand{\appendixprelim}[1]{%
  \clearpage
}

\newtheoremrep{theorem}{Theorem}[section]
\newtheoremrep{lemma}[theorem]{Lemma}
\newtheorem{corollary}[theorem]{Corollary}
\newtheoremrep{proposition}[theorem]{Proposition}

\newtheoremrep{claim}[theorem]{Claim}

\theoremstyle{definition}
\newtheoremrep{definition}[theorem]{Definition}
\newtheorem{remark}[theorem]{Remark}

\makeatletter

\let\c@lemma\c@theorem
\let\c@definition\c@theorem
\let\c@proposition\c@theorem
\let\c@remark\c@theorem
\let\c@claim\c@theorem
\let\c@corollary\c@theorem

\makeatother

\makeatletter
\@ifundefined{c@example}{\newtheorem{example}{Example}}{}
\makeatother

\DeclareMathOperator*{\argmax}{arg\,max}

\DeclarePairedDelimiterX\set[1]\lbrace\rbrace{\def\given{\;\delimsize\vert\;}#1}

\makeatletter
\newcommand*{\inlineequation}[2][]{%
  \begingroup
    \refstepcounter{equation}%
    \ifx\\#1\\%
    \else
      \label{#1}%
    \fi
    \relpenalty=10000 %
    \binoppenalty=10000 %
    \ensuremath{%
      #2%
    }%
    \quad\@eqnnum. \space
  \endgroup
}
\makeatother

\newcommand{\dsprop}{\delta}

\newcommand{\EFparam}[1]{\ensuremath{\mathrm{EF}{#1}}}
\newcommand{\EFc}{\EFparam{c}}

\newif\ifshowcomments
\showcommentsfalse

\newcommand{\ind}[1]{\mathbf{1}\!\left\{#1\right\}}
\newcommand{\sdiv}[2]{\frac{#1}{#2+\ind{#2=0}}}

\crefname{claim}{Claim}{Claims}
\Crefname{claim}{Claim}{Claims}
\Crefname{lemma}{Lemma}{Lemmas}

\newcommand{\GenAIDisclosureText}{%
During the preparation of this manuscript, the authors used generative AI
tools, including OpenAI ChatGPT and Google Gemini, in a supporting capacity.
These tools assisted with brainstorming candidate mathematical approaches,
improving the clarity and organization of author-prepared text, and
\LaTeX\ formatting. Their outputs were treated as suggestions and were not
relied upon as authoritative sources or as substitutes for mathematical
verification. The authors made all scientific and editorial decisions,
reviewed and revised all AI-assisted material, verified the mathematical
arguments and cited sources, and take full responsibility for the final
manuscript.%
}

\makeatletter
\newcommand{\GenAIDisclosureInPaper}{}
\if@ACM@anonymous
  \renewcommand{\GenAIDisclosureInPaper}{%
    \section*{Generative AI Usage Disclosure}
    \noindent\GenAIDisclosureText
  }
\fi
\makeatother

\begin{document}

\title[Perpetual Online Fairness]{Perpetual Fully-Online Approximate Fairness}
\ifarxiv
  \subtitle{arXiv preprint}
\fi

\author{Ido Kahana}
\authornote{Corresponding Author.}
\email{ynet.ido@gmail.com}
\affiliation{%
  \institution{Ariel University}
  \city{Ariel}
  \country{Israel}
}

\author{Erel Segal-Halevi}
\email{erelsgl@gmail.com}
\affiliation{%
  \institution{Ariel University}
  \city{Ariel}
  \country{Israel}
}

\author{Noam Hazon}
\email{noamh@ariel.ac.il}
\affiliation{%
  \institution{Ariel University}
  \city{Ariel}
  \country{Israel}
}
\renewcommand{\shortauthors}{Kahana, Segal-Halevi, and Hazon}

\begin{abstract}
Many decision processes run for a long and unknown duration: in each round new requests arrive, an irrevocable choice must be made immediately, and the system is judged by ongoing fairness requirements. 
Examples include food banks allocating donated items as they arrive, computing systems repeatedly scheduling scarce resources across users, and institutions making repeated public decisions (e.g., which proposals or cases to prioritize) while remaining fair over time.

We propose a general approach to such problems based on \emph{deficits}, which measure how far the current outcome is from satisfying each fairness requirement. 
The goal is to keep all deficits small at every time step, without knowing the horizon or the future agent valuations. 
This viewpoint also highlights a natural modeling question for long-running systems: how much of the past should be counted when fairness is evaluated?

We first study the full-history model, where all past rounds count equally. 
Within this model, we propose an efficient fully-online rule. 
For \(n\) agents, we prove anytime guarantees: after any \(t\) rounds, all tracked requirements remain satisfied up to a slack of order \(\tilde O(\sqrt{t/n})\). 
We instantiate the rule for online allocation of indivisible goods, yielding natural relaxations of proportionality and envy-freeness, and for online public decision-making.
We show that the square-root time dependence is unavoidable even for a weak proportionality requirement. We also prove that unrestricted classical envy-freeness up to \(c\) goods is even harder: against adaptive adversaries, its exact worst-case parameter at horizon \(T\) is \(\lceil T/n\rceil\).

We then study discounted-memory fairness, where older deficits carry smaller weight. The same fully-online rule applies to these discounted deficits, and the resulting threshold is controlled by the discount function. 
In particular, the time dependence is never worse than the full-history \(\sqrt t\)-type dependence, and it becomes time-uniform for the standard geometric discount.
Overall, our results show that memory is a central part of perpetual fairness. 
The question is not only which fairness requirement to impose, but also how the system should count past unfairness.
\end{abstract}



\maketitle

\section{Introduction}
\label{sec:intro}

Many real-world systems make decisions repeatedly, under uncertainty, and over a long and unknown duration.
At each time step, new information arrives, an action must be chosen immediately, and the choice is irrevocable: once taken, it becomes part of the history and cannot be changed later.
Meanwhile, the system is expected to remain fair not only at some final time, but throughout the process.

This combination is challenging for two reasons.
First, the decision maker cannot tune the rule to a known horizon, change the past, or wait until the end to repair unfairness.
The guarantee should hold after every prefix of the process.
Second, in a long-running system, fairness is also a question of memory: should every past decision count forever with the same weight, or should old decisions gradually become less important?
This paper studies both issues through a common deficit-based viewpoint.

\subsection{Motivating application: approximately-fair allocation of indivisible items}
\label{sub:propc}

Consider a food-bank that receives short shelf-life food donations, and must allocate them immediately among several organizations, which may have different preferences over the donations.
Because storage is limited (or forbidden by health regulations), the allocation must be made upon arrival \citep{aleksandrov2020online,aleksandrov2015online}.
The cumulative allocation at each time step should be \emph{fair}. 
As another example, consider a cloud-computing controller that repeatedly assigns CPU/GPU time, memory, or bandwidth to jobs as they arrive.
These assignments must be made at high frequency, and still satisfy some pre-specified fairness requirements.
In such systems a ``round'' can correspond to a millisecond-scale scheduling decision,
so it is natural to evaluate fairness across thousands (or millions) of rounds \citep{verma2015borg}.

Common fairness measures are based on approximations of \emph{proportionality} and \emph{envy-freeness}.
Proportionality asks that each agent receive subjective value at least \(1/n\) of the total value, where \(n\) is the number of agents.
Envy-freeness asks that each agent value her own allocation at least as much as any other agent's allocation.
Since exact proportionality or envy-freeness may be impossible even for a single indivisible item, these notions are usually relaxed.
The relaxation \(\PROPc\) requires that each agent's proportionality deficit can be covered by the value of at most \(c\) items allocated to others \citep{conitzer2017fair,segal2019democratic}.
Analogously, \(\EFc\) requires that any envy toward another agent disappears after removing at most \(c\) highly valued items from the other agent's bundle.


In the offline setting, where all items are known in advance, it is easy to attain \(\PROP1\) and \(\EF1\), for example by allocating the items in a round-robin fashion.
The online setting is much more challenging.
The following examples show that even natural online rules can yield very bad approximations to proportionality, already for two agents.

Consider first the round-robin algorithm. 
Suppose every item that comes in an odd round is worth $1$ to both agents, and every item that comes in an even round is worth $\epsilon\ll 1$ to both agents. Then, for any $c$, the allocation violates $\PROPc$ after $4c/(1-\epsilon) \in \Theta(c)$ rounds.

Second, consider a greedy algorithm, that allocates each item so that the smallest utility of an agent is maximized.
Suppose the first item is worth $1$ to both agents, and is given to agent 1.
The following items are worth $1$ to agent 1 and $\epsilon\ll 1$ to agent 2. Then the next $\lceil 1/\epsilon\rceil$ items will be allocated to agent 2.
For any $c$ such that $2c+1<1/\epsilon$, the allocation violates $\PROPc$ by round $2c+3$.

Third, consider a different greedy algorithm, that allocates each item to the agent for whom the \emph{deficit} (the proportional share minus the current utility) is largest. Suppose the item values are as in \Cref{tab:intro-greedy-counterexample}, where $\epsilon\ll 1$.
One can check that, from round $3$ onwards, every item is given to the agent who values it at $\epsilon$, and the deficit of each agent grows by $\frac{1-\epsilon}{2}$ every two rounds.
Therefore, for all $c\geq 1$, the allocation  is not $\PROPc$ after $2+4c/(1-\epsilon) \in \Theta(c)$ rounds. 
These examples show that even locally-sensible rules can accumulate persistent unfairness over time.
The main issue is that unfairness behaves like a deficit: once created, it remains in the history, and if the rule does not manage these deficits carefully, they may grow over time.

\begin{table}[ht]
\small
\centering
\caption{Two-agent instance (\(n=2\)) for the ``deficit-minimizing'' greedy rule.
Circled entries indicate the agent who receives the item in that round.}
\label{tab:intro-greedy-counterexample}
\begin{tabular}{lcccccccc}
\toprule
Round: & 1 & 2 & 3 & 4 & 5 & 6 & $\cdots$ \\
\midrule
Agent 1 item value:
& \tikz[baseline=(x.base)] \node[draw, circle, inner sep=1pt] (x) {1}; & 1 & 1 &
\tikz[baseline=(x.base)] \node[draw, circle, inner sep=1pt] (x) {$\epsilon$}; & 1 &
\tikz[baseline=(x.base)] \node[draw, circle, inner sep=1pt] (x) {$\epsilon$}; & $\cdots$ \\
\addlinespace[0.5ex]
Agent 2 item value:
& 1 & \tikz[baseline=(x.base)] \node[draw, circle, inner sep=1pt] (x) {$\epsilon$}; &
\tikz[baseline=(x.base)] \node[draw, circle, inner sep=1pt] (x) {$\epsilon$}; & 1 &
\tikz[baseline=(x.base)] \node[draw, circle, inner sep=1pt] (x) {$\epsilon$}; & 1 & $\cdots$ \\
\addlinespace[1ex]
\hdashline
\addlinespace[1ex]
Agent 1 bundle value:
& 1 & 1 & 1 & $1+\epsilon$ & $1+\epsilon$ & $1+2\epsilon$  & $\cdots$ \\
Agent 2 bundle value:
& 0 & $\epsilon$ & $2\epsilon$ & $2\epsilon$ & $3\epsilon$ & $3\epsilon$  & $\cdots$ \\
\addlinespace[1ex]
\hdashline
\addlinespace[1ex]
Agent 1 Proportional share:
& $\frac{1}{2}$ & 1 & $\frac{3}{2}$ & $\frac{3+\epsilon}{2}$ & $\frac{4+\epsilon}{2}$ & $\frac{4+2\epsilon}{2}$  & $\cdots$ \\
Agent 2 Proportional share:
& $\frac{1}{2}$ & $\frac{1 + \epsilon}{2}$ & $\frac{1 + 2\epsilon}{2}$ & $\frac{2 + 2\epsilon}{2}$ & $\frac{2 + 3\epsilon}{2}$ & $\frac{3 + 3\epsilon}{2}$  & $\cdots$\\
\addlinespace[1ex]
\hdashline
\addlinespace[1ex]
Agent 1 Deficit (\(Prop-value\)):
& -$\frac{1}{2}$ & 0 & $\frac{1}{2}$ & $\frac{1 - \epsilon}{2}$ & $\frac{2 - \epsilon}{2}$ & $\frac{2 - 2\epsilon}{2}$  & $\cdots$\\
Agent 2 Deficit (\(Prop-value\)):
& $\frac{1}{2}$ & $\frac{1 - \epsilon}{2}$ & $\frac{1 - 2\epsilon}{2}$ & $\frac{2 - 2\epsilon}{2}$ & $\frac{2 - 3\epsilon}{2}$ & $\frac{3 - 3\epsilon}{2}$  & $\cdots$\\
\bottomrule
\end{tabular}
\end{table}

The previous examples also clarify why final-time guarantees are not enough for our setting.

\citet{benade2018make} presented the first online algorithm that ``makes the average envy vanish over time.'' Under the normalization that every item value lies in \([0,1]\), their algorithm guarantees that the maximum cardinal envy is \(O(\sqrt{T})\). We refer to this bounded cardinal-envy guarantee as \(\mathrm{bEF}(c)\); it is different from classical \(\EFc\), where \(c\) counts the number of goods that may be removed from the envied bundle. Their algorithm has two shortcomings for our setting.
(1) Its guarantee is tuned to a specific horizon: the algorithm must know the total number of rounds \(T\) in advance, and fairness is guaranteed only at time \(T\). (2) It requires a known global scale, namely that every item value is at most \(1\).

Our goal is to develop an algorithm that is \emph{fully online} in that it requires no assumption about the horizon or about future inputs (see Definition \ref{def:fully-online} for details).%
\footnote{
\citet{chandak2024proportional} use a similar term in the context of online voting. They call an algorithm \emph{semi-online} if it requires information about the horizon, but not about future inputs.
}
Moreover, we aim for an algorithm that guarantees approximate fairness \emph{perpetually}, after each individual step.
In \Cref{app:benade-intermediate-linear} we consider several natural ways to extend the algorithm of \citet{benade2018make} to a fully online algorithm, and explain why they do not work.


The central idea in this paper is to track \emph{deficits}.
A deficit measures how far a particular fairness requirement is from being satisfied.
The algorithm then chooses actions that control the whole deficit vector over time.
This gives a common language for several fairness goals: proportionality-style requirements, envy-style requirements, and public-decision requirements can all be expressed as deficits that should be kept small.

\Cref{fig:intro-greedy-drift} illustrates this viewpoint: local choices may create drift, while controlling the deficit vector keeps all deficits below the threshold.

\begin{figure}[ht]
\centering
\begin{tikzpicture}[
  x=.058\linewidth,
  y=.62cm,
  >=Latex,
  font=\scriptsize,
  panel/.style={draw=black!10, fill=black!1, rounded corners=6pt, line width=.45pt},
  box/.style={draw=black!18, fill=white, rounded corners=4pt, line width=.45pt},
  greedybox/.style={draw=orange!75!black, fill=orange!10, rounded corners=3pt, line width=.75pt},
  altbox/.style={draw=green!45!black, fill=green!6, rounded corners=3pt, line width=.55pt},
  token/.style={draw=purple!60!black, fill=purple!4, circle, minimum size=.50cm, inner sep=0pt},
  arr/.style={-{Latex[length=1.8mm]}, draw=black!45, line width=.7pt},
  axis/.style={-{Latex[length=1.8mm]}, draw=purple!65!black, line width=.7pt},
  ct/.style={draw=purple!65!black, dashed, line width=.65pt},
  greedy/.style={draw=orange!75!black, line width=1.15pt},
  pot/.style={draw=green!38!black, line width=1.15pt}
]

\draw[panel] (0,0) rectangle (6.55,4.4);

\foreach \x/\h in {.75/.15,1.95/.30,3.15/.45,4.35/.60,5.55/.75}{
  \draw[box] (\x-.42,2.35) rectangle +( .84,1.18);
  \draw[greedybox] (\x-.31,3.05) rectangle +( .62,.27);
  \draw[altbox]    (\x-.31,2.58) rectangle +( .62,.27);

  \fill[orange!70!black] (\x-.18,3.185) circle (.04);
  \fill[green!45!black] (\x-.21,2.61) -- +( .07,.11) -- +( .14,0) -- cycle;

  \draw[orange!55!black, dashed, line width=.45pt] (\x,2.35) -- (\x,1.52);
  \node[token] at (\x,1.08) {};
  \draw[fill=purple!35, draw=purple!60!black, line width=.25pt] (\x-.14,.92) rectangle +( .06,\h);
  \draw[fill=purple!35, draw=purple!60!black, line width=.25pt] (\x-.02,.92) rectangle +( .06,\h+.11);
  \draw[fill=purple!35, draw=purple!60!black, line width=.25pt] (\x+.10,.92) rectangle +( .06,\h+.22);
  \node[text=purple!70!black] at (\x+.25,1.28) {$+$};
}

\node[text=orange!75!black] at (.75,3.95) {greedy};
\draw[-{Latex[length=1.5mm]}, draw=orange!70!black, line width=.55pt] (.75,3.78) -- (.75,3.34);

\foreach \a/\b in {.75/1.95,1.95/3.15,3.15/4.35,4.35/5.55}{
  \draw[-{Latex[length=1.4mm]}, draw=purple!50!black, line width=.55pt] (\a+.32,1.08) -- (\b-.32,1.08);
}
\node[text=black!60] at (6.10,3.27) {$\cdots$};
\node[text=black!60] at (6.10,1.08) {$\cdots$};

\draw[arr, line width=.9pt] (6.82,2.20) -- (7.55,2.20);

\draw[panel] (7.85,0) rectangle (16.55,4.4);
\fill[green!4] (8.55,.65) rectangle (15.95,2.55);
\fill[orange!4] (8.55,2.55) rectangle (15.95,4.05);

\draw[axis] (8.55,.65) -- (16.10,.65);
\draw[axis] (8.55,.65) -- (8.55,4.15);

\draw[ct] (8.55,2.55) -- (15.95,2.55);
\node[text=purple!70!black, anchor=south west] at (8.70,2.57) {$c_t$};

\draw[greedy]
  (8.55,.72) -- (9.05,1.02) -- (9.55,1.18) -- (10.05,1.48)
  -- (10.55,1.68) -- (11.05,1.93) -- (11.55,2.15)
  -- (12.05,2.42) -- (12.55,2.70) -- (13.05,2.95)
  -- (13.55,3.18) -- (14.05,3.47) -- (14.55,3.72) -- (15.10,3.95);
\node[text=orange!75!black, anchor=west] at (14.05,3.85) {greedy};

\draw[pot]
  (8.55,.72) -- (9.05,.96) -- (9.55,.82) -- (10.05,1.12)
  -- (10.55,.98) -- (11.05,1.28) -- (11.55,1.10)
  -- (12.05,1.40) -- (12.55,1.20) -- (13.05,1.53)
  -- (13.55,1.32) -- (14.05,1.62) -- (14.55,1.40) -- (15.10,1.72);
\node[text=green!38!black, anchor=west] at (13.20,1.25) {controlled};

\draw[draw=green!38!black, fill=green!7, line width=.5pt]
  (15.55,1.72) -- (15.78,1.82) -- (16.01,1.72)
  -- (15.97,1.42) -- (15.78,1.24) -- (15.59,1.42) -- cycle;
\draw[green!38!black, line width=.6pt] (15.67,1.55) -- (15.75,1.46) -- (15.90,1.67);

\draw[draw=orange!75!black, fill=orange!8, line width=.5pt]
  (15.47,3.95) -- (15.70,4.05) -- (15.93,3.95)
  -- (15.89,3.65) -- (15.70,3.47) -- (15.51,3.65) -- cycle;
\node[text=orange!75!black] at (15.70,3.79) {$!$};

\end{tikzpicture}
\caption{Local choices can create persistent drift; controlling the deficit vector keeps all deficits below the threshold.}
\Description{A schematic figure. Repeated local choices add small positive deficits, producing drift that crosses a threshold, while deficit control keeps all deficits below the threshold.}
\label{fig:intro-greedy-drift}
\end{figure}
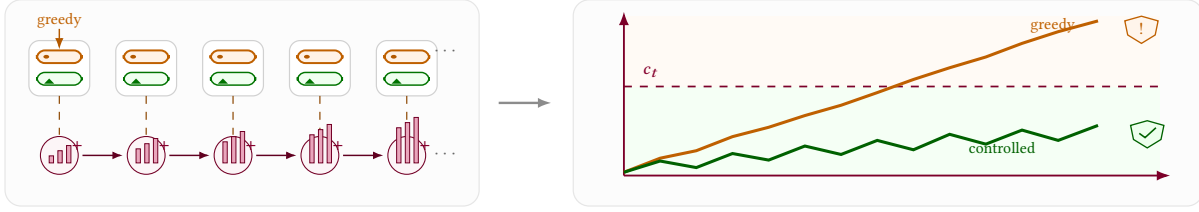

This viewpoint also makes the role of memory explicit.
In the full-history model, all past rounds count equally.
This is the strongest and most demanding interpretation of perpetual fairness, and we first study what can be guaranteed in this model.
We prove that a simple fully-online potential rule keeps all deficits small at every prefix, with a slack that grows on the order of \(\sqrt{t/n}\) up to logarithmic factors.
We also prove that this growth is unavoidable in general: if the whole history counts forever, no online algorithm can keep the fairness slack smaller in the worst case.

This impossibility result in the full-history model motivates a second question: what happens when fairness is evaluated with limited memory?
A hard window is the most direct approach, but we show that it can be misleading: every fixed window may look fair while the full-history deficit grows linearly.
We therefore also study a smoother model, which is very common in economics \citep{frederick2002time}, in which older rounds gradually receive less weight rather than disappearing abruptly.
The resulting guarantee is similar to the full-memory guarantee, except that the dependence on the time $t$ is replaced by dependence on a function \(\M_\gamma(t)\), which depends on the memory-function $\gamma$, and measures how much past history is effectively retained by the fairness test. With geometric discounting (the standard time-discount model in economics), \(\M_\gamma(t)\) is upper-bounded by a constant for all $t$, which yields a fairness guarantee that does not grow with time.

Although item allocation is our main motivating example, the framework is more general.
For instance, consider a public institution that repeatedly decides which proposals, projects, or cases to prioritize.
A supreme court may repeatedly decide which petitions to hear from a large pool of filings; only a small fraction receive full review and oral argument, and denials are effectively final.
Even so, the long-run sequence of agenda choices is expected to remain fair over time across constituencies, such as regions, litigant types, or issue areas \citep{supremeCourtAtWork}.
This is another setting with irrevocable decisions, an unknown horizon, and ongoing fairness requirements.



\subsection{Contributions}
\paragraph{A general framework: deficits (\Cref{sec:definition}).}
We consider a decision-maker who, in each round, chooses an action from a feasible set, determined by 
the current input.
The outcome is evaluated by a family of \emph{tracked requirements}.
For each tracked requirement, 
we track a \emph{deficit value}.
A deficit is simply the amount by which the corresponding fairness requirement
is currently behind. Thus different fairness requirements can be treated in the
same way: each one becomes a deficit that we try to keep small.
The goal is to ensure that all deficits  remain below a time-dependent threshold.
Formally, given a threshold $c_t>0$, call a tracked requirement \emph{$c_t$-violated} if its deficit exceeds $c_t$. An outcome is \emph{$c_t$-fair} if no tracked requirement is $c_t$-violated.
The decision-maker wants to ensure that the outcome is $c_t$-fair at every $t\geq 0$.

To instantiate our framework for approximately-proportional item allocation, we define $n$ tracked requirements, one per agent. The deficit of $i$ is the difference between $i$'s proportional share and $i$'s actual value, divided by the largest item value seen by $i$ so far. If all deficits are at most $c$, then the allocation satisfies (a scale based relaxation of) PROPc. For approximately-envy-free item allocation, we define $n(n-1)$ tracked requirements, one for each ordered pair of agents, and proceed analogously. See \Cref{sec:applications} for details.

\begin{fullversion}
\paragraph{Why perpetual guarantees are hard.}
In offline problems, or in online problems where the horizon $T$ is known, algorithms can sometimes be tuned to $T$ and analyzed only at the final time.
In perpetual systems this is not enough:
we want a single process that runs forever and maintains guarantees at \emph{every prefix}.
Moreover, greedy ``fix the worst tracked requirement right now'' instincts can behave poorly over long durations:
even if each individual step looks reasonable, small systematic imbalances can accumulate, and the slack needed to keep all tracked requirements satisfied may grow linearly with time (see \Cref{sub:propc}).
This leads to a basic question that applies across domains:
\emph{How fast must the threshold $c_{t}$ grow with time?}
\end{fullversion}


\paragraph{A generic perpetual online decision rule (\Cref{sec:algorithm}).}

We present a simple greedy rule that chooses an action at every round.
The rule is fully online; it does not need to know the horizon or any other information about the future, and it does not change past decisions. It uses only the information revealed in the current round, and requires a linear number of deficit computations.

The rule is based on a structural property that shows up repeatedly in online fairness problems.
In each round, there are typically \(n\) natural reference
actions: for example, in allocation, giving the item to each one of the \(n\)
agents; and in public decision-making, choosing each agent's favorite feasible
outcome.

The key requirement is that this menu is locally balanced. For each tracked requirement, some reference actions may increase its deficit, but other reference
actions repair it, and the total change over the whole menu is not positive. We
also require that the one-round changes are not too large.
We formalize these conditions in \Cref{sec:definition}.
After writing the domain in deficit form, these conditions are usually easy to check.
Under these conditions, we prove that at every time $t$, all deficits are at most $c_{t}$, where the threshold $c_t$ grows asymptotically as the square-root of time:
\[
c_{t} \;\in \; O\!\left(\sigma\sqrt{\frac{t\,\log m}{n}}\right),
\]
where $\sigma$ is a problem-specific constant (bounding the squared one-round changes), $m$ is the number of tracked requirements, and $n$ is the number of reference actions
(see \Cref{cor:ct-feasible-general} for the precise expression).

\paragraph{Applications (\Cref{sec:applications})}
To demonstrate the flexibility of our framework, we instantiate it for several common online fairness goals.

\paragraph{1. Approximately-proportional online item allocation (private goods; motivating domain).}
Items arrive one by one and must be allocated immediately and irrevocably to one of $n$ agents. 
Each agent $i$ has an additive valuation function denoted by $v_i$. Denote by $G^t$ the set of items that have arrived up to and including time $t$, and by $P_i^t$ the subset of these items allocated to $i$. The proportional share of agent $i$ is defined as $v_i(G^t)/n$. 
A common fairness goal is
\emph{$\PROPc$ (``proportional up to $c$ items'')}, which means: at any time, an agent may be below her proportional share, but the deficit can be covered by the value of at most $c$ items. Formally, $v_i(G^t)/n - v_i(P_i^t)$ is at most the highest value of any subset of $c$ items in $G^t\setminus P_i^t$.

We obtain the following scale-based relaxation, which we call 
\emph{$\PROPax{c}$}. It means that $v_i(G^t)/n - v_i(P_i^t)$ is at most $c$ times the highest value of any item in $G^t\setminus P_i^t$. Using our framework, we present a fully online algorithm that keeps the allocation $\PROPax{c_{t}}$ at any time $t$, where $c_{t}\in O\left(\sqrt{\frac{t\log n}{n}}\right)$.
The runtime of the algorithm is $O(n)$ per round (see \Cref{subsec:ex-propxc} for the precise statement).

\paragraph{2. Approximately-proportional public decision-making.}
In this setting, in each round $t=1,2,\ldots$, a decision maker chooses one feasible outcome from a finite nonempty set $C$
(e.g., which project to fund or which policy to implement).
Each agent $i$ has a (possibly time-varying) valuation function $v_i^{t}:C\to\mathbb{R}_{\ge 0}$, and utilities add across rounds.
That is, if $o^{t}\in C$ is the chosen outcome in round $t$, then agent $i$'s cumulative utility is
$u_i^t:=\sum_{r=1}^t v_i^{r}(o_r)$.
Following \citet{conitzer2017fair}, we measure proportionality against what an agent could have obtained by always selecting her favorite outcome in each round.
Let $M_i^{r}:=\max_{o\in C} v_i^{r}(o)$ and define the proportional share
\(
\Prop_i^t \ :=\ \frac{1}{n}\sum_{r=1}^t M_i^{r}.
\)
The notion $\mathrm{PROP}1$ in this setting requires
$u_i^t \ge \Prop_i^t-\max_{r\le t} M_i^{r}$, i.e., an agent's deficit can be covered by the value of her single highest-stakes round so far.
We extend this notion to \emph{$\PROPax{c}$}.
Let $V_i^t:=\max_{r\le t} M_i^{r}$ be agent $i$'s largest attainable single-round value so far, and require
\(
\Prop_i^t-u_i^t \ \le\ c\,V_i^t
\)
for all $i\in[n]$.
Using our framework, we present a fully online algorithm that keeps the outcome $\PROPax{c_{t}}$ at any time $t$,
where $c_{t}\in O\!\left(\sqrt{\frac{t\log n}{n}}\right)$. The
runtime of the algorithm is $O(n|C|)$ per round (see \Cref{subsec:ex-pdm}).

\paragraph{3. Approximately envy-free online item allocation.}
In the same private goods allocation setting, another common fairness goal is 
\emph{$\EFc$ (``envy-free up to $c$ items'')}, which means
that, for any two agents $i,j$, the envy $v_i(P_j^t) - v_i(P_i^t)$ is at most the highest value for $i$ of any subset of $c$ items in $P_j^t$.
Analogously, we define the scale-based relaxation \emph{$\EFax{c}$}, which means that agent $i$'s envy toward agent $j$ is at most $c$ times the highest value of an item given to $j$.
Using our framework, we present a fully online algorithm that keeps the allocation $\EFax{c_{t}}$ at any time $t$, where
$c_{t} \in O\!\left(\sqrt{\frac{t\log n}{n}}\right)$.
The runtime is $O(n^2)$ per round
(see \Cref{subsec:ex-efc}).
Thus, there is a trade-off between a stronger fairness condition ($\EFax{c}$) and a faster per round decision algorithm (for $\PROPax{c}$).

Note that the algorithm of \citet{benade2018make} guarantees that the maximum envy is at most $c$, but only under the assumption that all valuations are bounded in $[0,1]$. We call this fairness notion $\mathrm{bEF}(c)$ (bounded EF($c$)). Similarly, one could define $\mathrm{bPROP}(c)$ as a guarantee that the disproportionality is at most $c$, assuming all valuations are bounded in $[0,1]$. Our $\EFax{c}$ notion is stronger than $\mathrm{bEF}(c)$, as the envy-bound between $i$ and $j$ is based on the highest-valued item in $j$'s bundle, rather than the highest-valued item overall. Similarly, $\PROPax{c}$ is stronger than \(\mathrm{bPROP}(c)\). Moreover, our algorithm does not need to know the largest item value in advance.

\paragraph{4. Classical $\EFc$.}
Without restricting the possible item values in advance, classical \(\EFc\) admits no
sublinear online guarantee against an adaptive adversary. We prove that, for
every horizon \(T\), every online allocation algorithm can be forced to require
\[
c\ge \left\lceil\frac{T}{n}\right\rceil,
\]
even when all item values lie in \((0,1]\). This bound is tight, since round
robin guarantees \(\EFparam{\lceil t/n\rceil}\) after every prefix; see
\Cref{thm:classical-efc-linear-lb}.

As a complementary positive result, suppose that the finite set of possible
item values is known in advance and has size \(L\). Our threshold-based
instantiation is perpetual and guarantees
\[
c_t\in
O\!\left(
\ln(nL)+\sqrt{\frac{t\ln(nL)}{n}}
\right).
\]
Its running time is \(O(n^2L)\) per round. 

\paragraph{A negative result (\Cref{sec:lower_bound})}
To complement our positive results, we present a negative result for 
$\mathrm{bPROP}(c)$, which is weaker than $\PROPax{c}$, $\EFax{c}$, and the condition of \citet{benade2018make} (which we called $\mathrm{bEF}(c)$).
Specifically, for any $n\geq 2$, we present an instance with $n$ agents, where all agents' valuations are bounded in $[0,1]$, 
such that no online allocation algorithm can guarantee maximum deviation from proportionality below
\[
\Omega\!\left(\sqrt{\frac{t}{n}}\right)
\]
at all prefixes up to time \(t\), for all \(t\ge n\) (\Cref{thm:lb-framework}).
This implies a lower bound for \(\PROPax{c}\) and \(\EFax{c}\) (as well as
\(\mathrm{bEF}(c)\)) matching the \(t\)-dependence of our upper bounds and
their \(n\)-dependence up to logarithmic factors.

Thus, the square-root growth in our positive guarantees should not be viewed merely as an artifact of the potential analysis.
Rather, it reflects an inherent limitation of the unrestricted online adversarial model: even substantially weaker proportionality requirements cannot, in general, be maintained with a time-independent fairness threshold.
Note that our negative result is valid also when an upper bound on the item values is given in advance.

\paragraph{Limited-memory fairness (\Cref{sec:limited-memory}).}
Our impossibility result shows that full-history fairness cannot, in general, be
maintained with a time-independent threshold. This motivates fairness notions
that give less weight to old rounds. We first show that hard-window fairness is
not enough: every fixed window may look fair while the full-history
proportionality deficit grows linearly.

We then introduce discounted-memory fairness, where old rounds fade gradually according to a chosen memory function \(\gamma\). Under the corresponding discounted local balance conditions, the same potential rule gives a guarantee that depends on $t$ only through an expression \(\M_\gamma(t)\), induced by this choice of \(\gamma\). Thus different choices of \(\gamma\) interpolate between the full-history \(\sqrt t\)-type dependence in the worst case (\(\gamma\equiv 1\)), 
and time-independent guarantees when \(\gamma(t)\equiv \lambda\) for some $\lambda\in (0,1)$.

Thus, memory is not merely an implementation detail. It changes how fairness is
evaluated over time, and therefore changes what kind of perpetual guarantee is
possible.

\paragraph{Additional appendix result.}
Most technical proofs are delegated to the Appendix.
Additionally, in \Cref{app:optimal-policy}, we present an algorithm for computing the exact number of rounds by which \(\mathrm{bPROP}(c)\) can be maintained, as well as an optimal allocation policy that obtains this bound against an adaptive adversary.
Although the algorithm runs in exponential time, we could run it to completion to obtain an exact result for \(c=1\) and \(n=2\).

\begin{fullversion}
\paragraph{Our techniques (informal).}
We analyze a smooth 
``soft-max'' style potential that aggregates all deficits while remaining stable under small one-step changes. We use the local balance conditions to bound its growth over all prefixes.


Overall, the paper gives a way to (i) express diverse online fairness goals through deficits,
(ii) obtain perpetual guarantees from a single fully-online greedy rule, and
(iii) understand the fundamental limitations of what any online algorithm can achieve under adversarial inputs.
\end{fullversion}

\section{Related Work}
\label{sec:related}
We organize the related work around the main themes of the paper: online fair allocation, public decision-making, memory over time, and online balancing and related welfare-oriented approaches. We close the section with a short summary explaining how these lines motivate our modeling choices and position our results

\subsection{Online fair item allocation and full-history guarantees}
Most works on fair division of indivisible items study the offline setting; see \cite{AMANATIDIS2023103965} for a recent survey.
This literature already shows that, for indivisible goods, exact proportionality is usually too strong, so one must choose a weaker but meaningful fairness notion. For example, \citet{farhadi2019fair} study weighted maximin-share guarantees for agents with unequal entitlements, while \citet{segalHalevi2020fairDD} study proportionality under ordinal preference information with diminishing differences.
Our work asks a different question: in an online process, can we keep the loss from proportionality small after every prefix, even though the allocation decisions are irrevocable?
In recent years, however, there has been a growing interest in online or dynamic fair division; see the survey by \citet{aleksandrov2020online}.

Specifically, \citet{aleksandrov2015online} study online allocation of indivisible items in a restricted setting in which agents report $0/1$-type values (``like'' / ``dislike'') for each arriving item, motivated by applications such as food banks.
They analyze simple online mechanisms and their fairness properties over time.
These works are domain-specific, but they highlight the same basic difficulty that motivates our paper: decisions are irrevocable, inputs arrive over a long and unknown duration, and fairness must be monitored repeatedly rather than only at a single terminal time.

\citet{benade2018make} and also \citet{benade2023online} study online fair division with general additive valuations under a global normalization assumption (item values in $[0,1]$).
They present deterministic potential-based algorithms that guarantee sublinear terminal envy, in particular $\tilde O(\sqrt{T/n})$-type bounds at a horizon $T$.
However, their analysis is primarily \emph{final-time}: guarantees are proved for round $T$, whereas we target \emph{anytime} guarantees that hold after every prefix.
For completeness, \Cref{app:benade-intermediate-linear} gives a simple two-agent instance in which their deterministic rule can have \emph{linear} envy at time $t=\lfloor\sqrt{T}\rfloor$, so the final-time guarantee does not directly imply an $O(\sqrt{t})$-type anytime bound.%
\footnote{
	Their deterministic algorithm needs to know the number of rounds $T$ in advance, and is therefore not fully-online in the terminology used above.
	As usual, one can remove explicit horizon dependence by restarting the algorithm on a geometrically increasing sequence of horizons ($T=1,2,4,8,\ldots$), at the cost of constant-factor losses.
	However, this still does not directly give prefix-wise guarantees at \emph{all} intermediate times.
	Our rule is fully online by construction: it does not need to know or change $T$, and it maintains its bounds at every time step.}
Moreover, our solution does not need to know the maximum item value in advance, and it is explicitly scale-aware: the normalizing scale to a deficit may evolve over time (e.g., via running maxima), which is essential in perpetual settings.
On the lower-bound side, \citet{benade2018make} prove impossibility results for online envy minimization, showing that envy must grow polynomially with the horizon in worst-case instances. Our lower bound for full-history fairness (Theorem~\ref{thm:lb-framework}) implies an asymptotically stronger lower bound. Moreover, our bound holds even for $\PROPax{c}$, which is weaker than $\PROPc$, and much weaker than $\EFc$ for three or more agents. This supports the view that growing slack is the inevitable price of taking the full history into account.

Recently, \citet{yang2024greedy} analyze greedy-based online allocation with indivisible items under adversarial valuations. They show that an algorithm called PACE, originally developed for online market equilibrium, coincides with an integral greedy rule and, under mild bounded-ratio assumptions on agents' nonzero values, prove horizon-independent (as $T\to\infty$) bounds on multiplicative envy and Nash-welfare competitiveness (a ``best-of-many-worlds'' guarantee for PACE across stochastic, nonstationary, and restricted adversarial inputs). Their guarantees are asymptotic/final-horizon, whereas we seek prefix-wise guarantees that hold at every time step.


\citet{neoh2025additional} and \citet{choo2025approximate} study how additional information or predictions can improve online item-allocation fairness.
\citet{neoh2025additional} consider several forms of information about future goods, while \citet{choo2025approximate} focus on approximate proportionality and show, among other things, that natural greedy rules can fail for \(PROP1\)-type guarantees.
Relatedly, \citet{melissourgos2025onlineefx} study online \(\mathrm{EFX}\) allocations with predictions.
While predictions are very useful, in some cases they are not available; we develop algorithms that attain the best fairness possible when predictions are not available. Additionally, we target \emph{perpetual (Anytime)} additive guarantees rather than final-time guarantees.

Several recent works obtain stronger guarantees by imposing additional structure on the valuation domain or by studying related online models.
Particularly relevant to our classical \(\EFc\) instantiation (application 4), \citet{amanatidis2025personalized} study personalized two-value instances, while \citet{wang2026binary} study binary valuations and related restricted domains.
These works are complementary to ours: they exploit structured valuation classes to obtain stronger guarantees. Our positive result for classical \(\EFc\) instead assumes a finite value set known in advance, while \Cref{thm:classical-efc-linear-lb} shows that, without any restriction specified in advance on the possible item values, the exact worst-case parameter at horizon \(T\) is \(\lceil T/n\rceil\).

Other recent papers address online fair division under different fairness notions or arrival models.
\citet{kulkarni2025expost} study ex-post MMS guarantees when agents arrive online.
\citet{zhou2023mms} study online MMS allocation for indivisible goods and chores, and \citet{song2025chores} focus specifically on online MMS allocation for chores.
\citet{kulkarni2025subsidy} study online envy-freeable allocations with subsidies, asking when envy-freeness can be maintained online and how much subsidy is needed.
These works are not directly comparable to our prefix-wise \(\PROPax{c}\) and \(\EFax{c}\) guarantees, but they further illustrate the recent interest in fair allocation under irrevocable online decisions.

Several works consider stochastic rather than adversarial input models.
\citet{jiang2019online} study online fair division for two agents under stochastic arrivals.
\citet{zeng2020fairness} also consider stochastic adversaries and give high-probability fairness/efficiency guarantees.
\citet{gao2021online} study online market equilibrium under stochastic valuations.
Our work addresses the complementary worst-case setting of deterministic guarantees against an adaptive adversary.

A related direction allows \emph{recourse}, meaning retroactive modifications of past allocations, to improve fairness ex post; see \citet{he2019achieving}.
In contrast, our model is irrevocable: once an action is taken it is never changed later. Discounting changes how past unfairness is evaluated, but it does not allow the algorithm to correct past actions.

\subsection{Online fair public decision-making}
Another related model is \emph{public decision-making}.
In this model, each round selects a single outcome, and the chosen outcome may give positive utility to many agents (so it generalizes private-good allocation).
\citet{freeman2017fair} study the online version of public decision-making and analyze greedy mechanisms with welfare objectives.
\citet{kahana2023leximin} study fairness notions and impossibility results for online public decision-making.
These settings fit naturally into our framework viewpoint, and we include public decision-making as one of our concrete instantiations.

Fairness in public decisions is closely related to proportional representation in voting, where the goal is to represent voters fairly; see, e.g., \citet{betzler2013fullyProportional}.
Long-run collective decision models study a similar concern over a sequence of decisions: these include \emph{perpetual voting} and proportionality notions over an unbounded horizon; see, e.g., \citet{Lackner2020PerpetualVoting,Lackner2023PerpetualVotingProportionalDecisions}.
Public decision-making is also closely related to online allocation of public goods; for example, \citet{banerjee2023proportionally} study proportionally fair online allocation of public goods with predictions.

\subsection{Temporal, repeated, and limited fairness}
A closely related but different line of work is \emph{temporal fair division} \cite{cookson2025temporal,elkind2025temporal}.
Similar to our setting, the goal is to guarantee fairness repeatedly over time.
However, their information and decision models differ from ours: some of these models allow more future information or more planning within temporal blocks, whereas we require fully online, irrevocable decisions with worst-case inputs. Our limited-memory model changes how past rounds are counted, but does not give the algorithm future information or allow it to revise past decisions.

There is also growing work on \emph{repeated} allocation, where similar or identical resources recur and fairness is measured across repetitions; examples include repeated fair allocation of indivisible items \cite{igarashi2024repeated}, repeated matching \cite{caragiannis2023repeatedly}, and repeated house allocation \cite{micheel2024fairness}.
These frameworks differ from ours in that they usually optimize fairness over cycles or repetitions, often with structural assumptions on the repetition pattern, rather than providing worst-case, anytime guarantees for an arbitrarily long and adversarial sequence.
In contrast, our memory question concerns how the fairness requirement counts the past: we compare full-history fairness, hard-window fairness, and discounted-memory fairness under the same irrevocable online decision model.

The closest notion we found to our ``discounted past'' notion in the literature on online decision-making is the ``no dry spells'' axiom of \citet{Lackner2020PerpetualVoting}, which requires that no voter is kept unsatisfied for long periods of time.

Relatedly, \citet{chen2025banditsfairnessconstraints} study an online learning problem in which a learner repeatedly selects a set of options while maintaining fairness constraints over time; our setting is not a learning problem and instead controls fairness deficits after every prefix.

\subsection{Online vector balancing}

A closely related mathematical line studies \emph{online discrepancy} and \emph{online vector balancing}, where an algorithm makes irrevocable online assignments in order to keep many coordinates simultaneously balanced over time.
This line is closely connected to online envy minimization, especially when one seeks guarantees that hold at every prefix.
For example, \citet{bansal2020vectorbalancing} and \citet{bansal2021online} develop online discrepancy/vector-balancing techniques under stochastic/oblivious arrival models, and \citet{kulkarni2024optimaldiscrepancy} give optimal bounds for online discrepancy minimization in their setting.
More recently, \citet{halpern2025envydiscrepancy} make the relationship between online envy minimization (for $n=2$ agents) and multicolor discrepancy explicit and study equivalences and separations across adversary models.
These works typically focus on a single global balancing objective under a fixed global scale, whereas our deficit framework is designed to simultaneously track many heterogeneous tracked requirements whose natural scales can evolve over time, and to maintain guarantees prefix-wise at every time step.

\subsection{Online welfare optimization}

Our results are complementary to welfare- and equilibrium-oriented approaches to online allocation.
For example, recent work studies online Nash welfare maximization computation, sometimes with additional structural assumptions or prediction access \citep{banerjee2022online,huang2023online}.
These lines primarily aim to obtain a good competitive ratio to a single optimization objective, whereas our framework aims to  simultaneously control many tracked requirements over time, including scale-dependent relaxations such as  \(\PROPc\) and \(\EFc\).

\subsection{Comparison summary}
Table~\ref{tab:related-work} is not meant to list every paper discussed above. It uses representative examples to summarize the modeling choices that motivate the model in the next section.

\begin{table}[ht]
\caption{Comparison by modeling axis. The table gives representative examples rather than an exhaustive list of all cited papers. Its purpose is to show the modeling choices that lead to our framework.
}
\label{tab:related-work}
\Description{A compact related-work table organized by modeling axis. The table uses representative examples and is not intended to list all cited papers.}
\centering
\fontsize{6.6pt}{7.6pt}\selectfont
\setlength{\tabcolsep}{2.1pt}
\renewcommand{\arraystretch}{1.02}
\begin{tabularx}{\textwidth}{@{}
  >{\raggedright\arraybackslash}p{0.17\textwidth}
  >{\raggedright\arraybackslash}p{0.30\textwidth}
  >{\raggedright\arraybackslash}p{0.22\textwidth}
  >{\raggedright\arraybackslash}X
@{}}
\toprule
\textbf{Axis}
&
\textbf{Nearby work, examples}
&
\textbf{Usual choice}
&
\textbf{Role here}
\\
\midrule

Offline vs.\ online
&
Offline fair division
\citep{AMANATIDIS2023103965,farhadi2019fair,segalHalevi2020fairDD}
&
Full instance known
&
Fairness notions become online deficits.
\\
\addlinespace[1.5pt]

Terminal vs.\ prefix-wise
&
Online fair allocation
\citep{aleksandrov2020online,aleksandrov2015online,benade2018make,benade2023online}
&
Guarantee is checked at end.
&
Guarantee is checked after every prefix.
\\
\addlinespace[1.5pt]

Horizon knowledge
&
Horizon-based online bounds
\citep{benade2018make,benade2023online}
&
End time \(T\) known in advance.
&
No need to know when the process ends.
\\
\addlinespace[1.5pt]

Future information
&
Prediction-assisted models
\citep{neoh2025additional,choo2025approximate,melissourgos2025onlineefx,banerjee2023proportionally}
&
Predictions, extra input, normalization.
&
No future input is used.
\\
\addlinespace[1.5pt]

Input assumptions
&
Structured or stochastic models
\citep{amanatidis2025personalized,wang2026binary,jiang2019online,zeng2020fairness,gao2021online,chen2025banditsfairnessconstraints}
&
Restricted or random input
&
Adversarial inputs, controlled by local balance.
\\
\addlinespace[1.5pt]

Past decisions
&
Recourse, repair, subsidies
\citep{he2019achieving,kulkarni2025expost,kulkarni2025subsidy}
&
Can change or compensate later.
&
Actions are irrevocable.
\\
\addlinespace[1.5pt]

Fairness language
&
MMS, chores, envy, and variants
\citep{zhou2023mms,song2025chores}
&
Notion-specific.
&
General ``tracked requirements'' framework.
\\
\addlinespace[1.5pt]

Public decisions
&
Online public choice and perpetual voting
\citep{freeman2017fair,kahana2023leximin,betzler2013fullyProportional,Lackner2020PerpetualVoting,Lackner2023PerpetualVotingProportionalDecisions}
&
Domain-specific representation.
&
Public and private settings use one rule.
\\
\addlinespace[1.5pt]

Memory
&
Temporal and repeated fairness
\citep{cookson2025temporal,elkind2025temporal,igarashi2024repeated,caragiannis2023repeatedly,micheel2024fairness}
&
Temporal structure.
&
Full history, hard windows, or discounted memory.
\\
\addlinespace[1.5pt]

Scale
&
Vector balancing and discrepancy
\citep{bansal2020vectorbalancing,bansal2021online,kulkarni2024optimaldiscrepancy,halpern2025envydiscrepancy}
&
Fixed-scale imbalance.
&
Deficit scales may grow with the prefix.
\\
\addlinespace[1.5pt]

Objective
&
Welfare and equilibrium
\citep{banerjee2022online,huang2023online,yang2024greedy}
&
One aggregate objective
&
We control many deficits simultaneously.
\\

\bottomrule
\end{tabularx}
\end{table}

This comparison motivates the model below: at each round an online action is chosen, this action updates the tracked deficits, and the goal is to keep all deficits controlled after every prefix.

\section{Model}
\label{sec:definition}

This section defines the full-history deficit model used in the main results.
Limited-memory variants are introduced later in \Cref{sec:limited-memory}.

\subsection{Actions, histories, and deficits}
\label{subsec:norm-def-framework}

A decision maker takes an \emph{action} at each time step
$t = 0,1,2,\ldots$. We use the convention that at time step $t{+}1$ (for $t\ge 0$)
the decision maker observes the realized history up to time $t$, denoted $H^t$,
and then chooses an action from a nonempty feasible set
$A^{t+1}$, which may, in general, depend on the history $H^t$ (though we do not use this generality in our applications).
The chosen action is denoted $a^{t+1}\in A^{t+1}$.
The decision is \emph{irrevocable}: once $a^{t+1}$ is chosen, it becomes part of the
history and is never changed later. 
Formally, $H^{t+1}$ is obtained from $H^t$
by appending the round-$t{+}1$ input and the action $a^{t+1}$.

The objective of the decision maker is measured by a finite set of \emph{tracked requirements}, denoted by $Q$.
Throughout the paper, let $m:=|Q|$, and assume $m\ge 2$.
The \emph{deficit} of a tracked requirement $q\in Q$ after time step $t$ (equivalently,
at the beginning of time step $t{+}1$) is denoted by $z_q^t \in\mathbb{R}_{\ge 0}$.
We assume that all deficits are initially zero:
\[
z_q^0:=0
\qquad\text{for every }q\in Q.
\]
In the full-history model, these deficits may depend on the entire realized history \(H^t\).

At time step $t{+}1$ the decision maker receives an input that allows him to compute the \emph{hypothetical post-round} deficits $z_q^{t+1}(a)$ for every potential action $a\in A^{t+1}$.
Thus, we can assume that the \emph{input in time step $t{+}1$} is presented as a set $A^{t+1}$ of possible actions, and a family of functions
\[
z_{q}^{t+1}(\cdot): A^{t+1} \to \mathbb{R}_{\ge 0}
\qquad\text{for all } q\in Q.
\]
where $z_q^{t+1}(a)$ is the deficit of tracked requirement $q$ that would result from choosing
action $a$ at time step $t{+}1$. We call these functions the \emph{deficit functions}.
In other words, before choosing the action, the rule can evaluate what the next
deficit vector would be for each feasible action. This is the information used by
the online rule below.

\begin{definition}[Fully online rules]
\label{def:fully-online}
The rule \(\mathcal R\) is \emph{fully online} if, for every \(t\ge0\) and
every two sequences of inputs, possibly of different lengths, each containing
at least \(t+1\) rounds, if the sequences are identical up to and including
time \(t+1\), then the action taken by \(\mathcal R\) at time \(t+1\) is the
same for both sequences. For a randomized rule, ``the action is the same'' is
replaced by ``the conditional distribution over actions, given the realized
history and current input, is the same.''
\end{definition}
That is, a fully online rule may use the current and past inputs, but not the
horizon or any future input.

After $a^{t+1}$ is chosen, the realized post-round deficits are
\[
z_{q}^{t+1}:=z_{q}^{t+1}(a^{t+1})
\qquad\text{for all } q\in Q.
\]
The objective of the decision maker is to keep all deficits small. Formally:
\begin{definition}
For $c\ge 0$ and time $t$, define the set of \emph{$c$-violated requirements} as the set of tracked requirements with deficit larger than $c$:
\begin{equation}
\label{eq:Dtc-general}
\mathcal{D}^{t}(c)\ :=\ \{\ {q}\in Q\ :\ z_{q}^{t}>c\ \}.
\end{equation}
We say the outcome at time $t$ is \emph{$c$-fair} if no tracked requirement is $c$-violated, that is, $\mathcal{D}^{t}(c)=\emptyset$.
\end{definition}

\subsection{Local balance conditions}
\label{subsec:moments}
    Our analysis relies on a local balance condition. This condition is the only model-specific property that must be checked in each application: once it holds, the same potential-rule analysis applies unchanged. Informally, in each round there should be enough reference actions to repair, on average, the possible increase in each deficit.

For any real $x$, denote $[x]_+ := \max\{x,0\}$ and $[x]_- := \max\{-x,0\} = [-x]_+$.

\begin{definition}[Local balance conditions]
\label{def:first-moment}
\label{def:second-moment}
We say that an instance satisfies the \emph{local balance conditions with parameters $n$ and $\sigma^2 \ge 1$} if for every round $t+1\ge 1$
there exist $n$ (not necessarily distinct) \emph{reference actions}
\(
\hat a^{t+1}_{1},\ldots,\hat a^{t+1}_{n}\ \in\ A^{t+1}
\)
such that for every tracked requirement ${q}\in Q$ there exist numbers
\begin{equation}
\label{eq:first-moment-shift}
\Delta_q^{t+1}(1),\ldots,\Delta_q^{t+1}(n) \in \mathbb{R} \quad \text{ such that} \quad z_{q}^{t+1}(\hat a^{t+1}_{k})\ \le\ \bigl[z_{q}^{t}+\Delta_{q}^{t+1}(k)\bigr]_+ \text{, for all }k\in[n].
\end{equation}
The shifts must satisfy the following two local balance inequalities:
\par\noindent
\begin{minipage}[t]{0.4\linewidth}
{\setlength{\abovedisplayskip}{0pt}\setlength{\belowdisplayskip}{0pt}%
 \setlength{\abovedisplayshortskip}{0pt}\setlength{\belowdisplayshortskip}{0pt}%
\begin{equation}
\label{eq:first-moment-mean}
\sum_{k=1}^n \Delta_{q}^{t+1}(k)\ \le\ 0
\quad\text{\emph{(mean bound)}}
\end{equation}}
\end{minipage}\hfill
\begin{minipage}[t]{0.5\linewidth}
{\setlength{\abovedisplayskip}{0pt}\setlength{\belowdisplayskip}{0pt}%
 \setlength{\abovedisplayshortskip}{0pt}\setlength{\belowdisplayshortskip}{0pt}%
\begin{equation}
\label{eq:second-moment-sigma}
\sum_{k=1}^n \bigl(\Delta_{q}^{t+1}(k)\bigr)^2\ \le\ \sigma^2
\quad\text{\emph{(squared-size bound)}}
\end{equation}}
\end{minipage}

\end{definition}

\paragraph{Interpretation: local balance.}
The conditions have a simple meaning. Each tracked requirement \(q\in Q\) has a
fairness deficit. For example, it can be the proportionality deficit of one agent, or
the envy deficit of one agent toward another. In each round \(t+1\), we look at
\(n\) reference actions
\(\hat a^{t+1}_{1},\ldots,\hat a^{t+1}_{n}\). These are natural actions that
serve the current request in \(n\) different ways. For example, in fair division
they can be the actions of giving the current item to each one of the \(n\)
agents. In public decision-making, they can be the actions of choosing each
agent's favorite feasible outcome. For a fixed tracked requirement \(q\), the number
\(\Delta_q^{t+1}(k)\) is a raw one-step shift for tracked requirement \(q\) under the
\(k\)-th reference action; the displayed next deficit is clipped at zero through \eqref{eq:first-moment-shift}.
~
The mean bound condition,
\[
    \sum_{k=1}^n \Delta_q^{t+1}(k)\le 0,
\]
says that the current round does not force the deficit of \(q\) to increase. Some
reference actions may increase this deficit, but other reference actions repair it.
Overall, when we sum over the whole menu of reference actions, the total change
is not positive. Thus the round contains enough ways to compensate for the harm
done by other ways.
~
The squared-size bound condition,
\[
    \sum_{k=1}^n\left(\Delta_q^{t+1}(k)\right)^2\le \sigma^2,
\]
says that the changes in one round are not too large after normalization. This
is important because otherwise a round could have one very large harmful action
and one very large repairing action, and still look balanced in the first
condition. The second condition rules out this kind of large jump.

Together, the two conditions say that every round has enough small repair
options. No tracked deficit is forced to go up just because of the structure of the
round, and no single round can create a very large normalized increase in the
deficit.

The reference actions are used only to prove this balance.
The algorithm itself may choose any feasible action in \(A^{t+1}\), and will choose the one that gives the best next deficit vector.

\section{Generic Online Decision Rule}
\label{sec:algorithm}
In this section, we present a generic online decision rule for the framework of \Cref{sec:definition}.  
The same rule will later be applied to discounted deficits in \Cref{sec:limited-memory}.


\paragraph{Roadmap.}
We aggregate the deficits $(z_q^t)_{q\in Q}$ to a smooth potential function denoted by $\Psi^t$.
In each round $t{+}1$, the algorithm chooses an action that minimizes the next potential $\Psi^{t+1}$. \Cref{fig:deficit-potential-rule} illustrates this one-step choice. Each feasible action induces a candidate next deficit vector, the potential assigns one score to the whole vector, and the lowest-score vector is carried forward.
The analysis compares this minimum to the average next potential of the reference actions. Using the local balance conditions and a Taylor bound, we control the one-step increase in the potential. We then telescope this inequality to bound \(\Psi^t\) for all \(t\), and translate the potential bound into an upper bound on \(|\mathcal{D}^{t}(c)|\), the number of \(c\)-violated tracked requirements at time \(t\).

As a corollary, we get the main result in this section (Corollary~\ref{cor:best-ct}). 
It shows that, for $c_{t} \in \tilde{O}(\sqrt{t/n})$, $|\mathcal{D}^{t}(c_{t})|=0$.
This implies a perpetual fairness guarantee: the outcome at any time $t$ is $c_{t}$-fair.

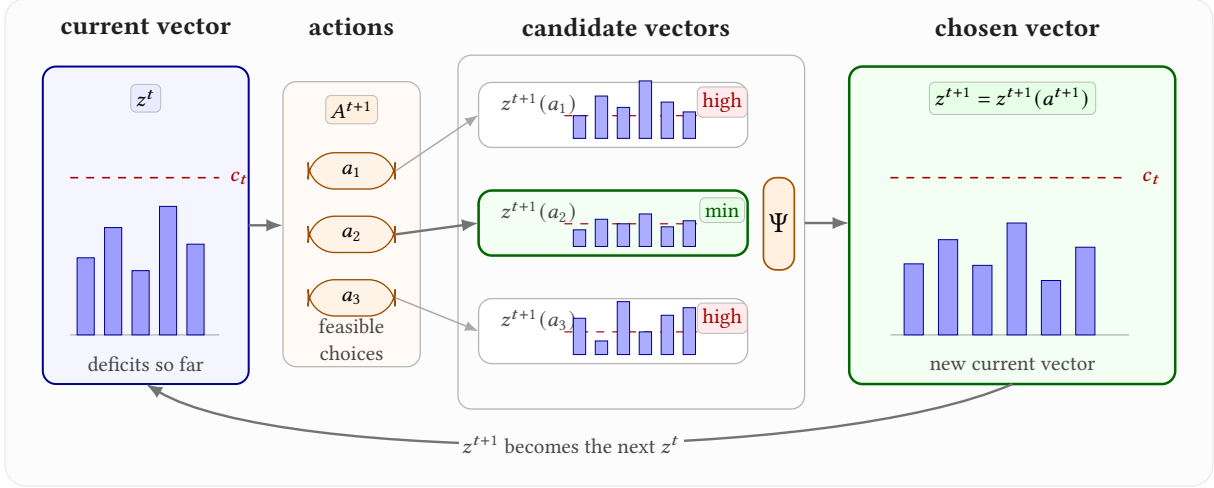
\begin{figure}[tbp]
\centering
\begin{tikzpicture}[
  x=.055\linewidth,
  y=1cm,
  >=Latex,
  font=\small,
  bg/.style={draw=black!10, fill=black!1, rounded corners=8pt},
  panel/.style={draw=black!28, fill=white, rounded corners=5pt, line width=.5pt},
  selected/.style={draw=green!40!black, fill=green!5, rounded corners=5pt, line width=.95pt},
  state/.style={draw=blue!55!black, fill=blue!4, rounded corners=5pt, line width=.75pt},
  action/.style={draw=orange!60!black, fill=orange!10, rounded corners=10pt, line width=.5pt, minimum width=1.15cm, minimum height=.48cm, font=\small},
  badge/.style={draw=black!25, fill=white, rounded corners=2pt, inner sep=2pt, font=\small},
  arrow/.style={-{Latex[length=2.2mm]}, line width=.9pt, draw=black!55},
  faint/.style={-{Latex[length=1.8mm]}, line width=.55pt, draw=black!34},
  tiny/.style={font=\small, text=black!72},
  head/.style={font=\bfseries\large, text=black!88}
]

\draw[bg] (0,0) rectangle (17.6,6.45);

\node[head] at (2.05,6.08) {current vector};
\node[head] at (5.05,6.08) {actions};
\node[head] at (9.05,6.08) {candidate vectors};
\node[head] at (14.75,6.08) {chosen vector};

\draw[state] (0.55,1.35) rectangle (3.55,5.55);
\node[badge, fill=blue!8, font=\small\bfseries] at (2.05,5.13) {$z^t$};

\draw[draw=black!35] (0.95,2.0) -- (3.12,2.0);
\draw[dashed, draw=red!65!black, line width=.6pt] (0.95,4.08) -- (3.12,4.08);
\node[font=\small, text=red!65!black, anchor=west] at (3.16,4.08) {$c_t$};

\foreach \x/\h in {1.05/1.02,1.45/1.42,1.85/.85,2.25/1.70,2.65/1.20}
  \draw[fill=blue!38, draw=blue!55!black, line width=.32pt] (\x,2.0) rectangle +(0.25,\h);

\node[tiny] at (2.05,1.62) {deficits so far};

\draw[arrow] (3.55,3.45) -- (4.05,3.45);

\draw[panel, fill=orange!4] (4.05,1.58) rectangle (6.05,5.35);
\node[badge, fill=orange!12, font=\small\bfseries] at (5.05,4.98) {$A^{t+1}$};

\node[action] (a1) at (5.05,4.17) {$a_1$};
\node[action] (a2) at (5.05,3.33) {$a_2$};
\node[action] (a3) at (5.05,2.49) {$a_3$};

\node[tiny, align=center] at (5.05,1.94) {feasible\\choices};

\draw[panel, fill=black!1] (6.60,1.02) rectangle (11.65,5.70);

\draw[panel] (6.90,4.48) rectangle (10.82,5.34);
\node[tiny, anchor=west] at (7.10,5.06) {$z^{t+1}(a_1)$};
\draw[dashed, draw=red!65!black, line width=.45pt] (8.15,4.90) -- (10.08,4.90);
\foreach \x/\h in {8.28/.30,8.60/.56,8.92/.41,9.24/.76,9.56/.48,9.88/.35}
  \draw[fill=blue!32, draw=blue!55!black, line width=.24pt] (\x,4.60) rectangle +(0.18,\h);
\node[badge, fill=red!8, text=red!60!black, font=\small] at (10.45,5.08) {high};

\draw[selected] (6.90,3.05) rectangle (10.82,3.91);
\node[tiny, anchor=west] at (7.10,3.63) {$z^{t+1}(a_2)$};
\draw[dashed, draw=red!65!black, line width=.45pt] (8.15,3.47) -- (10.08,3.47);
\foreach \x/\h in {8.28/.22,8.60/.36,8.92/.30,9.24/.43,9.56/.26,9.88/.34}
  \draw[fill=blue!32, draw=blue!55!black, line width=.24pt] (\x,3.17) rectangle +(0.18,\h);
\node[badge, fill=green!10, text=green!35!black, font=\small] at (10.45,3.65) {min};

\draw[panel] (6.90,1.62) rectangle (10.82,2.48);
\node[tiny, anchor=west] at (7.10,2.20) {$z^{t+1}(a_3)$};
\draw[dashed, draw=red!65!black, line width=.45pt] (8.15,2.04) -- (10.08,2.04);
\foreach \x/\h in {8.28/.48,8.60/.18,8.92/.70,9.24/.30,9.56/.52,9.88/.62}
  \draw[fill=blue!32, draw=blue!55!black, line width=.24pt] (\x,1.74) rectangle +(0.18,\h);
\node[badge, fill=red!8, text=red!60!black, font=\small] at (10.45,2.22) {high};

\draw[faint] (a1.east) -- (6.90,4.91);
\draw[arrow] (a2.east) -- (6.90,3.48);
\draw[faint] (a3.east) -- (6.90,2.05);

\draw[draw=orange!65!black, fill=orange!12, rounded corners=5pt, line width=.65pt]
  (11.05,2.86) rectangle (11.50,4.08);
\node[font=\Large] at (11.275,3.48) {$\Psi$};

\draw[arrow] (11.65,3.48) -- (12.30,3.48);

\draw[selected] (12.30,1.35) rectangle (17.05,5.55);
\node[badge, fill=green!10, font=\small\bfseries] at (14.68,5.13) {$z^{t+1}=z^{t+1}(a^{t+1})$};

\draw[draw=black!35] (12.90,2.0) -- (16.38,2.0);
\draw[dashed, draw=red!65!black, line width=.6pt] (12.90,4.08) -- (16.38,4.08);
\node[font=\small, anchor=west, text=red!65!black] at (16.45,4.08) {$c_t$};

\foreach \x/\h in {13.10/.94,13.60/1.26,14.10/.92,14.60/1.48,15.10/.72,15.60/1.16}
  \draw[fill=blue!38, draw=blue!55!black, line width=.32pt] (\x,2.0) rectangle +(0.28,\h);

\node[tiny] at (14.68,1.62) {new current vector};

\draw[arrow] (14.68,1.35) .. controls (12.0,.30) and (4.0,.30) .. (2.05,1.35);
\node[tiny, fill=black!1, inner sep=1.3pt] at (8.25,.55) {$z^{t+1}$ becomes the next $z^t$};

\end{tikzpicture}
\caption{One-step view of the \(p\)-potential rule. Starting from the current deficit vector \(z^t\), each feasible action \(a\in A^{t+1}\) creates a candidate next vector \(z^{t+1}(a)\). The potential \(\Psi\) scores whole candidate vectors, the rule chooses the lowest-score action, and the realized vector \(z^{t+1}\) becomes the state for the next round.}
\Description{A flow diagram showing the current deficit vector, several candidate next deficit vectors produced by feasible actions, a potential score for each candidate, and the chosen lowest-potential candidate becoming the next state.}
\label{fig:deficit-potential-rule}
\end{figure}
The next subsection states the rule formally. The analysis then follows the roadmap above: one-coordinate control, one-step potential growth, telescoping over time, and finally a tail bound on violated requirements.

\subsection{The \texorpdfstring{$p$}{p}-potential rule}
Our general decision rule is parameterized by a real number $p\ge 1$.
Given $p$, define for all $t\geq 0$,

\begin{align}
\label{eq:f-def-general}
f(u)\ &:=\ (u^2+4p^2\sigma^2)^p
\qquad (u\in\mathbb R).
&&
\Psi^t := \left(\sum_{{q}\in Q} f\!\bigl(z_{q}^{t}\bigr)\right)^{1/p}.
\end{align}


Intuitively, the potential is a single score for the whole deficit vector. A small potential
means that the deficits are collectively under control. We do not minimize only the
largest deficit, because this can ignore many other deficits that are almost
large. Instead, the potential looks at all deficits, while giving more weight to
larger ones.

The term \(4p^2\sigma^2\) is only there to make the function smooth enough for the
one-step analysis. It prevents the curvature of \(f\) from becoming too large
near zero. This is what lets us use the local balance conditions to control the
change in the potential from one round to the next.
For additional intuition on this choice of potential, see \Cref{sec:lorenz-intuition}.

For each round $t\ge 0$ and action $a\in A^{t+1}$, define the hypothetical
raw potential
\begin{align}
\label{eq:Psi-general}
\Psi^{t+1}(a)\ :=\ \left( \sum_{{q}\in Q} f\!\bigl(z_{q}^{t+1}(a)\bigr) \right)^{1/p}.
\end{align}
The $p$-potential rule chooses
\begin{equation}
\label{eq:rvm-choice-general}
a^{t+1}\ \in\ \arg\min_{a\in A^{t+1}} \Psi^{t+1}(a)
, \quad \text{breaking ties arbitrarily.}
\end{equation}
In words, the rule checks every feasible action, computes the deficit vector
that would result from this action, and chooses the action with the smallest
potential. This is a one-step rule: it uses only the current history and the current input, never changes past decisions, and does not need to know anything about future rounds.

\subsection*{Runtime complexity}
In round $t+1$, the rule minimizes $\Psi^{t+1}(a)$ over $a\in A^{t+1}$.
Evaluating $\Psi^{t+1}(a)$ requires $m$ evaluations of the deficit functions $\{z_{q}^{t+1}(a)\}_{{q}\in Q}$ for each given $a$.
Thus the naive per round time is $O(|A^{t+1}|\cdot m)$ evaluations, with $O(m)$ required storage.
Many instantiations admit faster minimization by exploiting structure in $A^{t+1}$; see the examples in \Cref{sec:applications}.


\subsection{Analysis under local balance}
\label{subsec:analysis-general}

We now prove that, under the local balance conditions of \Cref{subsec:moments}, 
the number of $c$-violated requirements $|\mathcal{D}^{t}(c)|$ is upper-bounded by an expression that depends on $c$ and $t$. From this upper bound, we deduce an expression $c_{t}$ such that the number of $c_{t}$-violated requirements must be $0$, that is, the outcome is always $c_{t}$-fair.

The proof follows the intuition behind the rule. Since the rule chooses the action with minimum next potential, its next potential is at most the average next potential over
the reference actions. The local balance conditions show that this average
cannot increase too much in one round. Repeating this argument over all rounds
gives a bound on the potential, and a bound on the potential gives a bound on
the deficits. \Cref{fig:deficit-trajectory-intuition} illustrates the resulting trajectory picture.

\begin{figure}[ht]
\centering
\begin{tikzpicture}[
  >=Latex,
  font=\small,
  axis/.style={->, draw=black!70, line width=.85pt},
  ct/.style={draw=blue!65!black, line width=1.35pt, dashed},
  run/.style={draw=green!40!black, line width=1.45pt},
  lin/.style={draw=red!55!black, line width=.9pt, dotted},
  note/.style={font=\scriptsize, align=center}
]

\draw[draw=black!10, fill=black!1, rounded corners=8pt]
  (0,0) rectangle (13.7,5.5);

\fill[blue!5]
  (.85,.75)
  -- (.85,1.18)
  .. controls (2.50,1.76) and (4.55,2.16) .. (6.50,2.46)
  .. controls (8.40,2.74) and (10.15,2.94) .. (12.05,3.10)
  -- (12.05,.75)
  -- cycle;

\draw[axis] (.85,.75) -- (12.35,.75) node[below right] {round \(t\)};
\draw[axis] (.85,.75) -- (.85,4.45)
  node[above, align=center] {largest\\deficit};

\draw[ct]
  (.85,1.18)
  .. controls (2.50,1.76) and (4.55,2.16) .. (6.50,2.46)
  .. controls (8.40,2.74) and (10.15,2.94) .. (12.05,3.10);

\node[note, text=blue!65!black, anchor=west] at (9.35,3.25)
  {proved threshold \(c_t = \tilde O(\sqrt{t/n})\)};

\draw[run]
  (1.10,.98)
  -- (1.65,1.25)
  -- (2.20,1.09)
  -- (2.90,1.48)
  -- (3.60,1.30)
  -- (4.30,1.68)
  -- (5.00,1.49)
  -- (5.80,1.86)
  -- (6.60,1.66)
  -- (7.40,2.03)
  -- (8.20,1.83)
  -- (9.00,2.19)
  -- (9.80,1.99)
  -- (10.60,2.35)
  -- (11.40,2.15)
  -- (12.00,2.50);

\foreach \x/\y in {
1.10/.98,1.65/1.25,2.20/1.09,2.90/1.48,
3.60/1.30,4.30/1.68,5.00/1.49,5.80/1.86,
6.60/1.66,7.40/2.03,8.20/1.83,9.00/2.19,
9.80/1.99,10.60/2.35,11.40/2.15,12.00/2.50}
{
  \fill[green!40!black] (\x,\y) circle (1.25pt);
}

\node[note, text=green!35!black, anchor=west] at (2.50,2.18)
  {\(p\)-potential rule};

\draw[lin] (1.10,.98) -- (11.90,4.12);
\node[note, text=red!55!black, anchor=west] at (9.75,4.14)
  {linear drift};

\end{tikzpicture}
\caption{Trajectory intuition for the upper bound. The green curve schematically represents the largest deficit along a run of the \(p\)-potential rule. It may go up and down from round to round, but the analysis keeps it below the proved threshold \(\tilde O(\sqrt{t/n})\). The dotted red line shows the kind of linear drift that the theorem rules out.}
\Description{A schematic graph of the largest deficit over rounds. A green zig-zag curve for the p-potential rule remains below a dashed blue curve labeled proved threshold c_t. A red dotted line labeled linear drift rises faster.}
\label{fig:deficit-trajectory-intuition}
\end{figure}
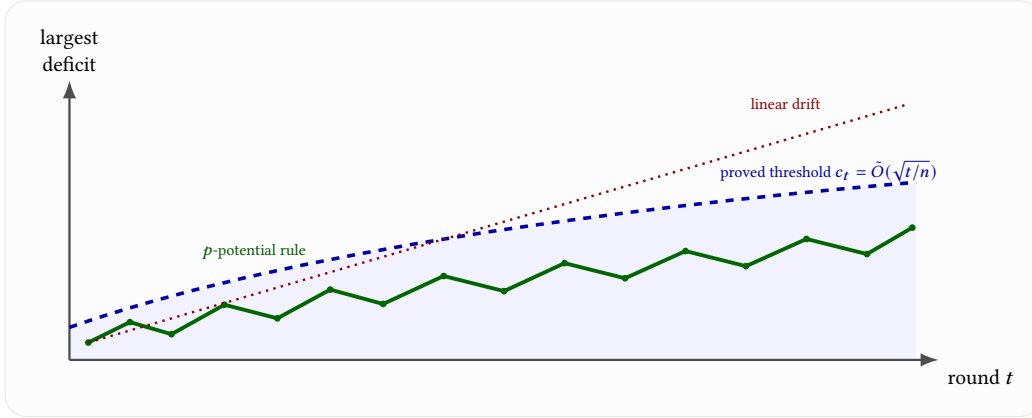

\begin{toappendix}
\subsection*{Auxiliary inequalities}
We use three standard analytic inequalities.
The Jensen bound converts a sum of powers into a power of a sum, the curvature bound controls the Taylor remainder of our potential component, and the concavity-transfer inequality converts a bound on the potential into a bound on its \(p\)-th root.
\begin{lemma}[Jensen power-sum bound]
\label[lemma]{lem:jensen-power-sum-general}
For any real numbers $F_1,\ldots,F_m>0$ and any $\alpha\in[0,1]$,
\[
\sum_{i=1}^m F_i^\alpha
\ \le\
m^{1-\alpha}\Bigl(\sum_{i=1}^m F_i\Bigr)^\alpha.
\]
\end{lemma}

\begin{proof}
For any $\alpha\in[0,1]$, the map $x\mapsto x^\alpha$ is concave on $\mathbb R_{>0}$.
Applying Jensen's inequality with weights $1/m$ gives
\[
\frac1m\sum_{i=1}^m F_i^\alpha
\ \le\
\Bigl(\frac1m\sum_{i=1}^m F_i\Bigr)^\alpha.
\]
Multiply by $m$ and rewrite $m\cdot m^{-\alpha}=m^{1-\alpha}$.
\end{proof}

\begin{lemma}[Global curvature bound]
\label[lemma]{lem:curvature-general}
Let $p\ge 1$, $\sigma\ge1$, and $f(u)=(u^2+4p^2\sigma^2)^p$.
For any $u\ge 0$ and any $\xi\in[-\sigma,\sigma]$, the second derivative of $f$ satisfies
\[
f''(u+\xi)\ \le\ 4\sqrt{e}\,p^2\,(u^2+4p^2\sigma^2)^{p-1}.
\]
\end{lemma}

\begin{proof}
First compute $f''$ at a point $u$:
\[
f'(u)=2pu\,(u^2+4p^2\sigma^2)^{p-1},
f''(u)=2p(u^2+4p^2\sigma^2)^{p-2}\bigl((2p-1)u^2+4p^2\sigma^2\bigr).
\]
Since $(2p-1)u^2+4p^2\sigma^2 \le 2p(u^2+4p^2\sigma^2)$ for $p\ge 1$, we get
\[
f''(u)\ \le\ 4p^2\,(u^2+4p^2\sigma^2)^{p-1}.
\]
Now fix $|\xi|\le \sigma$. Applying the above bound at $u+\xi$ gives
\[
f''(u+\xi)\le 4p^2\bigl((u+\xi)^2+4p^2\sigma^2\bigr)^{p-1}
=4p^2(u^2+4p^2\sigma^2)^{p-1}\,R(u)^{\,p-1},
\]
where
\[
R(u):=\frac{(u+\xi)^2+4p^2\sigma^2}{u^2+4p^2\sigma^2}
=1+\frac{2u\xi+\xi^2}{u^2+4p^2\sigma^2}
\le 1+\frac{2u\sigma+\sigma^2}{u^2+4p^2\sigma^2}.
\]

If $p=1$, then $R(u)^{p-1}=1$ and the claim follows.
Assume $p>1$. We claim that for all \(u\ge0\),
\[
\frac{2u\sigma+\sigma^2}{u^2+4p^2\sigma^2}\ \le\ \frac{1}{2p-2}.
\]
Indeed, 
this is equivalent (after multiplying by $(2p-2)(u^2+4p^2\sigma^2)$, which is positive) to
\[
u^2+4p^2\sigma^2-(2p-2)(2u\sigma+\sigma^2)\ \ge\ 0.
\]
The LHS equals
\[
u^2-4(p-1)u\sigma+4p^2\sigma^2-(2p-2)\sigma^2
=(u-2(p-1)\sigma)^2+(6p-2)\sigma^2\ \ge\ 0
\]
Therefore $R(u)\le 1+\frac{1}{2p-2}$ and since $p-1=\frac{2p-2}{2}$,
\[
R(u)^{p-1}\le \Bigl(1+\frac{1}{2p-2}\Bigr)^{(2p-2)/2}\le \sqrt{e}.
\]
Substituting back proves the lemma.
\end{proof}

\begin{lemma}[Concavity transfer]
\label[lemma]{lem:concavity-transfer-general}
Let $p\ge 1$ and $g(x):=x^{1/p}$ for $x\ge 0$.
For any $a>0$ and $b\ge 0$,
\[
g(b)-g(a)\ \le\ g'(a)\,(b-a)
\ =\ \frac{1}{p}\,a^{1/p-1}(b-a).
\]
\end{lemma}

\begin{proof}
As $1/p\in[0,1]$, the function $g$ is concave on $\mathbb R_{\ge 0}$.
A concave function lies below its tangent line at $a$:
\[
g(b)\ \le\ g(a)+g'(a)(b-a).
\]
Rearrange to obtain the claim.
\end{proof}
\end{toappendix}

\subsection*{Bounding the growth of the potential}
We now turn the local balance conditions into a bound on the growth of the potential.
The argument has three steps: first control one tracked requirement, then sum over all requirements in one round, and finally telescope over time.

\begin{lemmarep}[Bound on a single tracked requirement]
\label{lem:moment-taylor}
Assume the shift inequality \eqref{eq:first-moment-shift}, the mean bound \eqref{eq:first-moment-mean} and the squared-size bound \eqref{eq:second-moment-sigma}.
Fix a round $t+1$ and a tracked requirement ${q}\in Q$.
Let $z:=z_{q}^t$.
Then the average of the potential-component $f$ over all reference actions is bounded as follows:
\begin{equation}
\label{eq:moment-taylor}
\frac1n\sum_{k=1}^n f\!\bigl(z_{q}^{t+1}(\hat a^{t+1}_{k})\bigr)
\ \le\
f(z)\;+\;\frac{2\sqrt{e}\,p^2\,\sigma^2}{n}\,(z^2+4p^2\sigma^2)^{p-1}.
\end{equation}
\end{lemmarep}
\begin{proofsketch}
    Using the Taylor--Lagrange theorem, we present the potential-component $f$ at time $t+1$ as a sum of $f$ at time $t$ ($f(z)$), plus a first-order term, plus a second-order term (the Lagrange remainder).
    We then average this expression over all reference actions.

    The mean bound implies that, on average, the first-order term vanishes.  The squared-size bound implies an upper bound on the average of the second-order term. This yields the claimed upper bound on the average value of $f$.
\end{proofsketch}
\begin{proof}

Fix a round $t+1$ and a tracked requirement ${q}\in{Q}$.
Write $z:=z_{q}^t$ and $\Delta_k:=\Delta_{q}^{t+1}(k)$ for $k\in[n]$ ($k$ is an index of a reference actions).
Note that $z\ge 0$ since all deficits $z_{q}^t$ are nonnegative by definition.
We will repeatedly use that $f(u)=(u^2+4p^2\sigma^2)^p$ is even and nondecreasing in $|u|$ (hence nondecreasing on $\mathbb R_{\ge 0}$), and thus $f([x]_+)\le f(x)$ for all $x\in\mathbb R$.

\par\noindent
\noindent\textbf{Step 0 (reduce to $f(z+\Delta_k)$).}
By the mean bound shift \eqref{eq:first-moment-shift},
$
z_{q}^{t+1}(\hat a^{t+1}_{k})\le [z+\Delta_k]_+.
$
Since $f$ is nondecreasing on $\mathbb R_{\ge 0}$ and $[z+\Delta_k]_+\ge 0$, we get
\[
f\!\bigl(z_{q}^{t+1}(\hat a^{t+1}_{k})\bigr)\le f([z+\Delta_k]_+).
\]
Moreover $[x]_+\le |x|$ for all $x$, and $f$ is even and nondecreasing in $|x|$, hence
$f([x]_+)\le f(|x|)=f(x)$. Therefore
$
f([z+\Delta_k]_+)\le f(z+\Delta_k)
$,
and averaging yields
\[
\frac1n\sum_{k=1}^n f\!\bigl(z_{q}^{t+1}(\hat a^{t+1}_{k})\bigr)
\le
\frac1n\sum_{k=1}^n f(z+\Delta_k).
\]

\Cref{fig:taylor-cancellation-microscope} summarizes the next steps of the proof: the average first-order Taylor term is harmless by the mean bound, while the squared-size bound controls the remaining curvature term.

\begin{figure}[ht]
\centering
\begin{tikzpicture}[
  >=Latex,
  font=\small,
  outer/.style={draw=black!12, fill=black!1, rounded corners=9pt},
  panel/.style={draw=black!11, fill=white, rounded corners=7pt},
  title/.style={font=\bfseries, align=center},
  note/.style={font=\scriptsize, align=center},
  axis/.style={->, draw=black!60, line width=.7pt},
  curve/.style={draw=blue!60!black, line width=1.3pt},
  tangent/.style={draw=orange!80!black, dashed, line width=.95pt},
  good/.style={draw=green!45!black, fill=green!7, rounded corners=5pt, inner sep=3pt, align=center},
  cost/.style={draw=orange!75!black, fill=orange!7, rounded corners=5pt, inner sep=3pt, align=center},
  resultbox/.style={draw=blue!55!black, fill=blue!5, rounded corners=5pt, inner sep=4pt, align=center},
  dot/.style={circle, fill=black, inner sep=1.45pt},
  ldot/.style={circle, fill=green!45!black, inner sep=1.45pt},
  rdot/.style={circle, fill=red!60!black, inner sep=1.45pt}
]

\draw[outer] (0,0) rectangle (14,6.3);

\draw[panel, fill=blue!2] (.35,.72) rectangle (5.05,5.95);
\node[title, text=blue!55!black] at (2.70,5.60) {1. Reference moves};
\node[note, text=black!65] at (2.70,5.28) {one coordinate, one round};

\draw[axis] (.85,1.35) -- (4.65,1.35) node[below] {$u$};
\draw[axis] (.85,1.35) -- (.85,4.85) node[left] {$f(u)$};

\draw[curve] plot[smooth] coordinates {
  (.95,1.48) (1.35,1.58) (1.75,1.82) (2.15,2.18)
  (2.50,2.58) (2.85,3.03) (3.25,3.50) (3.70,3.90)
  (4.20,4.16) (4.55,4.23)
};
\draw[tangent] (1.35,1.77) -- (4.30,4.15);

\coordinate (zbase) at (2.70,1.35);
\coordinate (zpt) at (2.70,2.83);
\draw[dashed, black!45] (zbase) -- (zpt);
\node[dot] at (zpt) {};
\node[below=5pt] at (zbase) {$z_q^t$};

\foreach \x in {1.55,2.05,2.35} {
  \node[ldot] at (\x,1.35) {};
}
\foreach \x in {3.10,3.45,3.95} {
  \node[rdot] at (\x,1.35) {};
}

\draw[->, green!45!black, line width=.9pt] (2.62,1.02) -- (1.62,1.02);
\draw[->, red!60!black, line width=.9pt] (2.78,1.02) -- (3.85,1.02);

\node[note, text=black!70] at (2.70,.3)
  {$z+\Delta_k$ from the reference actions};

\draw[panel, fill=green!2] (5.30,.72) rectangle (8.95,5.95);
\node[title, text=green!40!black] at (7.125,5.60) {2. Tangent view};
\node[note, text=black!65] at (7.125,5.28) {first-order terms};

\draw[orange!80!black, line width=1.1pt] (5.85,3.35) -- (8.40,3.35);
\fill[black] (7.125,3.35) circle (1.7pt);
\draw[black!45, line width=.7pt] (7.125,3.35) -- (6.75,2.65) -- (7.50,2.65) -- cycle;
\node[note] at (7.125,2.35) {tangent at \(z\)};

\foreach \dx in {-1.00,-.62,-.32} {
  \draw[->, green!45!black, line width=.95pt] (7.125,4.15+0.18*\dx) -- ++(\dx,0);
}
\foreach \dx in {.42,.74,1.06} {
  \draw[->, red!60!black, line width=.95pt] (7.125,4.15+0.18*\dx) -- ++(\dx,0);
}

\node[good, text width=2.85cm] at (7.125,0.3)
  { \(\displaystyle \frac1n\sum_k \Delta_k\le0\)\\[-1mm]};

\draw[panel, fill=orange!2] (9.20,.72) rectangle (13.50,5.95);
\node[title, text=orange!75!black] at (11.35,5.60) {3. Curved view};
\node[note, text=black!65] at (11.35,5.28) {second-order cost};

\draw[axis] (9.70,1.35) -- (13.08,1.35) node[below] {$u$};

\draw[curve] plot[smooth] coordinates {
  (9.76,1.55) (10.10,1.63) (10.45,1.84) (10.80,2.16)
  (11.13,2.55) (11.45,2.98) (11.82,3.43) (12.22,3.78)
  (12.65,4.00) (13.00,4.05)
};
\draw[tangent] (10.05,1.75) -- (12.75,3.95);

\node[note, text=orange!80!black, anchor=west] at (11,4.10)
  {curve above tangent};

\foreach \x/\yt/\yc in {
  10.35/1.99/1.78,
  10.70/2.27/2.05,
  11.95/3.29/3.56,
  12.35/3.61/3.86,
  12.70/3.89/4.00
} {
  \draw[orange!75!black, line width=.9pt] (\x,\yt) -- (\x,\yc);
  \fill[orange!75!black] (\x,\yc) circle (1.15pt);
}

\node[cost, text width=3.25cm] at (11.35,0.3)
  {\(\displaystyle \frac1n\sum_k\Delta_k^2\le \frac{\sigma^2}{n}\)};

\draw[->, black!35, line width=.9pt] (5.05,3.35) -- (5.30,3.35);
\draw[->, black!35, line width=.9pt] (8.95,3.35) -- (9.20,3.35);

\end{tikzpicture}
\caption{Taylor cancellation for one tracked requirement. The reference actions move the current deficit \(z=z_q^t\) to nearby values \(z+\Delta_k\). In the tangent view, these moves contribute only the averaged first-order term, which is nonpositive by the mean bound. What remains is the curvature cost; the squared-size bound controls the sum of these small gaps. This is the local step behind Lemma~\ref{lem:moment-taylor}.}
\Description{A three-panel illustration of Taylor cancellation. The first panel shows reference shifts around the current deficit. The second panel shows that the shifts cancel in the tangent view. The third panel shows that only small curvature gaps remain, controlled by the squared-size bound.}
\label{fig:taylor-cancellation-microscope}
\end{figure}

\par\noindent
\noindent\textbf{Step 1 (Taylor--Lagrange for each $k$).}
The function $f(u)=(u^2+4p^2\sigma^2)^p$ is $C^2$ on $\mathbb R$ since $u^2+4p^2\sigma^2>0$.
Hence, by the Taylor--Lagrange theorem, for each $k\in[n]$, there exists $\theta_k\in(0,1)$ such that
\[
f(z+\Delta_k)
=
f(z)+f'(z)\Delta_k+\frac12 f''(z+\theta_k\Delta_k)\Delta_k^2.
\]
Averaging over $k\in[n]$ gives
\[
\frac1n\sum_{k=1}^n f(z+\Delta_k)
=
f(z)
+f'(z)\Bigl(\frac1n\sum_{k=1}^n \Delta_k\Bigr)
+\frac12\cdot \frac1n\sum_{k=1}^n f''(z+\theta_k\Delta_k)\Delta_k^2.
\]

\par\noindent
\noindent\textbf{Step 2 (drop the linear term).}
Since $z\ge 0$, we have $f'(z)\ge 0$.
By the mean bound \eqref{eq:first-moment-mean}, $\sum_{k=1}^n \Delta_k\le 0$, hence
$f'(z)\cdot \frac1n\sum_k \Delta_k\le 0$.
Dropping this (nonpositive) term yields
\[
\frac1n\sum_{k=1}^n f(z+\Delta_k)
\le
f(z)
+\frac12\cdot \frac1n\sum_{k=1}^n f''(z+\theta_k\Delta_k)\Delta_k^2.
\]

\par\noindent
\noindent\textbf{Step 3 (uniform curvature bound).}
By the squared-size bound \eqref{eq:second-moment-sigma}, $|\Delta_k|\le \sigma$ for every $k$. Since $\theta_k\in(0,1)$, we have $|\theta_k\Delta_k|\le \sigma$.
By \Cref{lem:curvature-general},
\[
f''(z+\theta_k\Delta_k)
\le
4\sqrt e\,p^2\,(z^2+4p^2\sigma^2)^{p-1}
\qquad \text{for every }k.
\]
Substituting,
\[
\frac1n\sum_{k=1}^n f(z+\Delta_k)
\le
f(z)
+\frac12\cdot 4\sqrt e\,p^2\,(z^2+4p^2\sigma^2)^{p-1}
\Bigl(\frac1n\sum_{k=1}^n \Delta_k^2\Bigr).
\]

\par\noindent
\noindent\textbf{Step 4 (use the squared-size bound).}
By \eqref{eq:second-moment-sigma}, $\sum_{k=1}^n \Delta_k^2\le \sigma^2$, hence
$\frac1n\sum_{k=1}^n \Delta_k^2\le \frac{\sigma^2}{n}$.
Therefore,
\[
\frac1n\sum_{k=1}^n f(z+\Delta_k)
\le
f(z)+\frac{2\sqrt e\,p^2\,\sigma^2}{n}\,(z^2+4p^2\sigma^2)^{p-1}.
\]
Combine with Step~0 to obtain \Cref{eq:moment-taylor}.
\end{proof}

We now sum the single-requirement estimate over all tracked requirements and use the fact that the rule chooses the action with minimum next potential.

\begin{lemmarep}[One-step growth bound]
\label{lem:transform-general}
Assume the local balance conditions.
For every $t\ge 0$,
\[
\Psi^{t+1}\ \le\ \Psi^t\ +\ \frac{2\sqrt e\,p\,\sigma^2}{n}\,m^{1/p}.
\]
\end{lemmarep}
\begin{proofsketch}
$\Psi^{t+1}$ is the minimum potential over all feasible actions.
The minimum is at most the average over the reference actions.
Applying Lemma~\ref{lem:moment-taylor} to each tracked requirement and then Jensen's inequality gives the stated one-step growth bound.
\end{proofsketch}
\begin{proof}
Fix a round $t+1$ and abbreviate $\hat a_k:=\hat a^{t+1}_{k}$.
~
Denote:
\begin{align}
    \Phi^{t+1}(a)\ :=\ \sum_{{q}\in Q} f\!\bigl(z_{q}^{t+1}(a)\bigr),
\end{align} 
so that $\Psi^{t+1}(a) = (\Phi^{t+1}(a))^{1/p}$.

\par\noindent
\noindent\textbf{Step 1 (minimum $\le$ average).}
Since $x\mapsto x^{1/p}$ is nondecreasing on $\mathbb R_{\ge 0}$, minimizing
$\Psi^{t+1}(a)=(\Phi^{t+1}(a))^{1/p}$ over $a\in A^{t+1}$ is equivalent to
minimizing $\Phi^{t+1}(a)$.
Hence, for the action $a^{t+1}$ chosen by the rule,
\[
(\Psi^{t+1})^p=\Phi^{t+1}
=\min_{a\in A^{t+1}}\Phi^{t+1}(a)
\le
\min_{k\in[n]}\Phi^{t+1}(\hat a_k)
\le
\frac1n\sum_{k=1}^n \Phi^{t+1}(\hat a_k).
\]

\par\noindent
\noindent\textbf{Step 2 (swap sums).}
Using $\Phi^{t+1}(a)=\sum_{{q}\in Q} f(z_{q}^{t+1}(a))$,
\begin{align*}
\frac1n\sum_{k=1}^n \Phi^{t+1}(\hat a_k)
=
\frac1n\sum_{k=1}^n\sum_{{q}\in Q} f\!\bigl(z_{q}^{t+1}(\hat a_k)\bigr)
=
\sum_{{q}\in Q}
\left(\frac1n\sum_{k=1}^n f\!\bigl(z_{q}^{t+1}(\hat a_k)\bigr)\right).
\end{align*}

\par\noindent
\noindent\textbf{Step 3 (apply the one-step Taylor bound for each tracked requirement).}
For each ${q}\in Q$, apply Lemma~\ref{lem:moment-taylor} with $z=z_{q}^t$:
\[
\frac1n\sum_{k=1}^n f\!\bigl(z_{q}^{t+1}(\hat a_k)\bigr)
\le
f(z_{q}^t)
+\frac{2\sqrt e\,p^2\,\sigma^2}{n}\bigl((z_{q}^t)^2+4p^2\sigma^2\bigr)^{p-1}.
\]
Summing over ${q}\in Q$ and using
\[
(\Psi^t)^p=\sum_{q\in Q} f(z_{q}^t),
\]
we get
\begin{equation}
\label{eq:onestep-raw}
(\Psi^{t+1})^p
\le
(\Psi^t)^p
+
\frac{2\sqrt e\,p^2\,\sigma^2}{n}
\sum_{{q}\in Q}\bigl((z_{q}^t)^2+4p^2\sigma^2\bigr)^{p-1}.
\end{equation}

\par\noindent
\noindent\textbf{Step 4 (Jensen on the $(p-1)$-powers).}
Define
\[
F_{q}:=f(z_{q}^t)=((z_{q}^t)^2+4p^2\sigma^2)^p,
\qquad
\alpha:=1-\frac1p\in[0,1].
\]
Then
\[
((z_{q}^t)^2+4p^2\sigma^2)^{p-1}=(F_{q})^\alpha
\]
and
\[
\sum_{q\in Q}F_{q}=(\Psi^t)^p.
\]
By \Cref{lem:jensen-power-sum-general},
\[
\sum_{{q}\in Q}\bigl((z_{q}^t)^2+4p^2\sigma^2\bigr)^{p-1}
=
\sum_{{q}\in Q}F_{q}^\alpha
\le
m^{1/p}(\Psi^t)^{p-1}.
\]
Plugging into \eqref{eq:onestep-raw} gives
\begin{equation}
\label{eq:ba-bound-general}
(\Psi^{t+1})^p-(\Psi^t)^p
\le
\frac{2\sqrt e\,p^2\,\sigma^2}{n}\,m^{1/p}(\Psi^t)^{p-1}.
\end{equation}

\par\noindent
\noindent\textbf{Step 5 (concavity transfer from $p$-powers).}
Let
\[
a:=(\Psi^t)^p,
\qquad
b:=(\Psi^{t+1})^p.
\]
Then
\[
\Psi^{t+1}-\Psi^t=b^{1/p}-a^{1/p}.
\]
Since $m\ge 1$, we have
\[
(\Psi^t)^p=\sum_{{q}\in Q} f(z_{q}^t)\ge m(4p^2\sigma^2)^p>0,
\]
so $a>0$.
By \Cref{lem:concavity-transfer-general},
\[
\Psi^{t+1}-\Psi^t
\le
\frac1p\,a^{1/p-1}(b-a)
=
\frac1p\cdot (\Psi^t)^{1-p}\cdot
\bigl((\Psi^{t+1})^p-(\Psi^t)^p\bigr).
\]
Using \eqref{eq:ba-bound-general}, we conclude that
\[
\Psi^{t+1}-\Psi^t
\le
\frac{2\sqrt e\,p\,\sigma^2}{n}\,m^{1/p}.
\]
\end{proof}

Telescoping the one-step growth bound gives the following anytime bound on the potential.

\begin{lemmarep}[Anytime potential bound]
\label{lem:potential-bound-general}
Assume the local balance conditions.
For all $t\ge 0$,
\[
\Psi^t\ \le\ m^{1/p}\cdot
\Bigl(4p^2\sigma^2+\frac{2\sqrt e\,p\,\sigma^2\,t}{n}\Bigr).
\]
\end{lemmarep}
\begin{proofsketch}
We sum the one-step increment bounds of Lemma~\ref{lem:transform-general}.
\end{proofsketch}

\begin{proof}
Recall that:
\begin{align}
    \Phi^{t+1}(a)\ :=\ \sum_{{q}\in Q} f\!\bigl(z_{q}^{t+1}(a)\bigr),
\end{align} 
so that $\Psi^{t+1}(a) = (\Phi^{t+1}(a))^{1/p}$.

Since $z_q^0=0$ for all $q\in Q$,
\begin{equation}
\label{eq:Psi0-general}
\Phi^0=m\cdot (4p^2\sigma^2)^p,
\qquad
\Psi^0=m^{1/p}\cdot 4p^2\sigma^2.
\end{equation}
By Lemma~\ref{lem:transform-general}, for each
$\ell=0,1,\ldots,t-1$,
\[
\Psi^{\ell+1}
\le
\Psi^\ell
+
\frac{2\sqrt e\,p\,\sigma^2}{n}\,m^{1/p}.
\]
Summing these $t$ inequalities yields
\[
\Psi^t
\le
\Psi^0
+
\frac{2\sqrt e\,p\,\sigma^2}{n}\,t\,m^{1/p}.
\]
Substituting \eqref{eq:Psi0-general} gives
\[
\Psi^t
\le
m^{1/p}\cdot
\Bigl(4p^2\sigma^2+\frac{2\sqrt e\,p\,\sigma^2\,t}{n}\Bigr).
\]
\end{proof}

\subsection*{From potential bounds to \texorpdfstring{$c_t$}{c_t}-fairness}
The previous lemma bounds the potential value \(\Psi^t\). We now translate this analytic bound back into the original fairness language: a requirement can be \(c\)-violated only if it contributes a large term to the potential.

\begin{theoremrep}[Few $c$-violated requirements]
\label{thm:violated-general}
Assume the local balance conditions.
Fix $t\ge 0$ and $c>0$. Then
\begin{equation}
\label{eq:violated-bound-general}
|\mathcal{D}^{t}(c)|
\ \le\
m\left(
\frac{4p^2\sigma^2+\frac{2\sqrt e\,p\,\sigma^2\,t}{n}}{c^2+4p^2\sigma^2}
\right)^{\!p}
\end{equation}
\end{theoremrep}
\begin{proofsketch}
    Every $c$-violated requirement contributes to $(\Psi^t)^p$ at least $(c^2 + 4p^2\sigma^2)^p$.
    Using the upper bound on $(\Psi^t)^p$ from Lemma~\ref{lem:potential-bound-general}, we get an upper bound on the number of $c$-violated requirements.
\end{proofsketch}

\begin{proof}
By definition, if ${q}\in\mathcal{D}^{t}(c)$ then $z_{q}^t>c$.
Since $z_{q}^t>c\ge 0$, we have $(z_{q}^t)^2>c^2$, hence
$(z_{q}^t)^2+4p^2\sigma^2\ge c^2+4p^2\sigma^2$, and therefore
\[
f(z_{q}^t)=\bigl((z_{q}^t)^2+4p^2\sigma^2\bigr)^p\ge (c^2+4p^2\sigma^2)^p.
\]

Summing over all ${q}\in\mathcal{D}^{t}(c)$ gives
\[
\sum_{{q}\in{Q}} f(z_{q}^t)
\ \ge\
\sum_{{q}\in\mathcal{D}^{t}(c)} f(z_{q}^t)
\ \ge\
|\mathcal{D}^{t}(c)|\,(c^2+4p^2\sigma^2)^p.
\]
Since $(\Psi^t)^p=\sum_{{q}\in{Q}} f(z_{q}^t)$, rearranging yields
\[
|\mathcal{D}^{t}(c)|
\le
\left(\frac{\Psi^t}{c^2+4p^2\sigma^2}\right)^p .
\]
Using Lemma~\ref{lem:potential-bound-general}, this gives
\[
|\mathcal{D}^{t}(c)|
\le
m\left(
\frac{4p^2\sigma^2+\frac{2\sqrt e\,p\,\sigma^2\,t}{n}}{c^2+4p^2\sigma^2}
\right)^p.
\]

\end{proof}
The next corollary turns the bound on the number of \(c\)-violated requirements into a \(c_t\)-fairness guarantee.
The idea is to choose \(c_t\) so that the upper bound on \(|\mathcal D^t(c_t)|\) is smaller than \(1\); since this quantity is an integer, it must then be \(0\).
\begin{corollary}[Explicit perpetual $c_{t}$-fairness]
\label[corollary]{cor:ct-feasible-general}
Assume the local balance conditions.
For every $t\ge 0$, define
\[
c_{t}\ :=\ \sigma\,m^{1/(2p)}\cdot \sqrt{\,4p^2+\frac{2\sqrt e\,p\,t}{n}\,}
\]
Then for every \(t\), the outcome is \(c_{t}\)-fair, that is,
$\mathcal{D}^{t}(c_{t})=\emptyset$.
\end{corollary}

\begin{proof}
By \Cref{thm:violated-general},
\[
|\mathcal{D}^{t}(c)|
\le
m\left(
\frac{4p^2\sigma^2+\frac{2\sqrt e\,p\,\sigma^2\,t}{n}}{c^2+4p^2\sigma^2}
\right)^{\!p}.
\]
Let
\[
c=c_t
=
m^{1/(2p)}\cdot
\sqrt{\,4p^2\sigma^2+\frac{2\sqrt e\,p\,\sigma^2\,t}{n}\,}.
\]
Then
\[
c_t^2
=
m^{1/p}\left(4p^2\sigma^2+\frac{2\sqrt e\,p\,\sigma^2\,t}{n}\right),
\]
and therefore
\[
c_t^2+4p^2\sigma^2
=
m^{1/p}\left(4p^2\sigma^2+\frac{2\sqrt e\,p\,\sigma^2\,t}{n}\right)+4p^2\sigma^2.
\]
Thus
\[
|\mathcal{D}^{t}(c_t)|
\le
m\left(
\frac{4p^2\sigma^2+\frac{2\sqrt e\,p\,\sigma^2\,t}{n}}
{m^{1/p}\!\left(4p^2\sigma^2+\frac{2\sqrt e\,p\,\sigma^2\,t}{n}\right)+4p^2\sigma^2}
\right)^p
<
m\left(\frac{1}{m^{1/p}}\right)^p
=1.
\]
Since $|\mathcal{D}^{t}(c_{t})|$ is an integer, it must be $0$.
\end{proof}

\paragraph{Choosing \(p\): many small violations vs.\ one worst violation.}

The parameter \(p\) determines how strongly the potential emphasizes large entries of the deficit vector. 
For a fixed threshold \(c\), \Cref{thm:violated-general} gives the tail bound
\[
|\mathcal{D}^{t}(c)|
\le
m\left(
\frac{4p^2\sigma^2+\frac{2\sqrt e\,p\,\sigma^2t}{n}}{c^2+4p^2\sigma^2}
\right)^p,
\]
so smaller \(p\) can be preferable when the goal is to bound how many tracked requirements exceed a fixed threshold \(c\). 
To prove \(c_t\)-fairness, however, we need the stronger conclusion that no requirement is violated, namely \(|\mathcal{D}^{t}(c_t)|=0\). 
For this purpose the dependence on \(m=|Q|\) becomes crucial. 
For example, \Cref{cor:ct-feasible-general} gives, with \(p=2\),
\[
c_t^{(p=2)}
=
m^{1/4}
\sigma\sqrt{16+\frac{4\sqrt e\,t}{n}},
\]
which still has polynomial dependence on \(m\). 
Choosing \(p=\ln(2m)\) instead gives \(m^{1/(2p)}\le\sqrt e\), and hence
\[
c_t
\le
\sqrt e\,\sigma
\sqrt{4\ln^2(2m)+\frac{2\sqrt e\,\ln(2m)t}{n}}.
\]
Thus \(p=\ln(2m)\) is the choice we use for the \(c_t\)-fairness guarantees in the applications.
\begin{corollary}[Logarithmic choice of the potential parameter]
\label[corollary]{cor:best-ct}
  If the $p$-potential algorithm is executed with \(p=\ln(2m)\), then for every time $t\geq 1$, the algorithm is $c_{t}$-fair for
\[
c_{t} \ :=\ \sqrt e\cdot \sqrt{\,4p^2\sigma^2+\frac{2\sqrt e\,p\,\sigma^2\,t}{n}\,}
\ =\ O\!\left(\sigma p+\sigma\sqrt{\frac{pt}{n}}\right)
\ =\ O\!\left(\sigma\ln m+\sigma\sqrt{\frac{t\ln m}{n}}\right).
\]
\end{corollary}
\begin{proof}
Apply \Cref{cor:ct-feasible-general} with \(p=\ln(2m)\). Then
\(m^{1/(2p)}=m^{1/(2\ln(2m))}\le \sqrt e\), and the displayed bound on
\(c_t\) follows immediately. The two big-\(O\) estimates follow from
\(p=\Theta(\ln m)\).
\end{proof}
Thus, after a short ``warm-up'' period of order $(n\ln{m})$ rounds, the square-root term dominates and we get 
$c_{t}\in O\left(\sigma\sqrt{\frac{t\ln m}{n}}\right)$.

\begin{remark}
\label{rem:why-deterministic}
One may allow a randomized online algorithm, i.e., at time step $t{+}1$ the algorithm
may sample $a^{t+1}$ from a distribution over $A^{t+1}$ based on the realized history $H^t$  and the current round input.
However, 
randomization does not improve worst-case performance:
the lower bound in \Cref{sec:lower_bound} applies to \emph{all} online algorithms, including randomized ones (equivalently, the adversary forces failure with probability $1$,
and therefore the same lower bound also holds in expectation). 
Accordingly, we present a deterministic rule that guarantees fairness even in the worst case, and against an adaptive adversary.
\end{remark}

Before turning to the applications, it is useful to separate the rest of the paper into two parts. First, we study the full-history model: we instantiate the framework in several domains and prove that the resulting \(\sqrt t\)-type dependence is unavoidable. Second, motivated by this lower bound, we study limited-memory fairness notions and show that discounted memory gives guarantees controlled by the \emph{squared history length}, which we define in that section.

\section{Applications}
\label{sec:applications}

This section instantiates the full-history deficit framework (\Cref{sec:definition}) for three fully-online fairness problems, and then gives one extension for classical \(\EFc\) which is not fully-online. Throughout this paper we assume \(n\ge2\).
Since we already motivated the domains and stated the informal fairness goals, we focus here on:
(i) the concrete definition of the tracked requirements and deficits \(z_q^t\),
(ii) a full verification of the local balance conditions (Definition~\ref{def:first-moment}),
and (iii) the resulting perpetual fairness guarantee.
We also summarize per round implementability and complexity in Table~\ref{tab:instantiation-complexity}.

Thus, each application has the same structure: define the deficit, choose the right
normalization, check local balance, and then apply the general theorem. The first application is written most explicitly, since it is the motivating example; the later applications follow the same template and can be read as parallel instantiations.


\subsection{Case 1: Approximately-proportional online item allocation}
\label{subsec:ex-propxc}

There are \(n\) agents.
At each round \(t{+}1\), an item \(g^{t+1}\) arrives and must be allocated irrevocably to one agent
\(a^{t+1}\in[n]\). Hence, the set of actions is the same at all times: $A^{t+1} = [n]$.

Agents have nonnegative additive valuations \(v_i(\cdot)\).
Let \(P_i^t\) be agent \(i\)'s bundle after \(t\) rounds and \(G^t=\bigcup_{i}P_i^t\) be the total set of items up to time $t$.
Agent \(i\)'s proportional share is \(\Prop_i^t:=\frac1n v_i(G^t)\), and the proportionality deficit is $d_i^t := \Prop_i^t - v_i(P_i^t)$. This is a full-history requirement: all items that have arrived up to time \(t\) are counted.
We scale the deficit of $i$ by the largest missed item (the largest item not allocated to $i$), defined as
\[
U_i^t \;:=\; \max\{v_i(g): g\in G^t\setminus P_i^t\},
\quad\text{with }U_i^t:=0 \text{ if } G^t = P_i^t.
\]
The \(\PROPax{c}\) condition at each time \(t\) is \(d_i^t \le c\,U_i^t\) for all \(i\in[n]\).
This scaling is useful because the same absolute deficit can have different
meaning in different histories. If the missed items have small value, then even
a small deficit may be important. If large missed items appeared, then a larger
absolute deficit can still be explained by only a few missed items.

We instantiate the framework with one tracked requirement per agent: \(Q=[n]\) and \(m=|Q|=n\).
The deficit of tracked requirement $i$ is defined as:
\[
z_i^t := \sdiv{[d_i^t]_+}{U_i^t},
\]
Here \(\sdiv{x}{y}\) denotes safe division: it equals \(x/y\) if \(y>0\), and equals \(0\) if \(y=0\) (in which case also \(x=0\)).

We claim that $c$-fairness implies \(\PROPax{c}\). This follows from the following lemma.
\begin{lemma}[Normalized proportionality deficit implies \(\PROPax{c}\)]
For any $i\in[n]$ and $t\geq 0$,
$z_i^t \leq c$
implies
$d_i^t\leq c U_i^t$.
\end{lemma}
\begin{proof}
If \(d_i^t\le 0\), then \(d_i^t\le 0\le c\,U_i^t\) since \(U_i^t\ge 0\), so the claim holds.

If \(U_i^t=0\), then every missed item has value \(0\) for agent \(i\) (by nonnegativity), hence
\(v_i(G^t)=v_i(P_i^t)\).
Therefore
\(
d_i^t=\frac1n v_i(G^t)-v_i(P_i^t)=-(1-\frac1n)v_i(G^t)\le 0,
\)
so again the claim holds.

If \(U_i^t>0\) and \(d_i^t>0\), so \([d_i^t]_+=d_i^t\), then \(z_i^t=\frac{d_i^t}{U_i^t}\le c\) trivially implies \(d_i^t\le c\,U_i^t\).
\end{proof}

Before proving that the balance conditions hold, we give an example of our generic fully-online rule, using the valuations in \Cref{tab:intro-greedy-counterexample}.
\begin{example}[How the potential rule avoids the greedy mistake]
Fix \(0<\epsilon<1/2\), and suppose ties are broken in favor of agent~1.
For the two tracked requirements, write the normalized deficit vector as \((z_1,z_2)\).
The potential rule uses
\[
f(u)=(u^2+4p^2\sigma^2)^p,
\qquad
\Phi(z_1,z_2)=f(z_1)+f(z_2),
\]
and chooses the recipient that minimizes the post-round value of \(\Phi\).

Using the valuations in \Cref{tab:intro-greedy-counterexample}, the first six rounds give the following post-round vectors and potential comparisons:
\[
\begin{array}{c|c|c|c}
\text{round} & \text{if given to agent }1 & \text{if given to agent }2 & \text{chosen agent} \\
\hline
1 &
(0,1/2) &
(1/2,0) &
1 \\
2 &
(0,(1+\epsilon)/2) &
(0,(1-\epsilon)/2) &
2 \\
3 &
(0,1/2) &
(1/2,(1-2\epsilon)/2) &
1 \\
4 &
(0,1) &
(0,0) &
2 \\
5 &
(0,\epsilon/2) &
(\epsilon/2,0) &
1 \\
6 &
(0,(1+\epsilon)/2) &
(0,0) &
2
\end{array}
\]
In round 1, for both choices 
$\Phi=f(1/2)+f(0)$, so the tie is broken in favor of agent 1.
Then, the choices follow directly from monotonicity of \(f\) on \([0,\infty)\).
For instance, in round \(2\),
\[
f(0)+f((1-\epsilon)/2)<f(0)+f((1+\epsilon)/2),
\]
so the item is given to agent~2.
The key round is round \(3\): although agent~2 currently has the larger deficit,
\[
f(0)+f(1/2)
<
f(1/2)+f((1-2\epsilon)/2),
\]
because \(f((1-2\epsilon)/2)>f(0)\).
Thus the potential rule gives the third item to agent~1, unlike the deficit-greedy rule from \Cref{tab:intro-greedy-counterexample}.

The resulting choices are therefore
\[
1,\ 2,\ 1,\ 2,\ 1,\ 2.
\]
The point is that the rule minimizes the whole post-round deficit vector through \(\Phi\), rather than simply repairing the currently largest deficit.
\end{example}

\Cref{fig:greedy-vs-potential-clean} isolates the key fork in the example.
At round \(3\), agent \(2\)'s current deficit is larger, but the arriving item has
value vector \((1,\epsilon)\). Thus giving the item to agent \(2\) repairs \(z_2\)
only slightly while creating a new large \(z_1\). The potential rule avoids this
mistake by comparing the whole next deficit vector.

\begin{figure}[ht]
\centering
\resizebox{0.8\linewidth}{!}{%
\begin{tikzpicture}[
  x=1cm, y=1cm,
  >=Stealth,
  font=\small,
  panel/.style={rounded corners=5pt, draw=black!30, fill=black!1, line width=0.65pt},
  redpanel/.style={rounded corners=5pt, draw=red!62!black, fill=red!3, line width=0.95pt},
  bluepanel/.style={rounded corners=5pt, draw=blue!62!black, fill=blue!3, line width=0.95pt},
  axis/.style={->, draw=black!78, line width=0.65pt},
  redarr/.style={-{Stealth[length=2.7mm]}, very thick, draw=red!72!black},
  bluearr/.style={-{Stealth[length=2.7mm]}, very thick, draw=blue!72!black},
  whitelabel/.style={fill=white, inner sep=1.4pt, rounded corners=1pt}
]

\draw[redpanel] (0,0) rectangle (7.05,5.55);
\node[font=\bfseries, text=red!72!black] at (3.525,5.18) {greedy: choose agent 2};

\begin{scope}[shift={(0.85,0.65)}]
  \draw[axis] (0,0) -- (5.55,0) node[right] {$z_1$};
  \draw[axis] (0,0) -- (0,4.15) node[above] {$z_2$};
  \node[font=\scriptsize, below left] at (0,0) {$0$};

  \coordinate (s)  at (0.00,2.35);
  \coordinate (g1) at (2.65,2.10);
  \coordinate (g2) at (3.55,2.70);
  \coordinate (g3) at (4.55,3.27);
  \coordinate (g4) at (5.15,3.72);

  \fill[black] (s) circle (2.6pt);
  \draw[redarr] (s) -- (g1);
  \draw[redarr] (g1) -- (g2);
  \draw[redarr] (g2) -- (g3);
  \draw[redarr] (g3) -- (g4);
  \foreach \p in {g1,g2,g3,g4}{
    \fill[red!72!black] (\p) circle (2.6pt);
  }

  \node[font=\bfseries\small, text=red!72!black, whitelabel] at (3.55,0.46)
    {drifts away};
  \node[font=\scriptsize, text=red!72!black, whitelabel, anchor=west] at (2.78,1.95)
    {new $z_1$};
\end{scope}

\draw[bluepanel] (7.65,0) rectangle (14.7,5.55);
\node[font=\bfseries, text=blue!72!black] at (11.175,5.18) {potential rule: choose agent 1};

\begin{scope}[shift={(8.50,0.65)}]
  \draw[axis] (0,0) -- (5.55,0) node[right] {$z_1$};
  \draw[axis] (0,0) -- (0,4.15) node[above] {$z_2$};
  \node[font=\scriptsize, below left] at (0,0) {$0$};

  \coordinate (s)  at (0.00,2.35);
  \coordinate (p1) at (0.00,2.90);
  \coordinate (p2) at (1.25,1.10);
  \coordinate (p3) at (0.15,0.28);
  \coordinate (p4) at (0.00,0.13);

  \fill[black] (s) circle (2.6pt);
  \draw[bluearr] (s) -- (p1);
  \draw[bluearr] (p1) -- (p2);
  \draw[bluearr] (p2) -- (p3);
  \draw[bluearr] (p3) -- (p4);
  \foreach \p in {p1,p2,p3,p4}{
    \fill[blue!72!black] (\p) circle (2.6pt);
  }

  \node[font=\bfseries\small, text=blue!72!black, whitelabel] at (3.45,0.46)
    {returns near $0$};
  \node[font=\scriptsize, text=blue!72!black, whitelabel, anchor=west] at (0.20,3.02)
    {$z_1$ stays $0$};
\end{scope}

\end{tikzpicture}%
}
\caption{The key mistake of the deficit-greedy rule in the two-agent example.
At round \(3\), agent \(2\) has the larger current deficit, but the arriving item
has value \(1\) for agent \(1\) and only \(\epsilon\) for agent \(2\).
Greedy follows the current deficit and creates a new \(z_1\), after which the
trajectory drifts away. The potential rule gives the item to agent \(1\), keeping
the trajectory near the origin.}
\Description{A two-panel comparison. The top strip states that at round 3 the current deficit points to agent 2, while the item value points to agent 1. In the left panel, the greedy rule chooses agent 2 and the trajectory drifts away from the origin. In the right panel, the potential rule chooses agent 1 and the trajectory returns near the origin.}
\label{fig:greedy-vs-potential-clean}
\end{figure}

Now, we prove that our deficits satisfy the balance conditions.

For a candidate recipient agent \(a\in[n]\), the hypothetical post-round deficits are
\begin{equation}
\label{eq:propx-deficit}
z_i^{t+1}(a)\;:=\;\sdiv{\bigl[d_i^{\,t+1}(a)\bigr]_+}{U_i^{t+1}(a)}\qquad(i\in[n]).
\end{equation}

\begin{lemmarep}[Local balance for normalized proportional deficits]
For each round \(t{+}1\) and an agent \(i\),
the next-step deficit functions \eqref{eq:propx-deficit} satisfy local balance conditions
(\ref{eq:first-moment-mean}, \ref{eq:second-moment-sigma}), with \(\sigma^2=1\).
\end{lemmarep}
\begin{proofsketch}
Fix a round \(t{+}1\) and agent \(i\), and let \(x_i:=v_i(g^{t+1})\). After normalizing by the updated missed-item scale, giving the item to \(i\) decreases \(i\)'s normalized deficit by roughly \((1-\frac1n)x_i\), while giving it to another agent increases it by roughly \(x_i/n\). Thus, across the \(n\) reference actions, the positive and negative shifts cancel in the mean bound, and their squared size is at most \(1\). This gives the local balance conditions with \(\sigma^2=1\).
\end{proofsketch}
\begin{proof}
Let 
\[
x_i := v_i(g^{t+1}),\qquad
\widetilde U_i := \max\{U_i^t,\,x_i\},\qquad
s := \sdiv{x_i}{\widetilde U_i}\in[0,1].
\]
Observe that \(\widetilde U_i=0\) implies \(x_i=0\), hence \(s=0\).

The arrival of $g^{t+1}$ increases \(v_i(G^t)\) by \(x_i\), hence increases \(\Prop_i^t\) by \(x_i/n\).
If the item is allocated to \(a=i\), then \(v_i(P_i^t)\) increases by \(x_i\), whereas 
if the item is allocated to \(a\neq i\), then \(v_i(P_i^t)\) is unchanged. Hence,
\begin{align*}
d_i^{t+1}(a)
=
\begin{cases}
    \bigl(\Prop_i^t+\tfrac{x_i}{n}\bigr)-\bigl(v_i(P_i^t)+x_i\bigr)
=d_i^t-\Bigl(1-\frac1n\Bigr)x_i & a=i
    \\
    \bigl(\Prop_i^t+\tfrac{x_i}{n}\bigr)-v_i(P_i^t)
=d_i^t+\frac{x_i}{n} & a\neq i.
\end{cases}
\end{align*}

If \(a=i\), then \(g^{t+1}\in P_i^{t+1}\), so it is \emph{not} missed by \(i\); If \(a\neq i\), then \(g^{t+1}\notin P_i^{t+1}\), so it \emph{is} missed by \(i\). Hence,
\[
U_i^{t+1}(a)=
\begin{cases}
U_i^t & a=i
\\
\max\{U_i^t,x_i\}=\widetilde U_i & a\neq i.
\end{cases}
\]

We now upper bound the deficits \(z_i^{t+1}(a)\).

\noindent\emph{Case \(a=i\).}
Since \([u]_+\ge u\) and \([\cdot]_+\) is monotone, \([u-b]_+\le [[u]_+-b]_+\) for every \(b\ge 0\) 
Hence
\[
[d_i^{t+1}(i)]_+
=
\Bigl[d_i^t-\Bigl(1-\frac1n\Bigr)x_i\Bigr]_+
\le
\Bigl[[d_i^t]_+-\Bigl(1-\frac1n\Bigr)x_i\Bigr]_+.
\]
If \(U_i^t=0\), then \(v_i(G^t)=v_i(P_i^t)\) and hence \(d_i^t\le 0\), so \(z_i^t=0\).
Moreover \(U_i^{t+1}(i)=0\) and \(d_i^{t+1}(i)=d_i^t-(1-\frac1n)x_i\le 0\), so \(z_i^{t+1}(i)=0\) as well.
Thus \eqref{eq:propx-shift-i} holds trivially. Assume henceforth \(U_i^t>0\).

Dividing by \(U_i^{t+1}(i)=U_i^t\) gives
\[
z_i^{t+1}(i)
\le
\Bigl[z_i^t-\Bigl(1-\frac1n\Bigr)\frac{x_i}{U_i^t}\Bigr]_+.
\]
Since \(\widetilde U_i=\max\{U_i^t,x_i\}\), we have \(\frac{x_i}{U_i^t}\ge \frac{x_i}{\widetilde U_i}=s\).
Subtracting a larger number only decreases the \([\,\cdot\,]_+\) value, so
\begin{equation}
\label{eq:propx-shift-i}
z_i^{t+1}(i)\le \Bigl[z_i^t-\Bigl(1-\frac1n\Bigr)s\Bigr]_+.
\end{equation}
\noindent\emph{Case \(a\neq i\).}
Since \([u]_+\ge u\) and \([\cdot]_+\) is monotone, \([u+v]_+\le [u]_+ + v\) for every \(v\ge 0\). Thus
\[
[d_i^{t+1}(a)]_+
=
\Bigl[d_i^t+\frac{x_i}{n}\Bigr]_+
\le
[d_i^t]_+ + \frac{x_i}{n}.
\]
Divide by \(U_i^{t+1}(a)=\widetilde U_i\) and use \(\widetilde U_i\ge U_i^t\) (so \(\sdiv{[d_i^t]_+}{\widetilde U_i}\le \sdiv{[d_i^t]_+}{U_i^t}=z_i^t\)):
\begin{equation}
\label{eq:propx-shift-noti}
z_i^{t+1}(a)
\le
z_i^t+\frac{1}{n}\cdot \sdiv{x_i}{\widetilde U_i}
=
z_i^t+\frac{s}{n}
\leq 
[z_i^t+\frac{s}{n}]_+
\qquad(a\neq i).
\end{equation}

Equations \eqref{eq:propx-shift-i}--\eqref{eq:propx-shift-noti} fit the required shift form
\eqref{eq:first-moment-shift} with reference actions \(\hat a^{t+1}_{a}=a\in[n]\) 
and the following increments:
\begin{align}
\label{eq:first-moment-shifts-propc}
\Delta_i^{t+1}(i)= -\Bigl(1-\frac1n\Bigr)s,
\qquad
\Delta_i^{t+1}(a)= \frac{s}{n}\quad(a\neq i).
\end{align}

It remains to prove the local balance conditions for the increments in \eqref{eq:first-moment-shifts-propc}.

\begin{align*}
\sum_{a=1}^n \Delta_i^{t+1}(a)
&=
-\Bigl(1-\frac1n\Bigr)s + (n-1)\cdot\frac{s}{n}
=0,
\\[2pt]
\sum_{a=1}^n \bigl(\Delta_i^{t+1}(a)\bigr)^2
&=
\Bigl(1-\frac1n\Bigr)^2 s^2+(n-1)\Bigl(\frac{s}{n}\Bigr)^2
=
\frac{n-1}{n}\,s^2
\le 1.
\end{align*}
Thus, the local balance conditions hold with \(\sigma^2=1\).
\end{proof}
The local balance lemma is the only model-specific ingredient needed for the general theorem.
Combining it with the implication from \(z_i^t\le c\) to \(\PROPax{c}\) gives the following perpetual item-allocation guarantee.
\begin{corollary}[Perpetual \texorpdfstring{$\PROPax{\cdot}$}{PROPxc} guarantee for online items]
\label{cor:propxc-application}
In Case~\ref{subsec:ex-propxc}, we have \(m=n\), \(n\) reference actions, and \(\sigma^2=1\).
Thus, by \Cref{cor:ct-feasible-general} 
(with \(p=\ln(2m)=\ln(2n)\))
, 
for every time \(t\) and agent \(i\in[n]\),
\[
v_i(P_i^t)\ \ge\ \Prop_i^t - c_{t}^{\PROPx}\,U_i^t,
\qquad
c_{t}^{\PROPx}
:=m^{1/(2p)}\cdot \sqrt{\,4p^2+\frac{2\sqrt e\,p\,t}{n}\,}.
\]
\end{corollary}


\paragraph{Faster per round evaluation (why it is \(O(n)\) and not \(O(n^2)\)).}
In any round, when agent $i$ does not receive an item, his updated deficit, denoted by $z_i^{\text{miss}}$, does not depend on the computation of any other deficit. Similarly, when agent $i$ receives an item, his updated deficit, denoted by $z_i^{\text{recv}}$, does not depend on the computation of any other deficit. In addition, recall that the per-round potential is a sum of per-agent deficits,
\(\Phi^t=\sum_{i=1}^n f(z_i^t)\).
Therefore, in each round, we can compute once the vector of $z_i^{\text{miss}}$ of all the agents, $z^{\text{miss}}=(z_1^{\text{miss}},\dots,z_n^{\text{miss}})$, and a baseline potential
\(
\Phi^{t+1}_{\text{miss}}=\sum_{i=1}^n f(z_i^{\text{miss}})
\). 
Then, evaluating a candidate recipient agent \(a\) is just a \emph{single-term swap}:
\(
\Phi^{t+1}(a)=\Phi^{t+1}_{\text{miss}}-f(z_a^{\text{miss}})+f(z_a^{\text{recv}}).
\)
That is, \(O(1)\) work per candidate, hence scanning all \(n\) candidates costs \(O(n)\) per round.

\begin{fullversion}
\begin{example}
Suppose the miss update gives contributions
\((f(z_1^{\text{miss}}),f(z_2^{\text{miss}}),f(z_3^{\text{miss}}),f(z_4^{\text{miss}}))=(5,2,7,4)\),
so \(\Phi^{t+1}_{\text{miss}}=5+2+7+4=18\).
If candidate \(a=3\) receives the item, only agent \(3\)'s contribution changes, say to \(f(z_3^{\text{recv}})=1\).
Then
\(
\Phi^{t+1}(3)=18-7+1=12.
\)
No other term is touched. Doing this swap for each \(a\in\{1,2,3,4\}\) is one \(O(1)\) update each \(\Rightarrow O(n)\).
\end{example}
\end{fullversion}

\subsection{Case 2: Approximately-proportional public decision-making (PDM)}
\label{subsec:ex-pdm}
There is a finite nonempty set of feasible outcomes (candidates) \(C\).
In each round \(t{+}1\), the decision maker chooses one outcome \(o^{t+1}\in C\) irrevocably, so
\(A^{t+1}=C\).
Each agent \(i\in[n]\) has a (time-varying) valuation function
\(v_i^{t+1}:C\to \mathbb{R}_{\ge 0}\), and utilities add across rounds:
\(
u_i^t \;:=\; \sum_{r=1}^t v_i^{r}(o_r).
\)
Following \citet{conitzer2017fair}, define the per round attainable maximum
\(
M_i^{t+1} \;:=\; \max_{o\in C} v_i^{t+1}(o),
\)
and the proportional share
\(
\Prop_i^t \;:=\; \frac{1}{n}\sum_{r=1}^t M_i^r.
\)
Again, this is a full-history requirement: all rounds \(1,\ldots,t\) enter the proportional share with equal weight.
Define the (unscaled) proportionality deficit
\(
d_i^t \;:=\; \Prop_i^t - u_i^t,
\)
and the running-max scale
\(
V_i^t \;:=\; \max_{r\le t} M_i^r,
\qquad\text{with }V_i^0:=0.
\)
The scale-based proportionality condition \(\PROPax{c}\) at time \(t\) is
\(
d_i^t \ \le\ c\,V_i^t
\)
for all $i\in[n]$,
equivalently \(u_i^t \ge \Prop_i^t - c\,V_i^t\).

We instantiate the framework with one tracked requirement per agent \(Q=[n]\), so \(m=|Q|=n\).
Define the deficit at time \(t\), and the hypothetical next-round deficits, as:
\begin{align}
\label{eq:pdm-deficit}
z_i^t \;:=\; \sdiv{[d_i^t]_+}{V_i^t}
~~
(i\in[n]) 
&&
z_i^{t+1}(o)
\;:=\;
\sdiv{\bigl[d_i^{t+1}(o)\bigr]_+}{V_i^{t+1}}
~~
(i\in[n], o\in C),
\end{align}
where \(V_i^{t+1}:=\max\{V_i^t,\,M_i^{t+1}\}\) (note that, in contrast to \(U_i^{t+1}\) in the item-allocation setting, here \(V_i^{t+1}\) depends only on round \(t{+}1\), not on the chosen \(o\)).

Since \(V_i^t\ge 0\), the constraint \(d_i^t\le c\,V_i^t\) is equivalent to \([d_i^t]_+\le c\,V_i^t\).
If \(V_i^t>0\) this is exactly \(z_i^t\le c\).
If \(V_i^t=0\), then \(M_i^r=0\) for all \(r\le t\), hence \(\Prop_i^t=u_i^t=0\) and \(z_i^t=0\).
Therefore \(\PROPax{c}\) holds at time \(t\) iff \(z_i^t\le c\) for all \(i\), i.e.\ \(\mathcal{D}^{t}(c)=\emptyset\).

\begin{lemmarep}[Local balance conditions for PDM]
\label{lem:pdm-moments}
For each round \(t{+}1\) and each agent \(i\in[n]\), the next-step deficits \eqref{eq:pdm-deficit}
satisfy the local balance conditions
(\ref{eq:first-moment-mean}, \ref{eq:second-moment-sigma})
with \(n\) reference actions and \(\sigma^2=1\).
\end{lemmarep}

\begin{proofsketch}
The calculation mirrors Case~\ref{subsec:ex-propxc}, with the missed-item scale replaced by the running-max
scale \(V_i^t\) (which does not depend on the chosen outcome \(o\)).
Full details are deferred to the appendix.
\end{proofsketch}

\begin{proof}
We now verify the local balance conditions.
Fix a round \(t{+}1\) and a tracked requirement \(i\in[n]\).
Abbreviate
\[
M:=M_i^{t+1},\qquad
V:=V_i^{t+1}=\max\{V_i^t,M\},\qquad
s:=\sdiv{M}{V}\in[0,1].
\]
For any outcome \(o\in C\), we have
\[
d_i^{t+1}(o)
=
\bigl(\Prop_i^t+\tfrac{1}{n}M\bigr)-\bigl(u_i^t+v_i^{t+1}(o)\bigr)
=
d_i^t+\frac{1}{n}M - v_i^{t+1}(o).
\]

\noindent\emph{Reference actions.}
For each \(k\in[n]\), choose an arbitrary favorite outcome of agent \(k\):
\[
\hat a^{t+1}_{k}\ \in\ \arg\max_{o\in C} v_k^{t+1}(o).
\]
These are our \(n\) reference actions.

\noindent\emph{Bound under \(\hat a^{t+1}_{i}\) (agent \(i\)'s favorite).}
Since \(v_i^{t+1}(\hat a^{t+1}_{i})=M\),
\[
d_i^{t+1}(\hat a^{t+1}_{i})
=
d_i^t-\Bigl(1-\frac1n\Bigr)M.
\]
Using \([x-b]_+\le [[x]_+-b]_+\) for \(b\ge 0\),
\[
\bigl[d_i^{t+1}(\hat a^{t+1}_{i})\bigr]_+
=
\Bigl[d_i^t-\Bigl(1-\frac1n\Bigr)M\Bigr]_+
\le
\Bigl[[d_i^t]_+-\Bigl(1-\frac1n\Bigr)M\Bigr]_+.
\]
Divide by \(V\) (if \(V=0\), then \(M=0\) and everything below is \(0\) by \(\sdiv{\cdot}{\cdot}\)):
\begin{align*}
z_i^{t+1}(\hat a^{t+1}_{i})
\le
\sdiv{\Bigl[[d_i^t]_+-\bigl(1-\frac1n\bigr)M\Bigr]_+}{V}
=
\Bigl[\sdiv{[d_i^t]_+}{V}-\Bigl(1-\frac1n\Bigr)\sdiv{M}{V}\Bigr]_+
\le
\Bigl[z_i^t-\Bigl(1-\frac1n\Bigr)s\Bigr]_+,
\end{align*}
where the last inequality uses \(V\ge V_i^t\), hence
\(\sdiv{[d_i^t]_+}{V}\le \sdiv{[d_i^t]_+}{V_i^t}=z_i^t\).

\noindent\emph{Bound under \(\hat a^{t+1}_{k}\) for \(k\neq i\).}
Since valuations are nonnegative, \(v_i^{t+1}(\hat a^{t+1}_{k})\ge 0\), so
\[
d_i^{t+1}(\hat a^{t+1}_{k})
=
d_i^t+\frac{1}{n}M - v_i^{t+1}(\hat a^{t+1}_{k})
\le
d_i^t+\frac{1}{n}M.
\]
Using \([x+b]_+\le [x]_+ + b\) for \(b\ge 0\),
\[
\bigl[d_i^{t+1}(\hat a^{t+1}_{k})\bigr]_+
\le
[d_i^t]_+ + \frac{1}{n}M,
\]
and dividing by \(V\) gives
\[
z_i^{t+1}(\hat a^{t+1}_{k})
\le
\sdiv{[d_i^t]_+}{V}+\frac{1}{n}\sdiv{M}{V}
\le
z_i^t+\frac{s}{n}
\le
\Bigl[z_i^t+\frac{s}{n}\Bigr]_+
\qquad(k\neq i),
\]
again using \(V\ge V_i^t\).

The bounds match the shift form \eqref{eq:first-moment-shift} with increments
\[
\Delta_i^{t+1}(i)= -\Bigl(1-\frac1n\Bigr)s,
\qquad
\Delta_i^{t+1}(k)= \frac{s}{n}\quad(k\neq i).
\]
\[
\sum_{k=1}^n \Delta_i^{t+1}(k)
=
-\Bigl(1-\frac1n\Bigr)s+(n-1)\frac{s}{n}
=0,
\quad
\sum_{k=1}^n \bigl(\Delta_i^{t+1}(k)\bigr)^2
=
\Bigl(1-\frac1n\Bigr)^2 s^2+(n-1)\Bigl(\frac{s}{n}\Bigr)^2
=
\frac{n-1}{n}\,s^2
\le 1.
\]
Hence the local balance conditions hold with \(\sigma^2=1\).
\end{proof} 

Thus the public-decision deficits satisfy the same local balance structure as in the item-allocation case.
Applying the general theorem yields the following prefix-wise proportionality guarantee for public decisions.

\begin{corollary}[Perpetual \texorpdfstring{$\PROPax{\cdot}$}{PROPxc} guarantee for public decisions]
\label{cor:pdm-application}
In Case~\ref{subsec:ex-pdm}, we have \(m=n\).
By Lemma~\ref{lem:pdm-moments}, the instance has \(n\) reference actions and \(\sigma^2=1\).
Thus, by \Cref{cor:ct-feasible-general},
(with \(p=\ln(2m)\))
for every time \(t\) and agent \(i\in[n]\),
\(
\qquad u_i^t \ \ge\ \Prop_i^t - c_{t}^{\textnormal{PDM}}\,V_i^t,
\qquad
c_{t}^{\textnormal{PDM}}
:=m^{1/(2p)}\cdot \sqrt{\,4p^2+\frac{2\sqrt e\,p\,t}{n}\,}.
\) 
\end{corollary}

\begin{fullversion}
\paragraph{Back to the court-agenda example.}
Interpret each term (or decision epoch) as a round, and let $C$ be the set of feasible agenda choices in that round (e.g., the set of petitions that could be selected for review).
Treat each constituency/issue-area (or any other fairness-relevant category) as an ``agent'' with value $v_i^{t}(o)$ for selecting outcome $o\in C$ in round~$t$.
Then $M_i^{t}=\max_{o\in C}v_i^{t}(o)$ represents the largest stake faced by group~$i$ in that round, and our $\PROPax{c_{t}^{\textnormal{PDM}}}$ guarantee bounds each group's cumulative deficit relative to this moving reference value at every prefix~$t$.
\end{fullversion}

\subsection{Case 3: Approximately envy-free online item allocation}
\label{subsec:ex-efc}

We return to the full-history online item-allocation model of \Cref{subsec:ex-propxc}.
For agents \(i\neq j\), define directed envy at time \(t\) as
\(
\Envy_{i\to j}^t := v_i(P_j^t)-v_i(P_i^t). \quad
\)
We use the scale
\(
\quad s_{(i,j)}^t := \max\{v_i(g): g\in P_j^t\},
\quad\text{with }s_{(i,j)}^t:=0\text{ if }P_j^t=\emptyset. \quad
\)
The $\EFax{c}$ condition is: \(\Envy_{i\to j}^t \le c\cdot s_{(i,j)}^t\).

We instantiate the framework with one tracked requirement for each ordered pair:
\[
Q:=\{(i,j)\in[n]\times[n]: i\neq j\},
\qquad
m=|Q|=n(n-1).
\]

For each tracked requirement \(q=(i,j)\in Q\), define the normalized envy at time \(t\) by
\[
z_{(i,j)}^{t}
\;:=\;
\sdiv{\bigl[\Envy_{i\to j}^{\,t}\bigr]_+}{s_{(i,j)}^t}.
\]
In round \(t{+}1\), let \(z_{(i,j)}^{t+1}(r)\) denote the value after hypothetically allocating \(g^{t+1}\) to agent \(r\):
\begin{equation}
\label{eq:ef-deficit}
z_{(i,j)}^{t+1}(r)
\;:=\;
\sdiv{\bigl[\Envy_{i\to j}^{\,t+1}(r)\bigr]_+}{s_{(i,j)}^{t+1}(r)},
\qquad
s_{(i,j)}^{t+1}(r):=
\begin{cases}
\max\{s_{(i,j)}^t,\ v_i(g^{t+1})\}, & r=j,\\
s_{(i,j)}^t, & r\neq j.
\end{cases}
\end{equation}

Since \(s_{(i,j)}^t\ge 0\), the constraint \(\Envy_{i\to j}^t\le c\,s_{(i,j)}^t\) is equivalent to
\([\Envy_{i\to j}^t]_+\le c\,s_{(i,j)}^t\).
If \(s_{(i,j)}^t>0\) this is exactly \(z_{(i,j)}^t\le c\).
If \(s_{(i,j)}^t=0\), then \(v_i(g)=0\) for every \(g\in P_j^t\), hence \(v_i(P_j^t)=0\) and (under nonnegative valuations) \(\Envy_{i\to j}^t\le 0\), so \(z_{(i,j)}^t=0\) as well.
Therefore \(\EFax{c}\) holds at time \(t\) iff \(z_{(i,j)}^t\le c\) for all \((i,j)\), i.e.\ \(\mathcal{D}^{t}(c)=\emptyset\).

\begin{lemmarep}[Local balance conditions for max-normalized envy]
\label{lem:ef-moments}
For each round \(t{+}1\) and each ordered pair \((i,j)\in Q\), the next-step normalized envy functions \eqref{eq:ef-deficit}
satisfy the local balance conditions
(\ref{eq:first-moment-mean}, \ref{eq:second-moment-sigma})
with reference actions \(\hat a^{t+1}_{r}=r\in[n]\) and \(\sigma^2=2\).
\end{lemmarep}

\begin{proofsketch}
Fix a round \(t{+}1\) and a pair \((i,j)\).
Let \(x=v_i(g^{t+1})\) and \(\widetilde s:=\max\{s_{(i,j)}^t,x\}\). Set \(\alpha:=\sdiv{x}{\widetilde s}\in[0,1]\).
If the item is allocated to \(j\), normalized envy increases by at most \(\alpha\); if allocated to \(i\), it decreases by at least \(\alpha\) after truncation by \([\cdot]_+\); otherwise it is unchanged.
This yields a shift representation with increments \(+\alpha\), \(-\alpha\), and zeros, so the mean is \(0\) and the squared \(\ell_2\)-norm is at most \(2\).
\end{proofsketch}

\begin{proof}
Fix a round \(t{+}1\) and a tracked requirement \(q=(i,j)\in Q\).
Write \(x:=v_i(g^{t+1})\), \(\widetilde s:=\max\{s_{(i,j)}^t,x\}\), and \(\alpha:=\sdiv{x}{\widetilde s}\in[0,1]\).

Consider the three relevant recipient types:

\noindent\emph{Recipient \(r=j\) (give the item to the envied agent).}
Then \(v_i(P_j)\) increases by \(x\) while \(v_i(P_i)\) remains unchanged, so \(\Envy_{i\to j}\) increases by \(x\).
The scale updates to \(\widetilde s\).  If \(\widetilde s=0\), then \(x=0\), \(s_{(i,j)}^t=0\), and both sides below are \(0\). Otherwise,
Thus
\[
z_{(i,j)}^{t+1}(j)
\le
\frac{[\Envy_{i\to j}^{t}+x]_+}{\widetilde s}
\le
\frac{[\Envy_{i\to j}^{t}]_+}{\widetilde s}+\frac{x}{\widetilde s}
\le
z_{(i,j)}^t + \alpha.
\]
where the last inequality uses \(\widetilde s\ge s_{(i,j)}^t\).

\noindent\emph{Recipient \(r=i\) (give the item to the envious agent).}
Then \(v_i(P_i)\) increases by \(x\), so \(\Envy_{i\to j}\) decreases by \(x\).
The scale \(s_{(i,j)}\) does not increase, since the item is not added to \(P_j\).
If \(s_{(i,j)}^t=0\), then \(v_i(P_j^t)=0\), hence \(\Envy_{i\to j}^t\le0\), so \(z_{(i,j)}^t=0\). After giving the item to \(i\), envy only decreases, and therefore \(z_{(i,j)}^{t+1}(i)=0\le[z_{(i,j)}^t-\alpha]_+\).

If \(s_{(i,j)}^t>0\), then
\[
z_{(i,j)}^{t+1}(i)
=
\frac{[\Envy_{i\to j}^{t}-x]_+}{s_{(i,j)}^t}
\le
\left[z_{(i,j)}^t-\frac{x}{s_{(i,j)}^t}\right]_+
\le
\left[z_{(i,j)}^t-\alpha\right]_+,
\]
because \(s_{(i,j)}^t\le\widetilde s\), so \(x/s_{(i,j)}^t\ge x/\widetilde s=\alpha\).

\noindent\emph{Recipient \(r\notin\{i,j\}\).}
Neither \(P_i\) nor \(P_j\) changes for this pair, so both envy and scale stay unchanged:
\[
z_{(i,j)}^{t+1}(r)\le z_{(i,j)}^t.
\]

Therefore the shift form \eqref{eq:first-moment-shift} holds with reference actions \(\hat a^{t+1}_{r}=r\in[n]\) and increments
\[
\Delta_{(i,j)}^{t+1}(j)=\alpha,\qquad
\Delta_{(i,j)}^{t+1}(i)=-\alpha,\qquad
\Delta_{(i,j)}^{t+1}(r)=0\ \ (r\notin\{i,j\}).
\]
Clearly,
\[
\sum_{r=1}^n \Delta_{(i,j)}^{t+1}(r)=(-\alpha)+\alpha+0=0,
\qquad
\sum_{r=1}^n \bigl(\Delta_{(i,j)}^{t+1}(r)\bigr)^2 = \alpha^2+\alpha^2 = 2\alpha^2 \le 2.
\]
Hence the local balance conditions hold with \(\sigma^2=2\).
\end{proof}

The normalized envy deficits are therefore another direct instance of the generic framework.
Since \(z_{(i,j)}^t\le c\) is exactly the max-normalized envy condition, the general theorem gives the following guarantee.

\begin{corollary}[Perpetual max-normalized \texorpdfstring{$\mathrm{EF}$}{EF} guarantee]
\label{cor:ef-application}
In Case~\ref{subsec:ex-efc}, we have \(m=n(n-1)\).
By Lemma~\ref{lem:ef-moments}, the instance has \(n\) reference actions and \(\sigma^2=2\).
Thus, by \Cref{cor:ct-feasible-general} (with \(p=\ln(2m)\)), for every time \(t\) and all \(i\neq j\),
\[
\Envy_{i\to j}^t \ \le\ c_{t}^{EF}\,s_{(i,j)}^t,
\qquad
c_{t}^{EF}:=m^{1/(2p)}\cdot \sqrt{\,8p^2+\frac{4\sqrt e\,p\,t}{n}\,}.
\] 
\end{corollary}

\paragraph{Faster per round evaluation (why it is \(O(n^2)\), not \(O(n^3)\)).}
Here \(Q=\{(i,j)\in[n]\times[n]: i\neq j\}\), so the raw sum is
\(\Phi^t=\sum_{i\neq j} f\!\bigl(z_{(i,j)}^{t}\bigr)\).
Allocating \(g^{t+1}\) to agent \(a\) changes only \(P_a\), so \emph{only pairs that involve \(a\)}
can change: \((i,a)\) and \((a,i)\) for all \(i\neq a\).
All pairs \((i,j)\) with \(i\neq a\neq j\) are unchanged, i.e.\ \(z_{(i,j)}^{t+1}(a)=z_{(i,j)}^t\).
Thus, if we maintain the current pair terms \(f(z_{(i,j)}^t)\), then for a fixed candidate \(a\) we can compute
\(
\Phi^{t+1}(a)=\Phi^{t}
-\sum_{i\neq a}\bigl(f(z_{(i,a)}^{t})+f(z_{(a,i)}^{t})\bigr)
+\sum_{i\neq a}\bigl(f(z_{(i,a)}^{t+1}(a))+f(z_{(a,i)}^{t+1}(a))\bigr),
\)
touching only \(2(n-1)=O(n)\) pairs. Scanning all \(n\) candidates therefore takes \(O(n^2)\) per round.

\begin{fullversion}
\begin{example}The ordered pairs are exactly the 6 terms: (1,2),(2,1),(1,3),(3,1),(2,3),(3,2). If candidate \(a=2\) receives the item, the only pairs that may change are those containing \(2\): (1,2),(2,1),(2,3),(3,2) (4 terms).
The remaining pairs (1,3),(3,1), 2 terms do not involve \(2\), hence are unchanged and need not be recomputed.
So evaluating candidate \(a=2\) touches only \(2(n-1)=4\) pair terms, i.e.\ \(O(n)\) work.
Doing this for each candidate \(a\in\{1,2,3\}\) costs \(n\cdot O(n)=O(n^2)\) per round.
\end{example}
\end{fullversion}

\subsection{Classical \texorpdfstring{$\EFc$}{EFc}}
\label{subsec:ex-efc-classical}
This application is technically different from the previous three applications. It is included to show that the framework can also recover a threshold-based route to classical \(\EFc\), but unlike the previous applications it assumes that the set of possible item values is known in advance.

The standard definition of \(\EFc\) requires that for every \(i\neq j\), agent \(i\)'s envy toward \(j\) can be eliminated
by removing at most \(c\) goods from \(P_j^t\).
For additive valuations this is equivalent to
\[
[\Envy_{i\to j}^t]_+\ \le\ \sum_{g\in B_{(i,j)}^t(c)} v_i(g),
\]
where \(B_{(i,j)}^t(c)\) denotes the \(c\) goods in \(P_j^t\) with largest \(v_i\)-value
(padding with \(0\)-valued dummy goods if fewer than \(c\) exist).


Without restricting the possible item values in advance, no online algorithm can guarantee
classical \(\EFc\) with \(c=o(t)\) against an adaptive adversary: the exact
worst-case parameter at horizon \(T\) is \(\lceil T/n\rceil\); see
\Cref{thm:classical-efc-linear-lb}. The result below gives a complementary
sublinear guarantee when the finite pool of potential item values is known in
advance. It is perpetual, but not fully online, because it
requires this value pool as prior information.

We denote the finite set of potential item values by \(\Theta\) and let \(L=|\Theta|\). We assume that every realized value \(v_i(g)\), for every agent \(i\) and every item \(g\), belongs to \(\Theta\). For example, if item values are limited to integers between $0$ and $1000$ (as in the popular fair division website \url{spliddit.org}), then we can take $L=1001$.

\paragraph{Translation to our framework (threshold deficits).}
Define the set of tracked requirements as the following set of triples:
\[
Q\ :=\ \{(i,j,\tau): i,j\in[n],\ i\neq j,\ \tau\in\Theta\},
\qquad m=|Q|= n(n-1)L.
\]

For a threshold \(\tau \ge 0\), define the \(\tau\)-\emph{count} of agent \(j\)'s bundle as seen by agent \(i\):
\[
C_{i\rightarrow j,\tau}^t \ :=\ \bigl|\{g\in P_j^t:\ v_i(g)\ge \tau\}\bigr|.
\]
The deficit of tracked requirement $(i,j,\tau)$ at time $t$ is defined as:
\(
d_{i\to j,\tau}^t \ :=\ C_{i\rightarrow j,\tau}^t\ -\ C_{i\rightarrow i,\tau}^t
\quad
z_{i\to j,\tau}^t \ :=\ [d_{i\to j,\tau}^t]_+.
\) For a candidate recipient \(a\in[n]\), define the hypothetical post-round deficits
\(
z_{(i,j,\tau)}^{t+1}(a)\ :=\ \bigl[d_{i\to j,\tau}^{\,t+1}(a)\bigr]_+,
\quad (i,j,\tau)\in Q.
\) We claim that $c$-fairness (bounding all deficits by \(c\)) implies \(\mathrm{EF}
_{\lceil c\rceil}\).

\begin{lemmarep}[Threshold deficits imply \texorpdfstring{$\EFc$}{EFc}]
\label{lem:threshold-to-efc}
Fix \(t\ge 0\) and \(c\ge 0\). If for every \((i,j,\tau)\in Q\),
\(
z_{i\to j,\tau}^t=[d_{i\to j,\tau}^t]_+\ \le\ c,
\)
then the allocation at time \(t\) is \(\EFparam{\lceil c\rceil}\), i.e.\ for all \(i\neq j\), letting \(B:=B_{(i,j)}^t(\lceil c\rceil)\), 
\(
[\Envy_{i\to j}^t]_+\ \le\ \sum_{g\in B} v_i(g).
\)
\end{lemmarep}
\begin{proofsketch}
Let $k=\lceil c\rceil$. The condition $[d_{i\to j,\tau}^t]_+\le c$ implies
$C_{i\rightarrow j,\tau}^t - C_{i\rightarrow i,\tau}^t \le k$ for every threshold $\tau\in\Theta$.
Sort $i$'s values of $P_j^t$ as $x_1\ge x_2\ge\cdots$ and of $P_i^t$ as $y_1\ge y_2\ge\cdots$ (pad by zeros).
If $x_{k+r}>y_r$ for some $r$, then at $\tau=x_{k+r}$ agent $j$ has at least $k+r$ items $\ge\tau$ while agent $i$
has at most $r-1$, giving a deficit of at least $k+1$, which is a contradiction. Hence $x_{k+r}\le y_r$ for all $r$, so the value of
$j$ after removing its top $k$ items is at most $v_i(P_i^t)$, and thus $i$'s envy is at most the value of those $k$ items.
\end{proofsketch}
\begin{proof}
Fix agents $i\neq j$ and let $k:=\lceil c\rceil$.

\paragraph{Step 1: deficit bound inequality.}
For every threshold $\tau\in\Theta$, the assumption $z_{i\to j,\tau}^t=[d_{i\to j,\tau}^t]_+\le c\le k$ implies
\begin{equation}\label{eq:d-upper-k}
d_{i\to j,\tau}^t
\;=\;
C_{i\rightarrow j,\tau}^t - C_{i\rightarrow i,\tau}^t
\;\le\;
k.
\end{equation}
(Indeed, if $d_{i\to j,\tau}^t>k$, then $[d_{i\to j,\tau}^t]_+=d_{i\to j,\tau}^t>k$, contradicting
$[d_{i\to j,\tau}^t]_+\le k$.)

\paragraph{Step 2: sorted value lists.}
Let $P_j^t$ be the (finite) bundle of agent $j$ at time $t$, and write it as
$P_j^t=\{g^{(j)}_1,\dots,g^{(j)}_{\ell_j}\}$ where $\ell_j:=|P_j^t|$.
Consider the multiset of real numbers
\[
\{\,v_i(g): g\in P_j^t\,\}.
\]
Sort these numbers in nonincreasing order (breaking ties arbitrarily) and denote the resulting list by
\[
x_1 \ge x_2 \ge \cdots \ge x_{\ell_j}.
\]
For convenience, extend this to an infinite sequence by defining $x_r:=0$ for all $r>\ell_j$.

Similarly, let $P_i^t=\{g^{(i)}_1,\dots,g^{(i)}_{\ell_i}\}$ where $\ell_i:=|P_i^t|$,
and sort the multiset $\{\,v_i(g): g\in P_i^t\,\}$ in nonincreasing order to obtain
\[
y_1 \ge y_2 \ge \cdots \ge y_{\ell_i},
\]
and define $y_r:=0$ for all $r>\ell_i$.

With this notation,
\[
v_i(P_j^t)=\sum_{r\ge 1} x_r
\qquad\text{and}\qquad
v_i(P_i^t)=\sum_{r\ge 1} y_r,
\]
where the sums are actually finite because all but finitely many terms are $0$.

\paragraph{Step 3: rank inequality.}
We claim that for every integer $r\ge 1$,
\begin{equation}\label{eq:key-rank}
x_{k+r}\ \le\ y_r.
\end{equation}

\emph{Proof of the claim.}
If $x_{k+r}=0$, then \eqref{eq:key-rank} holds because $y_r\ge 0$.
Assume now that $x_{k+r}>0$ and suppose for contradiction that $x_{k+r}>y_r$.
Set $\tau:=x_{k+r}$ (note $\tau>0$). Since \(\tau\) is a realized item value, \(\tau\in\Theta\).

By definition of sorting, at least the first $k+r$ numbers in the list $(x_1,x_2,\dots)$ are $\ge \tau$,
so there are at least $k+r$ goods in $P_j^t$ with $v_i(g)\ge\tau$:
\[
C_{i\rightarrow j,\tau}^t \ \ge\ k+r.
\]
On the other hand, $y_r<\tau$ means that \emph{at most} $r-1$ goods in $P_i^t$ have $v_i(g)\ge\tau$:
\[
C_{i\rightarrow i,\tau}^t \ \le\ r-1.
\]
Therefore,
\[
d_{i\to j,\tau}^t
=
C_{i\rightarrow j,\tau}^t - C_{i\rightarrow i,\tau}^t
\ \ge\
(k+r)-(r-1)
=
k+1,
\]
contradicting \eqref{eq:d-upper-k}. This proves \eqref{eq:key-rank}. \hfill$\diamond$

\paragraph{Step 4: \(\EFparam{k}\) bound.}
Summing \eqref{eq:key-rank} over all $r\ge 1$ gives
\[
\sum_{r\ge 1} x_{k+r} \ \le\ \sum_{r\ge 1} y_r \;=\; v_i(P_i^t).
\]
But $\sum_{r\ge 1} x_{k+r} = \sum_{r\ge 1} x_r - \sum_{r=1}^k x_r = v_i(P_j^t)-\sum_{r=1}^k x_r$,
so
\[
v_i(P_j^t) - \sum_{r=1}^k x_r \ \le\ v_i(P_i^t),
\]
i.e.
\[
v_i(P_j^t)-v_i(P_i^t)\ \le\ \sum_{r=1}^k x_r.
\]
Taking positive parts,
\[
[\Envy_{i\to j}^t]_+
=
[v_i(P_j^t)-v_i(P_i^t)]_+
\ \le\
\sum_{r=1}^k x_r.
\]

Finally, by definition, $B_{(i,j)}^t(k)$ is the multiset of the $k$ goods in $P_j^t$ with largest $v_i$-values,
padding with $0$-valued dummy goods if $|P_j^t|<k$; therefore its total $v_i$-value is exactly
$\sum_{r=1}^k x_r$. Hence
\[
[\Envy_{i\to j}^t]_+
\ \le\
\sum_{g\in B_{(i,j)}^t(k)} v_i(g),
\qquad\text{with }k=\lceil c\rceil.
\]
Since $i\neq j$ was arbitrary, the allocation at time $t$ is $\EFparam{\lceil c\rceil}$.
\end{proof}

It remains to check that these threshold deficits satisfy the local balance conditions.
The key point is simple: for a fixed threshold \(\tau\), the new item either crosses the threshold or it does not; if it does, only the counts of the recipient can change.

\begin{lemmarep}[Local balance conditions for threshold deficits]
\label{lem:threshold-moments}
For each round \(t{+}1\) and each tracked requirement \(q=(i,j,\tau)\in Q\),
the next-step deficit functions satisfy the local balance conditions
(\ref{eq:first-moment-mean}, \ref{eq:second-moment-sigma}) with \(n\) reference actions and \(\sigma^2=2\).
\end{lemmarep}
\begin{proofsketch}
Fix a tracked requirement $q=(i,j,\tau)$ and write $s:=\mathbf{1}[v_i(g^{t+1})\ge \tau]\in\{0,1\}$.
Only allocating the new item to $i$ or to $j$ can change the deficit
$d_{i\to j,\tau}=C_{i\rightarrow j,\tau}-C_{i\rightarrow i,\tau}$:
if the item goes to $i$ then $C_{i\rightarrow i,\tau}$ increases by $s$ so $d$ decreases by $s$; if it goes to $j$
then $C_{i\rightarrow j,\tau}$ increases by $s$ so $d$ increases by $s$; otherwise $d$ stays unchanged.
Taking positive parts yields an update of the form
$z_q^{t+1}(a)\le [z_q^t+\Delta_q^{t+1}(a)]_+$ with increments
$\Delta_q^{t+1}(i)=-s$, $\Delta_q^{t+1}(j)=+s$, and $\Delta(a)=0$ for $a\notin\{i,j\}$.
These increments sum to zero and have squared sum $s^2+s^2\le 2$.
Thus the local balance conditions hold with $n$ reference actions and $\sigma^2=2$.
\end{proofsketch}
\begin{proof}
Fix \(q=(i,j,\tau)\). Let
\[
s:=s_{i,\tau}\ :=\ \mathbf{1}\!\left[v_i(g^{t+1})\ge \tau\right]\ \in\ \{0,1\}.
\]
If the item is allocated to \(a=i\), then \(C_{i\rightarrow i,\tau}\) increases by \(s\) and \(C_{i\rightarrow j,\tau}\) is unchanged, hence
\(d_{i\to j,\tau}\) decreases by \(s\).
Similarly, if the item is allocated to \(a=j\), then 
\(d_{i\to j,\tau}\) increases by \(s\).
If \(a\notin\{i,j\}\), then neither count changes, hence \(d_{i\to j,\tau}\) is unchanged. Therefore,
\[
d_{i\to j,\tau}^{t+1}(a)=
\begin{cases}
d_{i\to j,\tau}^t - s, & a=i,\\
d_{i\to j,\tau}^t + s, & a=j,\\
d_{i\to j,\tau}^t, & a\notin\{i,j\}.
\end{cases}
\]
Taking positive parts and using \([u-s]_+\le [[u]_+-s]_+\) and \([u+s]_+\le [u]_+ + s\) for \(s\ge 0\), we get:
\[
z_q^{t+1}(i)\le [z_q^t-s]_+,\quad
z_q^{t+1}(j)\le [z_q^t+s]_+,\quad
z_q^{t+1}(a)\le [z_q^t+0]_+ \ (a\notin\{i,j\}).
\]
Thus \eqref{eq:first-moment-shift} holds with reference actions \(\hat a_r^{t+1}=r\in[n]\) and increments
\[
\Delta_q^{t+1}(i)=-s_{i,\tau},\qquad \Delta_q^{t+1}(j)=s_{i,\tau},\qquad \Delta_q^{t+1}(a)=0\ (a\notin\{i,j\}).
\]
Clearly,
\[
\sum_{a=1}^n \Delta_q^{t+1}(a)=(-s)+s+0=0,
\qquad
\sum_{a=1}^n \bigl(\Delta_q^{t+1}(a)\bigr)^2 = s^2+s^2 \le 2.
\]
Hence the local balance conditions hold with \(\sigma^2=2\).
\end{proof}

We can now combine the two ingredients: threshold deficits imply classical \(\EFc\), and the threshold deficits satisfy local balance.
Applying the general theorem gives the following classical \(\EFc\) guarantee.

\begin{corollary}[Perpetual classical \texorpdfstring{$\EFc$}{EFc} guarantee (threshold form)]
\label{cor:efc-classical-threshold}
In Case~\ref{subsec:ex-efc-classical}, the local balance conditions hold with \(n\) reference actions and \(\sigma^2=2\).
Set \(p:=\ln(2m)=\ln\!\bigl(2n(n-1)L\bigr)\).
Then by \Cref{cor:ct-feasible-general}, for every time \(t\), all thresholds \((i,j,\tau)\in Q\) satisfy
\(z_{i\to j,\tau}^t\le c_t^{\mathrm{EFc}}\), where
\[
c_t^{\mathrm{EFc}}\ :=\ 
m^{1/(2p)}\cdot \sqrt{\,8p^2+\frac{4\sqrt e\,p\,t}{n}\,}
\in O\left(
\ln(nL)+\sqrt{\frac{t \ln(n L)}{n}}
\right),
\]
where $L$ is the number of potential item values.

Consequently, by Lemma~\ref{lem:threshold-to-efc}, the time-\(t\) allocation is \(\EFparam{\lceil c_t^{\mathrm{EFc}}\rceil}\).
\end{corollary}

\subsection{Summary}
\Cref{tab:instantiation-complexity} summarizes all of our application results. Ignoring the warm-up term \(4p^2\sigma^2\), the general bound scales as
\(c_{t}\in O\!\left(\sigma\sqrt{\tfrac{t\log m}{n}}\right)\) when the set of tracked requirements is fixed with \(|Q|=m\).
Thus in rows~1--3 (where \(m\in\{n,n(n{-}1)\}\) and hence \(\log m\in\Theta(\log n)\)) we have
\(c_{t}\in O\!\left(\sigma\sqrt{\tfrac{t\log n}{n}}\right)\), whereas the classical \(\EFc\) row uses the fixed set
\(
Q=\{(i,j,\tau): i,j\in[n],\ i\neq j,\ \tau\in\Theta\},
\)
of size \(n(n-1)L\), and yields
\(
c_t^{\mathrm{EFc}}\in O\!\left(\sqrt{\tfrac{t\log(nL)}{n}}\right).
\) after ignoring the warm-up term.
\begin{table*}[ht]
\centering
\caption{Summary of our 4 instantiations: framework parameters, explicit perpetual guarantees (including the warm-up term), and per round implementability (after maintaining the standard running state needed to evaluate deficits).}
\label{tab:instantiation-complexity}
\small
\setlength{\tabcolsep}{3.5pt}
\renewcommand{\arraystretch}{1.15}
\resizebox{\textwidth}{!}{%
\begin{tabular}{@{}lcccccccc@{}}
\toprule
Instantiation
& \(|Q|\)
& \(\sigma^2\)
& \(p\)
& Warm-up term
& \(|A^{t+1}|\)
& Guarantee \(\;c_{t}\)
& Per-round time
& Storage
\\
\midrule
Item allocation, \(\PROPax{c}\) 
& \(n\)
& \(1\)
& \(\ln(2n)\)
& \(4\ln^2(2n)\)
& \(n\)
& \(  \approx \sqrt e\sqrt{\frac{2\sqrt e\,\ln(2n)\,t}{n}\,}\)
& \(O(n)\)
& \(O(n)\)
\\[2pt]

Public decisions, \(\PROPax{c}\) 
& \(n\)
& \(1\)
& \(\ln(2n)\)
& \(4\ln^2(2n)\)
& \(|C|\)
& \(  \approx \sqrt e\sqrt{\frac{2\sqrt e\,\ln(2n)\,t}{n}\,}\)
& \(O(n|C|)\)
& \(O(n)\)
\\[2pt]

Item allocation, \(\EFax{c}\) 
& \(n(n-1)\)
& \(2\)
& \(\ln\!\bigl(2n(n-1)\bigr)\)
& \(8\ln^2\!\bigl(2n(n-1)\bigr)\)
& \(n\)
& \( \approx \sqrt e\sqrt{ \frac{4\sqrt e\,\ln\!\bigl(2n(n-1)\bigr)\,t}{n}\,}\)
& \(O(n^2)\)
& \(O(n^2)\)
\\
Item allocation, \(\EFc\) (known value-set) 
& \(n(n-1)L\)
& \(2\)
& \(\ln\!\bigl(2n(n-1)L\bigr)\)
& \(8\ln^2\!\bigl(2n(n-1)L\bigr)\)
& \(n\)
& \( \approx \sqrt e\sqrt{ \frac{4\sqrt e\,\ln\!\bigl(2n(n-1)L\bigr)\,t}{n}\,}\)
& \(O(n^2 L)\)
& \(O(n^2 L)\)
\\
\bottomrule
\end{tabular}%
}
\end{table*}

\begin{fullversion}
\subsection{How to plug in your own fairness notions}
\paragraph{What is defined by the model.}
To instantiate our framework, you only need to define:
(i) a set of tracked requirements ${Q}$ (e.g.\ agents, pairs, coalitions),
(ii) for each round $t+1$ and action $a\in A^{t+1}$, the post-round deficits $z_{q}^{t+1}(a)$.
No additivity or linear structure is required at this level.

\paragraph{What is proved externally (local balance).}
All guarantees above use only:
(a) a reference-action comparison family
$\hat a^{t+1}_{1},\ldots,\hat a^{t+1}_{n}$ and increments $\Delta_q^{t+1}$ whose
sum is nonpositive, and
(b) the squared-size bound
$\sum_{k=1}^n(\Delta_{q}^{t+1}(k))^2\le \sigma^2$.
In most online fairness instantiations, the reference actions correspond to
$n$ natural ``reference actions'' (e.g.\ $n$ recipients in fair division, or
$n$ preferred choices in public decision-making), and the scaling is chosen so
that single-step normalized changes are at most $1$.
\end{fullversion}

\section{The Limits of Remembering the Full History}
\label{sec:lower_bound}

For the full-history model, \Cref{cor:best-ct} shows that it is possible to guarantee $c_{t}$-fairness for
\(c_{t}\in O\!\left(\sigma\ln m+\sigma\sqrt{\frac{t\ln m}{n}}\right)\), and after the corresponding warm-up period this becomes
$c_{t}\in O\!\left(\sigma\sqrt{\frac{t\log m}{n}}\right)$.
We now show that the \(\sqrt{t/n}\) dependence is unavoidable even for a fairness notion that is weaker than all three fairness notions studied in \Cref{sec:applications}.
We use the following weak bounded proportionality requirement, denoted \(\mathrm{bPROP}(c)\): all agent valuations lie in \([0,1]\), and each agent's proportionality deficit is always at most \(c\), that is,
\[
d_i^t \le c \qquad\text{for all } i\in[n]\text{ and all } t,
\]
where \(d_i^t=\Prop_i^t-v_i(P_i^t)\) is the proportionality deficit from Case~\ref{subsec:ex-propxc}.
This is a relaxation of \(\PROPax{c}\), because on bounded  instances the scale satisfies \(U_i^t\le 1\), so
\(\PROPax{c}\) implies \([d_i^t]_+\le c\,U_i^t\le c\).

We first prove that even \(\mathrm{bPROP}(c)\) cannot be maintained for more than \(O(nc^2)\) rounds; equivalently, any full-history guarantee must have \(c_t\in\Omega(\sqrt{t/n})\).
We then explain why this immediately implies the same \(\Omega(\sqrt{t/n})\) bound for
\(\PROPax{c}\) in online items and PDM, and for \(\EFax{\cdot}\) in online items. 

\subsection{A weak full-history requirement}
\label{subsec:hard-instance}
We work in the online item-allocation model of Case~\ref{subsec:ex-propxc}, restricted to \emph{bounded}
per-round values: in round \(t{+}1\) the adversary reveals an item's values
\(x^{t+1}\in[0,1]^n\) (where \(x_i^{t+1}=v_i(g^{t+1})\)), and the algorithm allocates the item to some
agent \(a^{t+1}\in[n]\).

Recall the proportional share \(\Prop_i^t=\frac1n v_i(G^t)\), utility \(u_i^t=v_i(P_i^t)\), and deficit
\(d_i^t:=\Prop_i^t-u_i^t\) from Case~\ref{subsec:ex-propxc}. This is a full-history fairness notion: all items that have arrived up to time \(t\) are counted.  On bounded instances, the weak relaxation
\(\mathrm{bPROP}(c)\) is simply the requirement \(d_i^t\le c\) for all agents \(i\) and times \(t\).

\subsection{The lower bound}
\label{subsec:lb-hard}

\paragraph{Setup.}
Fix \(c\ge1\). We measure how much slack each agent has before violating \(\mathrm{bPROP}(c)\):
\[
Z_i^t:=2c+u_i^t-\Prop_i^t .
\]
Initially,
\begin{equation}
\label{eq:Z-init-hard}
Z_i^0=2c
\qquad (i\in[n]).
\end{equation}
Moreover,
\begin{equation}
\label{eq:bprop-Z-equivalence-hard}
\mathrm{bPROP}(c)
\quad\Longleftrightarrow\quad
Z_i^t\ge c\ \text{ for all }i.
\end{equation}
Indeed, \(\mathrm{bPROP}(c)\) is equivalent to
\(\Prop_i^t-u_i^t\le c\), which is equivalent to
\(2c+u_i^t-\Prop_i^t\ge c\).

\paragraph{Adversary.}
In each time $t+1$, while the outcome remains \(c\)-fair, the adversary sends a good $g^{t+1}$ with values
\begin{equation}
\label{eq:adv-x-hard}
v_i(g^{t+1}) = x_i^{t+1}:=\frac{Z_i^t}{Z_i^t+c}
\qquad (i\in[n]).
\end{equation}
During a \(c\)-fair prefix, \(Z_i^t\ge c\), and therefore
\begin{equation}
\label{eq:x-range-hard}
\frac12\le x_i^{t+1}<1
\qquad (i\in[n]).
\end{equation}
Thus all revealed values lie in \([0,1]\). We also denote the average revealed value by
\[
\bar{x}^{t+1}:=\frac1n\sum_{i=1}^n x_i^{t+1}.
\]

Let \(w=a^{t+1}\) be the winner chosen by the algorithm in round \(t+1\). Since
\(\Prop_i^t\) increases by \(x_i^{t+1}/n\) for every agent \(i\), while
\(u_i^t\) increases by \(x_i^{t+1}\) only for the winner, the shifted slack
updates as
\begin{equation}
\label{eq:Z-update-hard}
Z_i^{t+1}
=
\begin{cases}
Z_i^t-\dfrac{x_i^{t+1}}{n}, & i\neq w,\\[2mm]
Z_w^t+\left(1-\dfrac1n\right)x_w^{t+1}, & i=w.
\end{cases}
\end{equation}

\paragraph{Proof idea.}
The proof has two parts. First, we prove three bounds that hold during every
\(c\)-fair prefix: the potential \(\Phi^t\) never increases, the total shifted
slack satisfies \(S^t<3nc\), and the average revealed value satisfies
\(\bar{x}^{t+1}\le3/4\). Second, we use these bounds to count rounds. Rounds in which the winner received a low value decrease the potential by at least fixed amount, while rounds in which the winner received a high value increase the total shifted slack. These two counting bounds imply that any
\(c\)-fair prefix has length at most \(O(nc^2)\).

\paragraph{Why the adversary is hard.}
\(Z_i^t\) measures how much slack agent \(i\) still has before
violating \(\mathrm{bPROP}(c)\). The adversary gives larger values to agents
with larger slack. Thus the algorithm faces a tradeoff: choosing low-value
winners burns potential, while choosing high-value winners makes the total slack
grow. Since both the potential and the total slack are bounded while the prefix
is \(c\)-fair, the algorithm cannot keep all agents \(c\)-fair for too many
rounds,  \Cref{fig:slack-lb} illustrates this.

\begin{figure}[ht]
\centering
\begin{tikzpicture}[
  x=.055\linewidth,
  y=1cm,
  >=Latex,
  font=\small,
  bg/.style={draw=black!10, fill=black!1, rounded corners=8pt},
  panel/.style={draw=black!28, fill=white, rounded corners=5pt, line width=.5pt},
  tank/.style={draw=black!45, fill=white, rounded corners=2pt, line width=.45pt},
  arrow/.style={-{Latex[length=2.1mm]}, line width=.85pt, draw=black!58},
  faint/.style={-{Latex[length=1.7mm]}, line width=.55pt, draw=black!35},
  drain/.style={-{Latex[length=1.8mm]}, line width=.75pt, draw=red!65!black},
  gain/.style={-{Latex[length=1.8mm]}, line width=.75pt, draw=green!45!black},
  tiny/.style={font=\scriptsize, text=black!72},
  head/.style={font=\bfseries\small, text=black!88},
  badge/.style={draw=black!25, fill=white, rounded corners=2pt, inner sep=1.7pt, font=\scriptsize}
]

\draw[bg] (0,0) rectangle (17.9,6.75);

\node[head] at (2.35,6.35) {slack vector};
\node[head] at (6.75,6.35) {adversary prices slack};
\node[head] at (11.10,6.35) {one update};
\node[head] at (15.75,6.35) {after repeats};

\draw[panel, fill=blue!2] (.45,1.05) rectangle (4.25,5.95);
\node[badge, fill=blue!8] at (2.35,5.58) {$Z^t$};

\draw[dashed, draw=red!65!black, line width=.6pt] (.90,2.25) -- (3.80,2.25);
\node[tiny, text=red!65!black, anchor=west] at (3.87,2.25) {$c$};

\foreach \x/\h/\lab in {1.00/2.15/1,1.95/1.45/2,2.90/2.65/3}{
  \draw[tank] (\x,1.65) rectangle +(0.55,3.20);
  \fill[blue!34] (\x,1.65) rectangle +(0.55,\h);
  \draw[draw=blue!55!black, line width=.32pt] (\x,1.65) rectangle +(0.55,\h);
  \node[tiny] at (\x+.275,1.38) {$Z_{\lab}^t$};
}

\node[tiny, align=center] at (2.35,.40) {fair while all slacks stay above $c$};

\draw[arrow] (4.25,3.48) -- (5.05,3.48);

\draw[panel, fill=orange!3] (5.05,1.05) rectangle (8.55,5.95);
\node[badge, fill=orange!12] at (6.80,5.58) {$x_i=\frac{Z_i^t}{Z_i^t+c}$};

\draw[draw=black!35] (5.70,2.00) -- (7.95,2.00);

\foreach \x/\h/\lab/\val in {5.95/1.55/1/.68,6.70/1.05/2/.59,7.45/1.95/3/.73}{
  \draw[fill=orange!35, draw=orange!65!black, line width=.35pt] (\x,2.00) rectangle +(0.42,\h);
  \node[tiny] at (\x+.21,1.72) {$x_{\lab}$};
  \node[tiny] at (\x+.21,2.18+\h) {$\val$};
}

\node[tiny, align=center] at (6.80,4.70) {more slack\\larger value};
\draw[faint] (6.00,4.25) -- (7.35,4.25);
\draw[faint] (7.35,4.25) -- (7.95,4.55);

\draw[arrow] (8.55,3.48) -- (9.25,3.48);

\draw[panel, fill=green!2] (9.25,1.05) rectangle (13.05,5.95);
\node[badge, fill=green!10] at (11.15,5.58) {algorithm picks $w=3$};

\draw[dashed, draw=red!65!black, line width=.6pt] (9.70,2.25) -- (12.60,2.25);
\node[tiny, text=red!65!black, anchor=west] at (12.67,2.25) {$c$};

\foreach \x/\h/\lab/\col in {9.85/1.88/1/blue,10.80/1.24/2/blue,11.75/3.02/3/green}{
  \draw[tank] (\x,1.65) rectangle +(0.55,3.20);
  \fill[\col!34] (\x,1.65) rectangle +(0.55,\h);
  \draw[draw=\col!55!black, line width=.32pt] (\x,1.65) rectangle +(0.55,\h);
  \node[tiny] at (\x+.275,1.38) {$Z_{\lab}^{t+1}$};
}

\draw[drain] (10.12,5.03) -- (10.12,4.63);
\node[tiny, text=red!65!black] at (9.85,5.20) {$-$};

\draw[drain] (11.07,5.03) -- (11.07,4.63);
\node[tiny, text=red!65!black] at (10.80,5.20) {$-$};

\draw[gain] (12.03,4.55) -- (12.03,5.05);
\node[tiny, text=green!38!black] at (12.38,5.05) {$+$};

\node[tiny, align=center] at (11.15,.40) {non-winners drain; winner rises};

\draw[arrow] (13.05,3.48) -- (13.75,3.48);

\draw[panel, fill=red!2] (13.75,1.05) rectangle (17.45,5.95);
\node[badge, fill=red!7] at (15.60,5.58) {eventually};

\draw[dashed, draw=red!65!black, line width=.6pt] (14.15,2.25) -- (17.02,2.25);
\node[tiny, text=red!65!black, anchor=west] at (17.08,2.25) {$c$};

\foreach \x/\h/\lab/\col in {14.30/1.55/1/blue,15.23/.42/2/red,16.16/2.35/3/blue}{
  \draw[tank] (\x,1.65) rectangle +(0.55,3.20);
  \fill[\col!34] (\x,1.65) rectangle +(0.55,\h);
  \draw[draw=\col!55!black, line width=.32pt] (\x,1.65) rectangle +(0.55,\h);
  \node[tiny] at (\x+.275,1.38) {$Z_{\lab}$};
}

\node[badge, fill=red!9, text=red!65!black] at (15.50,3.15) {violation};
\draw[drain] (15.50,2.95) -- (15.50,2.15);

\node[tiny, align=center] at (15.60,.40) {some slack falls below $c$};

\draw[arrow] (11.15,1.05) .. controls (9.4,.28) and (4.9,.28) .. (2.35,1.05);
\node[tiny, fill=black!1, inner sep=1.2pt] at (6.75,.40) {repeat the same adaptive rule};

\end{tikzpicture}
\caption{Lower-bound intuition as a draining process. \(Z_i^t\) is the slack agent \(i\) has before violating \(\mathrm{bPROP}(c)\). The adversary converts slack into values \(x_i=Z_i^t/(Z_i^t+c)\), so agents with more slack look more valuable. The winner's slack rises, but non-winners lose slack; repeated adaptive rounds eventually push some slack below \(c\).}
\Description{A diagram showing how agents' slack changes over time. The chosen agent gains slack, the other agents lose slack, and repeated rounds can make one agent fall below the fairness threshold.}
\label{fig:slack-lb}
\end{figure}

We use two elementary claims in the lower-bound proof.
The first gives a quadratic lower bound on a logarithmic gap, and the second says that concavity plus cancellation of the linear terms yields an actual decrease in the potential.

\begin{claimrep}[A useful logarithmic inequality]
\label{claim:log-gap-hard}
For every \(0\le s\le1\),
\[
s-\ln(1+s)\ge \frac{s^2}{4}.
\]
\end{claimrep}
\begin{proofsketch}
Consider \(g(s):=s-\ln(1+s)-s^2/4\). Since \(g(0)=0\) and \(g'(s)\ge0\) on \([0,1]\), the claim follows.
\end{proofsketch}
\begin{proof}
Let
\[
g(s):=s-\ln(1+s)-\frac{s^2}{4}.
\]
Then \(g(0)=0\). For \(0\le s\le1\),
\[
g'(s)
=
1-\frac1{1+s}-\frac{s}{2}
=
\frac{s(1-s)}{2(1+s)}
\ge0.
\]
Thus \(g(s)\ge0\) on \([0,1]\), proving the claim.
\end{proof}

\begin{claimrep}[Concavity gap]
\label{claim:concavity-gap-hard}
Let \(\psi\) be concave and differentiable.
Let $(z_i)_{i=1}^n$ and 
$(\Delta_i)_{i=1}^n$ be real numbers satisfying \(z_i>0\),
\(z_i+\Delta_i>0\), and
\[
\sum_{i=1}^n \psi'(z_i)\Delta_i=0.
\]
Define

%

\[
G_i:=\psi(z_i)-\psi(z_i+\Delta_i)+\psi'(z_i)\Delta_i .
\]
Then \(G_i\ge0\) for every \(i\), and
\[
\sum_{i=1}^n\psi(z_i)-\sum_{i=1}^n\psi(z_i+\Delta_i)
=
\sum_{i=1}^n G_i.
\]
In particular,
\[
\sum_{i=1}^n\psi(z_i+\Delta_i)
\le
\sum_{i=1}^n\psi(z_i).
\]
\end{claimrep}
\begin{proofsketch}
This is exactly the first-order concavity inequality, summed over coordinates. The assumed cancellation of the linear terms leaves only the nonnegative concavity gaps.
\end{proofsketch}
\begin{proof}
Concavity gives
\[
\psi(z_i+\Delta_i)-\psi(z_i)\le \psi'(z_i)\Delta_i,
\]
so \(G_i\ge0\). Also, since \(\sum_i\psi'(z_i)\Delta_i=0\),
\[
\sum_i G_i
=
\sum_i\psi(z_i)-\sum_i\psi(z_i+\Delta_i)
+
\sum_i\psi'(z_i)\Delta_i
=
\sum_i\psi(z_i)-\sum_i\psi(z_i+\Delta_i).
\]
The final inequality follows from \(G_i\ge0\).
\end{proof}
With these two elementary claims in hand, we can prove the lower bound.
The logarithmic inequality quantifies how much potential is lost in certain rounds, while the concavity-gap claim shows that the adversary's updates never increase the potential during a \(c\)-fair prefix.
\begin{theoremrep}[Lower bound]
\label{thm:lb-framework}
\label{thm:lb-rounds}
For every \(n\ge2\) and \(c\ge1\), every online algorithm, deterministic or
randomized, can be forced by an adaptive adversary so that no \(c\)-fair prefix
has length at least \(4900nc^2\).%
\footnote{
The constant $4900$ can be substantially reduced, but the proof would be more complicated. 
}

Moreover, for every \(T\geq n\), even if the horizon $T$ is known to the algorithm, any
worst-case prefix-wise guarantee must have
\[
c_T\in\Omega\!\left(\sqrt{\frac{T}{n}}\right).
\]
\end{theoremrep}
\begin{proofsketch}
We apply the adversary and shifted slack variables defined above. The formal proof turns the preceding proof idea
into two counting inequalities. First, the logarithmic-potential argument bounds the number of rounds in which
the chosen winner has relatively low value. Second, the bound on the total shifted slack bounds the number of
remaining rounds. Both bounds are \(O(nc^2)\), and the explicit constants give the stated \(4900nc^2\) prefix-length
bound.

To translate this into a time-indexed lower bound for \(t\ge n\), the construction
above gives the \(\Omega(\sqrt{t/n})\) bound, using zero-valued padding when the
violation occurs before the target horizon. The argument is pathwise, so it also
applies to randomized algorithms against an adaptive adversary.
\end{proofsketch}
\begin{proof}
For the proof, we use \emph{two} potential functions, $\Phi$ and $S$, defined by:
\[
\Phi^t:=\sum_{i=1}^n\psi(Z_i^t)
\qquad
\text{ where }
\psi(z):=z+c\ln z,
\qquad
S^t:=\sum_{i=1}^n Z_i^t .
\]

Note that this definition of $\Phi^t$ is unrelated to the definition of $\Phi^t$ in the proofs of Lemmas~\ref{lem:transform-general} and \ref{lem:potential-bound-general}.

Fix any online algorithm and run the adversary above until \(c\)-fairness fails.
The argument is pathwise: for every realized sequence of choices made by the
algorithm, the adaptive adversary defined above forces the same bound. Therefore
the lower bound also applies to randomized algorithms against an adaptive
adversary. We analyze an
arbitrary \(c\)-fair prefix.

We first record a one-round calculation that will be used several times. Fix any
round \(t+1\) in this prefix, and write
\[
z_i:=Z_i^t,\qquad
x_i:=x_i^{t+1},\qquad
\Delta_i:=Z_i^{t+1}-Z_i^t .
\]
Since the whole prefix is \(c\)-fair, both the state before and the state after
this round are \(c\)-fair. Hence, by \eqref{eq:bprop-Z-equivalence-hard},
\[
z_i=Z_i^t\ge c>0,
\qquad
z_i+\Delta_i=Z_i^{t+1}\ge c>0.
\]
Thus the logarithmic potential $\Phi^t$ is well-defined at all points used below.

By the adversary's definition,
\[
x_i=\frac{z_i}{z_i+c}.
\]
Also,
\[
\psi'(z)=1+\frac{c}{z}=\frac{z+c}{z},
\qquad
\psi''(z)=-\frac{c}{z^2}<0.
\]
Thus \(\psi\) is concave, and
\[
\psi'(z_i)=\frac{z_i+c}{z_i}=\frac1{x_i}.
\]
By the update rule \eqref{eq:Z-update-hard},
\[
\Delta_i=-\frac{x_i}{n}\quad (i\neq w),
\qquad
\Delta_w=\left(1-\frac1n\right)x_w .
\]
Therefore the first-order terms cancel:
\begin{align}
\sum_{i=1}^n\psi'(z_i)\Delta_i
&=
\sum_{i\neq w}
\frac1{x_i}\left(-\frac{x_i}{n}\right)
+
\frac1{x_w}\left(\left(1-\frac1n\right)x_w\right) =
-\frac{n-1}{n}
+
\frac{n-1}{n}
=
0.
\label{eq:linear-cancel-hard}
\end{align}

\paragraph{\textbf{Step 1: the potential and the average value are bounded.}}
By \(\psi''(z)=-c/z^2<0\), the function \(\psi\) is concave. Applying
Claim~\ref{claim:concavity-gap-hard} together with the cancellation identity
\eqref{eq:linear-cancel-hard} gives
\[
\Phi^{t+1}\le\Phi^t.
\]
Therefore the potential never increases during the prefix.

Since \(Z_i^0=2c\), we have \(\Phi^0=n\psi(2c)\),
so 
\[
\Phi^t\le \Phi^0=n\psi(2c)
\]
throughout the prefix.

Also, during a
\(c\)-fair prefix, \(Z_i^t\ge c\), and \(\psi\) is increasing because
\(\psi'(z)=1+c/z>0\). Hence \(\Phi^t\ge n\psi(c)\). Therefore the total possible
decrease of the potential during any \(c\)-fair prefix is at most
\begin{equation}
\label{eq:potential-budget-hard}
\Phi^0-n\psi(c)
=
n(\psi(2c)-\psi(c))
=
nc(1+\ln2)
<
2nc.
\end{equation}

Also, \(c\)-fairness gives \(Z_i^t\ge c\), and therefore
\[
\psi(Z_i^t)
=
Z_i^t+c\ln Z_i^t
\ge
Z_i^t+c\ln c.
\]
Hence
\begin{align}
S^t+nc\ln c
&\le
\sum_{i=1}^n\psi(Z_i^t)
=
\Phi^t \le
n\psi(2c)
=
2nc+nc\ln(2c).
\end{align}
Thus
\begin{equation}
\label{eq:S-bound-hard}
S^t
\le
nc(2+\ln2)
<
3nc.
\end{equation}

Now define
\[
\chi(z):=\frac{z}{z+c}.
\]
Then \(x_i^{t+1}=\chi(Z_i^t)\). Moreover,
\[
\chi'(z)=\frac{c}{(z+c)^2}>0,
\qquad
\chi''(z)=-\frac{2c}{(z+c)^3}<0.
\]
Thus \(\chi\) is increasing and concave. Jensen's inequality for the concave
function \(\chi\) gives
\[
\frac1n\sum_{i=1}^n\chi(Z_i^t)
\le
\chi\left(\frac1n\sum_{i=1}^n Z_i^t\right).
\]
Using this, the definition of \(S^t\), the total-slack bound
\eqref{eq:S-bound-hard}, and the fact that \(\chi\) is increasing, we get
\begin{equation}
\label{eq:xbar-bound-hard}
\bar{x}^{t+1}
=
\frac1n\sum_{i=1}^n \chi(Z_i^t)
\le
\chi\left(\frac{S^t}{n}\right)
<
\chi(3c)
=
\frac34.
\end{equation}

\Cref{fig:lower-bound-two-meters} summarizes the counting argument used next:
low-value winner rounds are charged to the logarithmic-potential meter \(\Phi\),
whereas high-value winner rounds are charged to the total-slack meter \(S\).

\begin{figure}[ht]
\centering
\begin{adjustbox}{max width=\linewidth, max totalheight=.64\textheight, center}
\begin{tikzpicture}[
  >=Latex,
  font=\small,
  outer/.style={draw=black!10, fill=black!1, rounded corners=8pt},
  round/.style={circle, draw=black!28, fill=white, minimum size=4.2mm, inner sep=0pt},
  gate/.style={draw=black!18, fill=black!2, rounded corners=8pt, line width=.8pt},
  caseL/.style={draw=green!45!black, fill=green!7, rounded corners=6pt,
    line width=1pt, inner sep=5pt, align=center},
  caseH/.style={draw=orange!80!black, fill=orange!8, rounded corners=6pt,
    line width=1pt, inner sep=5pt, align=center},
  Lbox/.style={draw=green!45!black, fill=green!6, rounded corners=6pt,
    line width=.8pt, inner sep=5pt, align=center},
  Hbox/.style={draw=orange!80!black, fill=orange!7, rounded corners=6pt,
    line width=.8pt, inner sep=5pt, align=center},
  meter/.style={draw=black!22, fill=white, rounded corners=5pt, line width=.8pt},
  capstyle/.style={draw=black!20, fill=black!3, rounded corners=4pt},
  small/.style={font=\scriptsize, align=center},
  bound/.style={draw=black!16, fill=black!2, rounded corners=5pt,
    inner sep=4pt, align=center},
  final/.style={draw=purple!65!black, fill=purple!5, rounded corners=6pt,
    inner sep=5pt, align=center},
  arrow/.style={->, black!45, line width=.9pt},
  arrowL/.style={->, green!45!black, line width=1pt},
  arrowH/.style={->, orange!80!black, line width=1pt}
]

\draw[outer] (0,-.20) rectangle (15.8,7.45);

\node[small, text=black!65] at (7.9,7.00)
  {all rounds in the still-\(c\)-fair prefix};

\draw[black!25, line width=.7pt] (2.15,6.62) -- (13.65,6.62);

\foreach \x in {2.30,2.90,3.50,4.10,4.70,5.30,5.90,6.50,7.10,7.70,8.30,8.90,9.50,10.10,10.70,11.30,11.90,12.50,13.10,13.50} {
  \node[round] at (\x,6.62) {};
}

\node[small, text=black!55] at (2.30,6.27) {\(1\)};
\node[small, text=black!55] at (13.50,6.27) {\(T_{\mathrm{fair}}\)};

\node[gate, minimum width=12.0cm, minimum height=1.55cm] (gate) at (7.9,5.35) {};
\node[font=\scriptsize\bfseries, text=black!65] at (7.9,5.93)
  {classify each round by exactly one case};

\node[caseL, text width=4.55cm] (caseLeft) at (4.85,5.22) {%
  \(\boldsymbol{x_w^{t+1}\le 4/5}\)\\[-1pt]
  \scriptsize send round \(t\) to \(L\)
};

\node[caseH, text width=4.55cm] (caseRight) at (10.95,5.22) {%
  \(\boldsymbol{x_w^{t+1}>4/5}\)\\[-1pt]
  \scriptsize send round \(t\) to \(H\)
};

\draw[black!35, line width=1.2pt] (7.9,4.72) -- (7.9,5.62);

\node[
  circle,
  draw=black!32,
  fill=white,
  line width=.8pt,
  inner sep=2.2pt,
  font=\scriptsize\bfseries,
  text=black!70
] at (7.9,5.20) {OR};

\draw[arrow] (7.9,6.40) -- (7.9,6.13);

\node[Lbox, text width=4.25cm] (L) at (3.65,3.85)
  {\(\mathbf{L}\): low-value winner rounds};

\node[Hbox, text width=4.25cm] (H) at (12.15,3.85)
  {\(\mathbf{H}\): high-value winner rounds};

\draw[arrowL] (caseLeft.south) -- (L.north);
\draw[arrowH] (caseRight.south) -- (H.north);

\draw[meter] (2.50,1.78) rectangle (4.80,3.05);
\draw[capstyle] (2.95,3.05) rectangle (4.35,3.22);

\foreach \y in {2.02,2.26,2.50,2.74} {
  \draw[black!12] (2.68,\y) -- (4.62,\y);
}

\fill[green!22] (2.68,1.92) rectangle (4.62,2.55);
\draw[green!45!black, line width=.9pt] (2.68,2.55) -- (4.62,2.55);

\node[font=\bfseries\large, text=green!38!black] at (3.65,2.30) {\(\Phi\)};
\node[bound, text width=3.75cm] (LB) at (3.65,1.18) {\(L<800nc^2\)};

\draw[meter] (11.00,1.78) rectangle (13.30,3.05);
\draw[capstyle] (11.45,3.05) rectangle (12.85,3.22);

\foreach \y in {2.02,2.26,2.50,2.74} {
  \draw[black!12] (11.18,\y) -- (13.12,\y);
}

\fill[orange!24] (11.18,1.92) rectangle (13.12,2.42);
\draw[orange!80!black, line width=.9pt] (11.18,2.42) -- (13.12,2.42);

\node[font=\bfseries\large, text=orange!80!black] at (12.15,2.30) {\(S\)};
\node[bound, text width=3.85cm] (HB) at (12.15,1.18) {\(H<20nc+5L\)};

\draw[arrowL] (L.south) -- (3.65,3.22);
\draw[arrowH] (H.south) -- (12.15,3.22);

\draw[->, dashed, black!38, line width=.8pt]
  (LB.east) -- node[small, fill=white, inner sep=1pt, text=black!55]
  {use \(L\)-bound} (HB.west);

\node[final, text width=5.9cm] (F) at (7.9,.28)
  {\(T_{\mathrm{fair}}=L+H<4900nc^2\)};

\draw[arrow] (LB.south) .. controls (4.25,.72) and (5.35,.48) .. (F.west);
\draw[arrow] (HB.south) .. controls (11.55,.72) and (10.45,.48) .. (F.east);

\end{tikzpicture}
\end{adjustbox}
\caption{Partition-and-charge view of the fair-prefix bound.}
\Description{Fair-prefix rounds are split into L or H, charged to Phi(logarithmic-
potential) or S, and combined.}
\label{fig:lower-bound-two-meters}
\end{figure}

\paragraph{\textbf{Step 2: low-value winner rounds are few.}}
Denote by $L$, the number of rounds in the $c$-fair prefix in which the winner receives a value at most $4/5$:
\[
L:=\#\{\text{rounds in the prefix with }x_w^{t+1}\le4/5\}.
\]
We show that
\[
L<800nc^2.
\]

Fix one low-value winner round and apply the one-round calculation above to this
round. By Claim~\ref{claim:concavity-gap-hard} and
\eqref{eq:linear-cancel-hard},
\[
\Phi^t-\Phi^{t+1}
=
\sum_{i=1}^n
\left[
\psi(z_i)-\psi(z_i+\Delta_i)+\psi'(z_i)\Delta_i
\right],
\]
and every summand is non-negative. Hence it is enough to keep only the winner's
summand.

For the winner,
\[
\Delta_w
=
\left(1-\frac1n\right)x_w
=
\frac{n-1}{n}\cdot\frac{z_w}{z_w+c}.
\]
Define
\[
s_w:=\frac{n-1}{n(z_w+c)}.
\]
Then \(\Delta_w=z_ws_w\). The winner's summand is
\[
\psi(z_w)-\psi(z_w(1+s_w))+\psi'(z_w)z_ws_w.
\]
Using \(\psi(z)=z+c\ln z\) and \(\psi'(z)=1+c/z\),
\begin{align*}
&\psi(z_w)-\psi(z_w(1+s_w))+\psi'(z_w)z_ws_w\\
&=
(z_w+c\ln z_w)
-
(z_w(1+s_w)+c\ln z_w+c\ln(1+s_w))
+
(z_ws_w+cs_w)\\
&=
c\bigl(s_w-\ln(1+s_w)\bigr).
\end{align*}
Thus
\begin{equation}
\label{eq:low-round-drop-start-hard}
\Phi^t-\Phi^{t+1}
\ge
c\bigl(s_w-\ln(1+s_w)\bigr).
\end{equation}

By assumption, the winner value is at most $4/5$:
\[
\frac{z_w}{z_w+c}=x_w\le\frac45,
\]
so \(z_w\le4c\). Hence, using \(n\ge2\),
\[
s_w
=
\frac{n-1}{n(z_w+c)}
\ge
\frac1{2(z_w+c)}
\ge
\frac1{10c}.
\]
Also, since the prefix is \(c\)-fair, \(z_w\ge c\), and therefore
\[
s_w
\le
\frac1{z_w+c}
\le
\frac1{2c}
\le
\frac12.
\]
Thus \(0\le s_w\le1\). By Claim~\ref{claim:log-gap-hard} and
\eqref{eq:low-round-drop-start-hard},
\begin{equation}
\label{eq:low-round-drop-hard}
\Phi^t-\Phi^{t+1}
\ge
c\cdot\frac{s_w^2}{4}
\ge
\frac{c}{4}\left(\frac1{10c}\right)^2
=
\frac1{400c}.
\end{equation}

By the potential budget \eqref{eq:potential-budget-hard}, the total possible
decrease of \(\Phi^t\) during a \(c\)-fair prefix is less than \(2nc\). Together
with the per-round drop in \eqref{eq:low-round-drop-hard}, this gives
\begin{equation}
\label{eq:L-bound-hard}
L\cdot \frac1{400c}<2nc,
\qquad\text{so}\qquad
L<800nc^2.
\end{equation}

\paragraph{\textbf{Step 3: high-value winner rounds are controlled by low-value winner rounds.}}
Denote by $H$, the number of rounds in the $c$-fair prefix in which the winner receives a value higher than $4/5$:
\[
H:=\#\{\text{rounds in the prefix with }x_w^{t+1}>4/5\}.
\]
We first compute the change in \(S^t\). From the explicit update rule
\eqref{eq:Z-update-hard},
\begin{align}
S^{t+1}-S^t
&=
\sum_{i\neq w}\left(-\frac{x_i^{t+1}}{n}\right)
+
\left(1-\frac1n\right)x_w^{t+1} =
-\frac1n\sum_{i=1}^n x_i^{t+1}
+
x_w^{t+1} =
x_w^{t+1}-\bar{x}^{t+1}.
\label{eq:S-change-hard}
\end{align}

In a high-value winner round \(x_w^{t+1}>4/5\). Using
\eqref{eq:S-change-hard} and the average-value bound \eqref{eq:xbar-bound-hard},
\[
S^{t+1}-S^t
=
x_w^{t+1}-\bar{x}^{t+1}
>
\frac45-\frac34
=
\frac1{20}.
\]
In contrast, in a low-value winner round \(x_w^{t+1}\le4/5\). Also, by
\eqref{eq:x-range-hard}, \(x_w^{t+1}\ge1/2\). Using again
\eqref{eq:S-change-hard} and \eqref{eq:xbar-bound-hard},
\[
S^{t+1}-S^t
=
x_w^{t+1}-\bar{x}^{t+1}
\ge
\frac12-\frac34
=
-\frac14.
\]

At the beginning, \(S^0=2nc\). Therefore, after a \(c\)-fair prefix with \(H\)
high-winner rounds and \(L\) low-winner rounds,
\[
S^t
\ge
2nc+\frac{H}{20}-\frac{L}{4}.
\]
On the other hand, the total-slack bound \eqref{eq:S-bound-hard} gives
\(S^t<3nc\). Thus
\[
2nc+\frac{H}{20}-\frac{L}{4}<3nc.
\]
Rearranging,
\begin{equation}
\label{eq:H-bound-hard}
H<20nc+5L.
\end{equation}

\paragraph{\textbf{Step 4: conclusion.}}
Combining two bounds \eqref{eq:L-bound-hard} and \eqref{eq:H-bound-hard}, and using
\(c\ge1\), we get
\[
H
<
20nc+5L
<
20nc+5\cdot800nc^2
\le
4020nc^2.
\]
Therefore
\[
H+L
<
4020nc^2+800nc^2
=
4820nc^2
<
4900nc^2.
\]
Thus no \(c\)-fair prefix has length at least \(4900nc^2\).

Finally, we derive the time-indexed form for every horizon. First, there is
a simple short-horizon adversary: reveal items with value \(1\) for every agent.
If \(T<n\), after \(T\) such rounds at least one agent has received no item, and
hence has proportionality deficit \(T/n\).

If \(n\le T<4900n\), reveal \(n-1\) items with value \(1\) for every agent.
At time \(n-1\), at least one agent has received no item, so her proportionality
deficit is \((n-1)/n\). Then reveal only zero-valued items until time \(T\).
Utilities, proportional shares, and deficits do not change, so at time \(T\) the
deficit is still at least
\[
\frac{n-1}{n}\ge\frac12
>
\frac1{140}\sqrt{\frac{T}{n}},
\]
where the last inequality uses \(T<4900n\).

It remains to handle \(T\ge4900n\). Set
\[
c=\frac1{70}\sqrt{\frac{T}{n}}.
\]
Then \(c\ge1\) and \(4900nc^2=T\). By the first part of the theorem, the adaptive adversary forces a violation by some
time \(\tau\le T\). After time \(\tau\), the adversary reveals only zero-valued items
until time \(T\). This does not change utilities, proportional shares, or deficits,
so the violation persists until time \(T\). Hence, for \(T\ge4900 n\), any
prefix-wise guarantee must satisfy
\[
c_T>\frac1{70}\sqrt{\frac{T}{n}}.
\]

Combining the three cases, for every \(T\) any worst-case prefix-wise guarantee
must satisfy
\[
c_T
\in
\Omega\!\left(
\min\left\{\frac{T}{n},\sqrt{\frac{T}{n}}\right\}
\right).
\]
In particular, for \(T\ge n\),
\[
c_T\in\Omega\!\left(\sqrt{\frac{T}{n}}\right).
\]
\end{proof}
\subsection{Transferring the lower bound to our motivating models}
\label{subsec:lb-instantiations}

The hard instance of \Cref{subsec:hard-instance} targets a very weak, purely additive notion:
on bounded inputs (\([0,1]\)-valuations), each agent's proportionality deficit stays within \(c\) at every time \(t\).
We now formalize the transfer via three implication lemmas showing that, on bounded instances,
our max-normalized guarantees imply this weak bounded \(c\)-fairness.

\begin{lemmarep}[\(\PROPax{c}\) (items) \(\Rightarrow\) $\mathrm{bPROP}(c)$]
\label{lem:lb-items-propx-implies-weak}
In the online item-allocation model of Case~\ref{subsec:ex-propxc}, assume \(v_i(g)\in[0,1]\) for all agents \(i\) and items \(g\).
If the allocation is \(\PROPax{c}\) at time \(t\), then
\(
v_i(P_i^t)\ \ge\ \Prop_i^t - c
\qquad\text{for all } i\in[n].
\)
\end{lemmarep}

\begin{proofsketch}
On bounded instances \(U_i^t\le 1\), so \(d_i^t\le c\,U_i^t\) implies \(d_i^t\le c\).
\end{proofsketch}

\begin{proof}
Work in the notation of Case~\ref{subsec:ex-propxc} and fix \(t\) and \(i\).
Let \(d_i^t:=\Prop_i^t-v_i(P_i^t)\) and \(U_i^t:=\max\{v_i(g):g\in G^t\setminus P_i^t\}\).
Since \(v_i(g)\le 1\) for every item \(g\), we have \(U_i^t\le 1\).
If the allocation is \(\PROPax{c}\) at time \(t\), then \(d_i^t\le c\,U_i^t\le c\),
equivalently \(v_i(P_i^t)\ge \Prop_i^t-c\).
\end{proof}

\begin{lemmarep}[\(\PROPax{c}\) (PDM) \(\Rightarrow\) bounded \(c\)-fairness]
\label{lem:lb-pdm-propx-implies-weak}
In the public decision-making model of Case~\ref{subsec:ex-pdm}, assume \(v_i^{r}(o)\in[0,1]\) for all agents \(i\), rounds \(r\), and outcomes \(o\in C\).
If the outcome sequence is \(\PROPax{c}\) at time \(t\), then
\(
 u_i^t\ \ge\ \Prop_i^t - c
\qquad\text{for all } i\in[n].
\)
\end{lemmarep}

\begin{proofsketch}
Bounded valuations imply \(V_i^t\le 1\), so \(d_i^t\le c\,V_i^t\) implies \(d_i^t\le c\).
Moreover, the hard instance of \Cref{subsec:hard-instance} is realizable in this model by taking \(C=[n]\) and
setting \(v_i^{t}(i)=x_i^t\) and \(v_i^{t}(o)=0\) for \(o\neq i\).
\end{proofsketch}

\begin{proof}
Work in the notation of Case~\ref{subsec:ex-pdm} and fix \(t\) and \(i\).
Let \(d_i^t:=\Prop_i^t-u_i^t\) and \(V_i^t:=\max_{r\le t}M_i^r\), where \(M_i^r=\max_{o\in C}v_i^r(o)\).
Since \(v_i^{r}(o)\le 1\), we have \(M_i^r\le 1\) for every \(r\le t\), hence \(V_i^t\le 1\).
If the sequence is \(\PROPax{c}\) at time \(t\), then \(d_i^t\le c\,V_i^t\le c\),
equivalently \(u_i^t\ge \Prop_i^t-c\).

For the realization of the hard instance, take \(C=[n]\) and define valuations by
\(v_i^{t}(i)=x_i^t\in[0,1]\) and \(v_i^{t}(o)=0\) for all \(o\neq i\).
Then \(M_i^t=x_i^t\) and \(\Prop_i^t=\frac1n\sum_{r=1}^t x_i^r\).
Choosing outcome \(a^{t+1}\in[n]\) yields
\(u_i^{t+1}=u_i^t+x_i^{t+1}\) if \(a^{t+1}=i\) and \(u_i^{t+1}=u_i^t\) otherwise, so
\(
d_i^{t+1}
= d_i^t+\frac{x_i^{t+1}}{n}-\mathbf{1}[a^{t+1}=i]\cdot x_i^{t+1},
\)
matching the hard-instance recurrence for the shifted slack \(Z_i^t:=2c-d_i^t\), up to the fixed additive shift.
\end{proof}

\begin{lemmarep}[\(\EFax{c}\) (items) \(\Rightarrow\) bounded \(c\)-fairness]
\label{lem:lb-items-efx-implies-weak}
In the online item-allocation model of Case~\ref{subsec:ex-propxc}, assume \(v_i(g)\in[0,1]\) for all agents \(i\) and items \(g\).
If the allocation is \(\EFax{c}\) at time \(t\), then
\(
v_i(P_i^t)\ \ge\ \Prop_i^t - c
\qquad\text{for all } i\in[n].
\)
\end{lemmarep}

\begin{proofsketch}
For each \(i\), pick \(j\) maximizing \(v_i(P_j^t)\), so \(d_i^t\le \Envy_{i\to j}^t\).
Under \(\EFax{c}\) and boundedness, \(\Envy_{i\to j}^t\le c\,s_{(i,j)}^t\le c\).
\end{proofsketch}

\begin{proof}
Work in the notation of Case~\ref{subsec:ex-propxc} and Case~\ref{subsec:ex-efc}, and fix \(t\) and an agent \(i\).
Let \(d_i^t:=\Prop_i^t-v_i(P_i^t)\), and choose \(j\in[n]\) maximizing \(v_i(P_j^t)\).
Since \(\frac1n\sum_{k=1}^n v_i(P_k^t)=\Prop_i^t\), we have \(v_i(P_j^t)\ge \Prop_i^t\) and thus
\[
d_i^t \le v_i(P_j^t)-v_i(P_i^t)=\Envy_{i\to j}^t.
\]
If \(d_i^t\le 0\) we are done, so assume \(d_i^t>0\), which implies \(j\neq i\).
Because the allocation is \(\EFax{c}\) at time \(t\), we have \(\Envy_{i\to j}^t\le c\,s_{(i,j)}^t\),
where \(s_{(i,j)}^t=\max\{v_i(g):g\in P_j^t\}\).
Under the boundedness assumption \(v_i(g)\le 1\) for all items \(g\), we have \(s_{(i,j)}^t\le 1\).
Therefore \(d_i^t\le \Envy_{i\to j}^t\le c\), i.e., \(v_i(P_i^t)\ge \Prop_i^t-c\).
\end{proof}

Together, Lemmas~\ref{lem:lb-items-propx-implies-weak}--\ref{lem:lb-items-efx-implies-weak} show that on bounded instances,
our max-normalized notions \(\PROPax{\cdot}\) (items and PDM) and \(\EFax{\cdot}\) (items) imply the weak bounded \(c\)-fairness requirement.
Hence the lower bound proved in \Cref{subsec:lb-hard} transfers to all three motivating settings.

In \Cref{sec:detailed_compare} we prove that the envy-freeness notion of \citet{benade2018make} also implies $\mathrm{bPROP}(c)$. Hence, our negative result also implies their negative result (and it is stronger, as it holds for proportionality, a much weaker fairness notion).


The lower bound shows that, under full-history fairness, a time-independent threshold is impossible in the unrestricted fully-online model.
This motivates the limited-memory models studied next: if old rounds are counted differently, stronger perpetual guarantees may become possible.

\section{Limited-Memory Fairness: Windows and Discounting}
\label{sec:limited-memory}
In the full-history model, the deficits \(z_q^t\) are computed from the entire realized prefix \(H^t\).
The lower bound of \Cref{sec:lower_bound} shows that this full-history fairness requirement cannot, in general, be maintained with a time-independent threshold.
We therefore study limited-memory fairness notions, where older rounds count less when fairness is evaluated.
We first show that hard windows can hide repeated long-run unfairness, and then give a smoother discounted-memory alternative that fits within our framework. \Cref{fig:memory-comparison} summarizes the three ways of counting the past that we compared in this section.

\begin{figure}[ht]
\centering
\begin{tikzpicture}[
  x=.055\linewidth,
  y=1cm,
  >=Latex,
  font=\small,
  bg/.style={draw=black!10, fill=black!1, rounded corners=8pt},
  panel/.style={draw=black!25, fill=white, rounded corners=5pt, line width=.5pt},
  arrow/.style={-{Latex[length=2.1mm]}, line width=.8pt, draw=black!55},
  cut/.style={dashed, draw=red!65!black, line width=.65pt},
  tiny/.style={font=\scriptsize, text=black!70},
  rowlab/.style={font=\bfseries\small, text=black!88, anchor=east},
  badge/.style={draw=black!22, fill=white, rounded corners=2pt, inner sep=2pt, font=\scriptsize}
]

\draw[bg] (0,0) rectangle (17.8,5.5);

\node[font=\bfseries\large] at (8.9,5.95) {How much does each past round count?};
\node[tiny] at (8.9,5.62) {weights assigned at the current time \(t\)};

\node[rowlab] at (3.25,4.75) {full history};
\node[rowlab] at (3.25,3.35) {hard window};
\node[rowlab] at (3.25,1.95) {discounting};

\draw[panel, fill=blue!2]   (3.55,4.25) rectangle (15.95,5.12);
\draw[panel, fill=orange!3] (3.55,2.85) rectangle (15.95,3.72);
\draw[panel, fill=green!3]  (3.55,1.45) rectangle (15.95,2.32);

\foreach \x in {4.05,5.05,6.05,7.05,8.05,9.05,10.05,11.05,12.05,13.05}{
  \draw[fill=blue!34, draw=blue!55!black, line width=.35pt]
    (\x,4.45) rectangle +(0.58,.44);
}
\node[badge, fill=blue!7] at (14.72,4.68) {\(w_r=1\)};
\node[tiny] at (14.72,4.38) {all rounds count};

\foreach \x in {4.05,5.05,6.05,7.05,8.05,9.05}{
  \draw[fill=black!8, draw=black!22, line width=.3pt]
    (\x,3.05) rectangle +(0.58,.44);
  \draw[draw=black!26, line width=.35pt]
    (\x+.08,3.10) -- (\x+.50,3.44);
}
\foreach \x in {10.05,11.05,12.05,13.05}{
  \draw[fill=orange!42, draw=orange!65!black, line width=.35pt]
    (\x,3.05) rectangle +(0.58,.44);
}
\draw[cut] (9.75,2.88) -- (9.75,3.70);
\node[tiny, text=red!65!black] at (9.75,3.86) {cutoff};
\node[badge, fill=orange!9] at (14.72,3.28) {\(w_r=0/1\)};
\node[tiny] at (14.72,2.98) {old rounds vanish};

\foreach \x/\h/\shade in {
  4.05/.12/10,
  5.05/.16/14,
  6.05/.20/18,
  7.05/.25/24,
  8.05/.31/31,
  9.05/.37/39,
  10.05/.43/47,
  11.05/.49/55,
  12.05/.54/63,
  13.05/.58/70
}{
  \draw[fill=green!\shade, draw=green!55!black, line width=.35pt]
    (\x,1.63) rectangle +(0.58,\h);
}
\node[badge, fill=green!8] at (14.72,1.88) {\(w_r=\prod\gamma\)};
\node[tiny] at (14.72,1.58) {smooth fading};

\node[tiny] at (4.34,0.95) {old};
\node[tiny] at (8.34,0.95) {\(\cdots\)};
\node[tiny] at (10.34,0.95) {\(t-W+1\)};
\node[tiny] at (13.34,0.95) {\(t\)};
\draw[arrow] (3.95,1.16) -- (13.85,1.16);
\node[tiny] at (14.55,1.16) {recent};

\draw[draw=black!35, line width=.55pt] (13.34,1.34) -- (13.34,5.22);
\node[tiny, fill=black!1, inner sep=1pt] at (13.34,5.37) {now};

\end{tikzpicture}
\caption{Three ways of counting the past. Full-history fairness gives every past round equal weight. A hard window keeps only the most recent rounds and drops older rounds at a sharp cutoff. Discounted memory keeps every round but assigns gradually smaller weight to older rounds.}
\Description{A comparison of three memory rules over time: full-history memory gives equal weight to all past rounds, hard-window memory keeps only recent rounds, and discounted memory gives gradually smaller weight to older rounds.}
\label{fig:memory-comparison}
\end{figure}

\subsection{Hard windows can hide repeated unfairness}
We illustrate the issue using the weak bounded proportionality notion \(\mathrm{bPROP}(c)\) defined in \Cref{subsec:hard-instance}.
To keep the discussion simple, we describe windowing in the online item-allocation setting with \(n=2\) agents.
Fix a window length \(W\ge 1\).
A windowed version of \(\mathrm{bPROP}(c)\) checks proportionality only inside the most recent \(W\) rounds.
Thus, at each time, the window test ignores all items that arrived before the current window.
We write \(\mathrm{bPROP}_W(c)\) for the requirement that, in every length-\(W\) window, each agent receives at least her proportional share of the value in that window, up to an additive loss of \(c\).

The following example shows why this kind of guarantee can be misleading.
Every short window will look approximately fair, but the same agent will be slightly disadvantaged in every repetition.
Since the window forgets older rounds, these small losses never need to be repaired.

\paragraph{A numeric example.}
In \Cref{ex:window-bprop-numbers} we show that \(\mathrm{bPROP}_W(c)\) can hold forever while the \emph{full-history} proportional deficit grows without bound.
The point of the example is simple: each local window contains almost the right balance of value, but the same agent is always on the losing side of the small imbalance. This is visualized in \Cref{fig:hard-window-reset}.

\begin{example}[Hard window hides unbounded long-run unfairness]
\label{ex:window-bprop-numbers}
Take \(n=2\) and window length \(W=3\).
Assume both agents value each arriving item by its listed number (so values are identical across agents).

Consider the repeating sequence of item values:
\[
1,\ 0.3,\ 0.3,\ 1,\ 0.3,\ 0.3,\ 1,\ 0.3,\ 0.3,\ \ldots
\]
Allocate every value-\(1\) item to agent~1, and allocate every value-\(0.3\) item to agent~2.

Let us first check the windowed guarantee.
Every three consecutive rounds contain exactly one item of value \(1\) and two items of value \(0.3\).
The total value in such a window is therefore \(1.6\), so each agent's proportional share inside the window is \(0.8\).

In the same window, agent~1 receives the value-\(1\) item, while agent~2 receives the two value-\(0.3\) items, for total value \(0.6\).
Thus agent~1 is above her proportional share, and agent~2 is below her proportional share by only \(0.2\).
In particular, every length-\(3\) window satisfies \(\mathrm{bPROP}_3(1)\).

But full-history proportionality fails badly.
In each three-round block, the total value is \(1.6\), so each agent's proportional share for that block is \(0.8\).
Agent~2 receives only \(0.6\), so she falls behind by \(0.2\) in every block.

These losses accumulate.
After \(K\) blocks, at time \(t=3K\), agent~2's full-history deficit equals \(0.2K\), which grows linearly with time without bound.
\end{example}

\begin{figure}[ht]
\centering
\resizebox{0.98\linewidth}{!}{%
\begin{tikzpicture}[
  >=Stealth,
  font=\small,
  card/.style={rounded corners=6pt, draw=green!45!black, fill=green!5, line width=0.9pt},
  one/.style={rounded corners=2pt, draw=blue!60!black, fill=blue!9, minimum width=0.72cm, minimum height=0.52cm, font=\footnotesize, inner sep=1pt},
  small/.style={rounded corners=2pt, draw=orange!75!black, fill=orange!13, minimum width=0.72cm, minimum height=0.52cm, font=\footnotesize, inner sep=1pt},
  redbox/.style={rounded corners=6pt, draw=red!60!black, fill=red!4, line width=0.9pt},
  token/.style={circle, draw=red!65!black, fill=red!10, minimum size=0.58cm, inner sep=0pt, font=\scriptsize},
  arrow/.style={-{Stealth[length=2.1mm]}, thick, draw=black!45},
  redarrow/.style={-{Stealth[length=2.0mm]}, thick, draw=red!65!black},
  label/.style={font=\scriptsize, fill=white, inner sep=1.2pt, rounded corners=1pt}
]

\node[font=\bfseries] at (7.0,5.08) {Hard windows reset; full history accumulates};

\node[font=\bfseries\small, text=green!45!black, anchor=east] at (-0.15,3.62) {window view};
\node[font=\scriptsize, text=green!45!black, anchor=east, align=right] at (-0.15,3.30) {only current\\block is checked};

\node[font=\bfseries\small, text=red!65!black, anchor=east] at (-0.15,1.22) {full history};
\node[font=\scriptsize, text=red!65!black, anchor=east, align=right] at (-0.15,0.88) {losses are\\kept};

\foreach \i/\x in {1/0.55,2/5.00,3/9.45}{
  \draw[card] (\x,2.75) rectangle +(3.25,1.72);

  \node[font=\scriptsize, text=black!60] at (\x+1.625,4.22) {window \i};

  \node[one]   at (\x+0.75,3.68) {\(1\)};
  \node[small] at (\x+1.625,3.68) {\(.3\)};
  \node[small] at (\x+2.50,3.68) {\(.3\)};

  \node[font=\scriptsize] at (\x+0.75,3.27) {\(A_1\)};
  \node[font=\scriptsize] at (\x+1.625,3.27) {\(A_2\)};
  \node[font=\scriptsize] at (\x+2.50,3.27) {\(A_2\)};

  \node[font=\scriptsize, text=green!45!black] at (\x+1.625,2.96)
    {loss \(0.2\): OK};
}

\draw[arrow] (3.95,3.62) -- (4.85,3.62);
\draw[arrow] (8.40,3.62) -- (9.30,3.62);

\node[label, text=green!45!black] at (7.0,4.67)
  {same local pattern repeats};

\draw[redbox] (0.55,0.25) rectangle (13.35,2.02);

\node[font=\scriptsize, text=black!60] at (6.95,1.78)
  {each block contributes one more hidden loss};

\node[token] (t1) at (2.05,1.10) {\(+.2\)};
\node[font=\Large, text=red!65!black] at (3.45,1.10) {\(+\)};

\node[token] (t2) at (4.85,1.10) {\(+.2\)};
\node[font=\Large, text=red!65!black] at (6.25,1.10) {\(+\)};

\node[token] (t3) at (7.65,1.10) {\(+.2\)};
\node[font=\Large, text=red!65!black] at (9.05,1.10) {\(+\cdots\)};

\node[font=\Large, text=red!65!black] at (11.35,1.10) {\(=0.2K\)};

\draw[redarrow] (2.175,2.73) -- (t1.north);
\draw[redarrow] (6.625,2.73) -- (t2.north);
\draw[redarrow] (11.075,2.73) -- (t3.north);

\end{tikzpicture}%
}
\caption{Hard windows can hide repeated unfairness.
Each length-\(3\) window sees only the current block \(1,0.3,0.3\), so agent~2 is
short by only \(0.2\) inside the window and the windowed test passes.
But the same loss is repeated in every block; a full-history deficit does not
reset, so after \(K\) blocks agent~2's deficit is \(0.2K\).}
\Description{A two-row schematic. The top row shows three repeated hard-window snapshots, each containing the same block with values 1, 0.3, and 0.3 allocated to agents A1, A2, and A2. Each window passes with local loss 0.2. The bottom row shows the full-history deficit accumulating one +0.2 loss from each block, yielding 0.2K after K blocks.}
\label{fig:hard-window-reset}
\end{figure}

This is the core problem with a hard window.
The window test only sees the unfairness that remains inside the current window.
Once a deficit is older than \(W\) rounds, it disappears from the window test completely.
As a result, the mechanism may repeat the same unfair pattern again and again: every fixed window still looks fair, but the long-run deficit keeps growing.

This motivates a smoother memory model, where old rounds gradually lose weight
rather than disappearing abruptly. The key point is that discounted memory
changes how fairness is evaluated over time.
We are no longer asking whether the full history is fair; we are asking whether the discounted history is fair.

\subsection{Discounted memory in the general framework}
\label{subsec:discounted-memory-framework}
\label{subsec:apx-future-inflation-framework}

We now describe discounted memory directly in the general deficit framework of
\Cref{subsec:norm-def-framework}. The domain still determines what the tracked
requirements \(q\in Q\) mean and how the current round affects them. The memory
rule determines how much of the previously accumulated deficit remains visible.

The framework has three ingredients: a memory function, discounted deficits, and
a discounted analogue of local balance.

\begin{definition}[Memory function and squared history length]
\label{def:memory-function}
A \emph{memory function} is a sequence
\[
\gamma:\mathbb Z_{\ge0}\to(0,1].
\]
In round \(t+1\), the previously accumulated discounted deficit is multiplied by
the current memory factor \(\gamma(t)\). Thus full history corresponds to
\(\gamma(t)\equiv1\), while constant geometric discounting corresponds to
\(\gamma(t)\equiv\lambda\) for some \(\lambda\in(0,1)\).

For \(1\le r\le t\), the \emph{retained weight} of round \(r\) at time \(t\) is
\[
w_{r,t}^{\gamma}
:=
\prod_{\ell=r}^{t-1}\gamma(\ell),
\]
with the empty product equal to \(1\). Thus the current round has weight \(1\),
and older rounds have weights obtained by multiplying the memory factors that
occur after them.
\end{definition}

\begin{definition}[Discounted deficits and discounted fairness]
\label{def:discounted-deficits-fairness}
Fix a memory function \(\gamma\). 
Similarly to the full-memory case, we define a deficit for each requirement $q$ and round $t$. To emphasize that the deficit values can be different for different memory functions, we add the superscript $\gamma$:
\[
z_q^{t+1,\gamma}(\cdot):A^{t+1}\to\mathbb R_{\ge0}
\qquad(q\in Q).
\]
After the algorithm chooses \(a^{t+1}\in A^{t+1}\), the realized discounted
deficit is
\[
z_q^{t+1,\gamma}:=z_q^{t+1,\gamma}(a^{t+1}).
\]
We set \(z_q^{0,\gamma}:=0\) for all \(q\in Q\).

Similarly to the full-memory case, for a threshold \(c_t^\gamma\), the outcome is \emph{discounted \(c_t^\gamma\)-fair} at time \(t\) if
\[
z_q^{t,\gamma}\le c_t^\gamma
\qquad\text{for every }q\in Q .
\]
\end{definition}
Thus discounted fairness is measured on the discounted deficits induced by the chosen memory function \(\gamma\).

\begin{definition}[Discounted local balance conditions]
\label{def:discounted-moments}
Fix a memory function \(\gamma\). We say that an instance satisfies the
\emph{\(\gamma\)-discounted local balance conditions with parameters \(n\) and
\(\sigma^2\ge1\)} if for every round \(t+1\ge1\) there exist \(n\) not
necessarily distinct \emph{reference actions}
\[
\hat a^{t+1}_1,\ldots,\hat a^{t+1}_n\in A^{t+1}
\]
such that for every tracked requirement \(q\in Q\) there exist numbers
\[
\Delta_q^{t+1}(1),\ldots,\Delta_q^{t+1}(n)\in\mathbb R
\]
such that
\begin{equation}
\label{eq:discounted-shift}
z_q^{t+1,\gamma}(\hat a^{t+1}_k)
\le
\bigl[\gamma(t)z_q^{t,\gamma}+\Delta_q^{t+1}(k)\bigr]_+
\qquad\text{for every }k\in[n].
\end{equation}
The shifts must satisfy the same local balance inequalities as in the full-memory case --- mean bound \eqref{eq:first-moment-mean} and squared-size bound \eqref{eq:second-moment-sigma}.
\end{definition}
The only difference from the full-history model is that, before the current round contributes its one-step shift, the old deficit is contracted by the memory factor \(\gamma(t)\).

\paragraph{Potential notation for discounted deficits.}
Fix \(p\ge1\), and let \(\sigma^2\) be the squared-size parameter in the
discounted local balance conditions. We define $f$ as in the full-history analysis
\eqref{eq:f-def-general}.
For the realized discounted deficit vector at time \(t\), write
\[
\Phi^{t,\gamma}:=\sum_{q\in Q} f(z_q^{t,\gamma}),
\qquad
\Psi^{t,\gamma}:=(\Phi^{t,\gamma})^{1/p}.
\]
For a candidate action \(a\in A^{t+1}\), write
\[
\Phi^{t+1,\gamma}(a)
:=
\sum_{q\in Q} f(z_q^{t+1,\gamma}(a)),
\qquad
\Psi^{t+1,\gamma}(a):=(\Phi^{t+1,\gamma}(a))^{1/p}.
\]
Thus the \(p\)-potential rule chooses an action
\[
a^{t+1}\in\arg\min_{a\in A^{t+1}}\Psi^{t+1,\gamma}(a),
\]
equivalently minimizing \(\Phi^{t+1,\gamma}(a)\). After the action is chosen,
we write
\[
\Phi^{t+1,\gamma}:=\Phi^{t+1,\gamma}(a^{t+1}),
\qquad
\Psi^{t+1,\gamma}:=\Psi^{t+1,\gamma}(a^{t+1}).
\]

With this notation in place, the full-history potential argument extends almost
verbatim. The only change is that the raw time \(t\) is replaced by an expression that depends on the memory function $\gamma$, which we call the \emph{squared history length} at time \(t\), and denote by \(\M_\gamma(t)\):
\begin{definition}[Squared history length]
\label{def:squared-history-length}
\begin{equation}
\label{eq:memory-length-explicit}
\M_\gamma(t)
:=
\sum_{r=1}^{t}\bigl(w_{r,t}^{\gamma}\bigr)^2
=
\sum_{r=1}^{t}\prod_{\ell=r}^{t-1}\gamma(\ell)^2 .
\end{equation}
Equivalently,
\begin{equation}
\label{eq:memory-length-recursion}
\M_\gamma(0):=0,
\qquad
\M_\gamma(t+1):=1+\gamma(t)^2\M_\gamma(t).
\end{equation}
\end{definition}


The squared history length is the effective number of fully-remembered rounds
for the squared-potential analysis. Note that $M_{\gamma}(t) = t$ when $\gamma\equiv 1$. 
We return to concrete choices of
\(\gamma\), and to the resulting values of \(\M_\gamma(t)\), in
\Cref{subsec:interpreting-discounted-memory}.

\begin{lemmarep}[Potential bound under discounted memory]
\label{lem:discount-uniform-bound}
Assume the \(\gamma\)-discounted local balance conditions hold with parameters
\(n\) and \(\sigma^2\ge1\). Let \(m:=|Q|\ge2\). For the fixed \(p\ge1\) in the
potential notation above, if the \(p\)-potential rule is run on the discounted
deficits, then, for every time \(t\),
\[
\max_{q\in Q} z_q^{t,\gamma}
\le
m^{1/(2p)}
\sqrt{
4p^2\sigma^2+
\frac{2\sqrt e\,p\sigma^2}{n}\M_\gamma(t)
}.
\]
\end{lemmarep}

\begin{proofsketch}
The proof is the full-history potential argument applied after contracting the
old deficit vector by \(\gamma(t)\). In round \(t+1\), the Taylor estimate is
applied with baseline \(\gamma(t)z_q^{t,\gamma}\) instead of \(z_q^t\). This
creates the same one-step additive term as before, but the previously
accumulated potential excess is multiplied by \(\gamma(t)^2\). Unrolling these
multipliers gives exactly \(\M_\gamma(t)\).
\end{proofsketch}

\begin{proof}
Fix a round \(t+1\). Since the chosen action minimizes
\(\Phi^{t+1,\gamma}(a)\), it is no worse than the average over the reference
actions from \Cref{def:discounted-moments}. Hence
\[
(\Psi^{t+1,\gamma})^p
=
\Phi^{t+1,\gamma}
\le
\frac1n\sum_{k=1}^n\Phi^{t+1,\gamma}(\hat a_k^{t+1}).
\]
Expanding the definition of \(\Phi^{t+1,\gamma}(\cdot)\), this gives
\[
(\Psi^{t+1,\gamma})^p
\le
\sum_{q\in Q}
\frac1n\sum_{k=1}^n
f\!\left(z_q^{t+1,\gamma}(\hat a_k^{t+1})\right).
\]

By the discounted local balance condition \eqref{eq:discounted-shift} from
\Cref{def:discounted-moments},
\[
z_q^{t+1,\gamma}(\hat a_k^{t+1})
\le
\bigl[
\gamma(t)z_q^{t,\gamma}
+
\Delta_q^{t+1}(k)
\bigr]_+ ,
\]
where
\[
\sum_{k=1}^n\Delta_q^{t+1}(k)\le0,
\qquad
\sum_{k=1}^n\bigl(\Delta_q^{t+1}(k)\bigr)^2\le\sigma^2 .
\]
Therefore, for each fixed \(q\), the one-dimensional Taylor estimate
\Cref{lem:moment-taylor} applies with baseline
\(\gamma(t)z_q^{t,\gamma}\). Thus
\[
\frac1n\sum_{k=1}^n
f\!\left(z_q^{t+1,\gamma}(\hat a_k^{t+1})\right)
\le
f\!\left(\gamma(t)z_q^{t,\gamma}\right)
+
\frac{2\sqrt e\,p^2\sigma^2}{n}
\left(
(\gamma(t)z_q^{t,\gamma})^2+4p^2\sigma^2
\right)^{p-1}.
\]
Summing over \(q\in Q\), we get
\[
(\Psi^{t+1,\gamma})^p
\le
\sum_{q\in Q} f\!\left(\gamma(t)z_q^{t,\gamma}\right)
+
\frac{2\sqrt e\,p^2\sigma^2}{n}
\sum_{q\in Q}
\left(
(\gamma(t)z_q^{t,\gamma})^2+4p^2\sigma^2
\right)^{p-1}.
\]
By the Jensen power-sum bound \Cref{lem:jensen-power-sum-general},
\[
\sum_{q\in Q}
\left(
(\gamma(t)z_q^{t,\gamma})^2+4p^2\sigma^2
\right)^{p-1}
\le
m^{1/p}
\left(
\sum_{q\in Q} f\!\left(\gamma(t)z_q^{t,\gamma}\right)
\right)^{(p-1)/p}.
\]
Applying the concavity-transfer inequality
\Cref{lem:concavity-transfer-general} yields
\[
\Psi^{t+1,\gamma}
\le
\left(
\sum_{q\in Q} f\!\left(\gamma(t)z_q^{t,\gamma}\right)
\right)^{1/p}
+
\frac{2\sqrt e\,p\sigma^2}{n}m^{1/p}.
\]

We now compare the contracted potential with the previous potential. For every
\(q\in Q\),
\[
(\gamma(t)z_q^{t,\gamma})^2+4p^2\sigma^2
=
\gamma(t)^2\left((z_q^{t,\gamma})^2+4p^2\sigma^2\right)
+
\left(1-\gamma(t)^2\right)4p^2\sigma^2 .
\]
Define two nonnegative vectors \(U,V\in\mathbb R_{\ge 0}^{Q}\) by
\[
U_q
:=
\gamma(t)^2\left((z_q^{t,\gamma})^2+4p^2\sigma^2\right),
\qquad
V_q
:=
\left(1-\gamma(t)^2\right)4p^2\sigma^2 .
\]
Then, for each \(q\in Q\),
\[
U_q+V_q
=
(\gamma(t)z_q^{t,\gamma})^2+4p^2\sigma^2 .
\]
Hence, by the definitions of \(f\) and the \(\ell_p\)-norm,
\[
\left(
\sum_{q\in Q} f\!\left(\gamma(t)z_q^{t,\gamma}\right)
\right)^{1/p}
=
\left(
\sum_{q\in Q}
\left(
(\gamma(t)z_q^{t,\gamma})^2+4p^2\sigma^2
\right)^p
\right)^{1/p}
=
\|U+V\|_p .
\]
Applying Minkowski's inequality in \(\ell_p(Q)\), namely
\[
\|U+V\|_p\le \|U\|_p+\|V\|_p ,
\]
gives
\[
\left(
\sum_{q\in Q} f\!\left(\gamma(t)z_q^{t,\gamma}\right)
\right)^{1/p}
\le
\|U\|_p+\|V\|_p .
\]
We compute the two terms separately. First,
\[
\|U\|_p
=
\left(
\sum_{q\in Q}
\left[
\gamma(t)^2
\left((z_q^{t,\gamma})^2+4p^2\sigma^2\right)
\right]^p
\right)^{1/p}
=
\gamma(t)^2
\left(
\sum_{q\in Q}
\left((z_q^{t,\gamma})^2+4p^2\sigma^2\right)^p
\right)^{1/p}.
\]
By the definitions of \(f\) and \(\Psi^{t,\gamma}\), this is
\[
\|U\|_p
=
\gamma(t)^2
\left(
\sum_{q\in Q} f(z_q^{t,\gamma})
\right)^{1/p}
=
\gamma(t)^2\Psi^{t,\gamma}.
\]
Second, since \(V_q\) is constant over \(q\in Q\),
\[
\|V\|_p
=
\left(
\sum_{q\in Q}
\left[
\left(1-\gamma(t)^2\right)4p^2\sigma^2
\right]^p
\right)^{1/p}
=
\left(1-\gamma(t)^2\right)m^{1/p}4p^2\sigma^2 .
\]
Therefore
\[
\left(
\sum_{q\in Q} f\!\left(\gamma(t)z_q^{t,\gamma}\right)
\right)^{1/p}
\le
\gamma(t)^2\Psi^{t,\gamma}
+
\left(1-\gamma(t)^2\right)m^{1/p}4p^2\sigma^2 .
\]

Combining this with the previous one-step bound gives
\[
\Psi^{t+1,\gamma}
\le
\gamma(t)^2\Psi^{t,\gamma}
+
\left(1-\gamma(t)^2\right)m^{1/p}4p^2\sigma^2
+
\frac{2\sqrt e\,p\sigma^2}{n}m^{1/p}.
\]
Equivalently,
\[
\frac{\Psi^{t+1,\gamma}}{m^{1/p}}-4p^2\sigma^2
\le
\gamma(t)^2
\left(
\frac{\Psi^{t,\gamma}}{m^{1/p}}-4p^2\sigma^2
\right)
+
\frac{2\sqrt e\,p\sigma^2}{n}.
\]

Since \(z_q^{0,\gamma}=0\) for all \(q\in Q\), we have
\[
\Psi^{0,\gamma}
=
\left(
\sum_{q\in Q}
(4p^2\sigma^2)^p
\right)^{1/p}
=
m^{1/p}4p^2\sigma^2,
\]
and therefore
\[
\frac{\Psi^{0,\gamma}}{m^{1/p}}-4p^2\sigma^2=0 .
\]
Iterating the one-dimensional recursion from time \(0\) to time \(t\) gives
\[
\frac{\Psi^{t,\gamma}}{m^{1/p}}-4p^2\sigma^2
\le
\frac{2\sqrt e\,p\sigma^2}{n}
\sum_{r=1}^{t}\prod_{\ell=r}^{t-1}\gamma(\ell)^2
=
\frac{2\sqrt e\,p\sigma^2}{n}\M_\gamma(t),
\]
where the last equality is \eqref{eq:memory-length-explicit}. Hence
\[
\Psi^{t,\gamma}
\le
m^{1/p}
\left(
4p^2\sigma^2+
\frac{2\sqrt e\,p\sigma^2}{n}\M_\gamma(t)
\right).
\]

Finally, for every \(q\in Q\),
\[
\left((z_q^{t,\gamma})^2+4p^2\sigma^2\right)^p
\le
\Phi^{t,\gamma}
=
(\Psi^{t,\gamma})^p .
\]
Therefore
\[
(z_q^{t,\gamma})^2+4p^2\sigma^2
\le
\Psi^{t,\gamma}
\le
m^{1/p}
\left(
4p^2\sigma^2+
\frac{2\sqrt e\,p\sigma^2}{n}\M_\gamma(t)
\right).
\]
Dropping the nonnegative term \(4p^2\sigma^2\) on the left and taking square
roots gives
\[
z_q^{t,\gamma}
\le
m^{1/(2p)}
\sqrt{
4p^2\sigma^2+
\frac{2\sqrt e\,p\sigma^2}{n}\M_\gamma(t)
}.
\]
Taking the maximum over \(q\in Q\) proves the claim.
\end{proof}

\subsection{Interpreting the discounted-memory bound}
\label{subsec:interpreting-discounted-memory}

The discounted bound separates the potential analysis from the choice of memory
rule. Once a memory function \(\gamma\) is fixed, all dependence on time enters
through the squared history length \(\M_\gamma(t)\). In the general theorem,
with the standard choice \(p=\ln(2m)\), where \(m=|Q|\), the bound has the form
\[
O\!\left(
\sigma\ln m+
\sigma\sqrt{\frac{\M_\gamma(t)\ln m}{n}}
\right).
\]
Here \(\sigma^2\) is the squared-size parameter in the local balance condition,
and \(n\) is the number of reference actions in that condition. Thus
\(\M_\gamma(t)\) plays the same role that \(t\) played in the full-history
bound. A larger value of \(\M_\gamma(t)\) means that the fairness test retains
more of the past; a smaller value means that old rounds fade more aggressively.

The recursion
\[
\M_\gamma(t+1)=1+\gamma(t)^2\M_\gamma(t)
\]
also gives a useful way to read the guarantee. Since \(0<\gamma(t)\le1\), we
always have \(\M_\gamma(t)\le t\). Therefore discounted memory is never worse
than full history in its time dependence. The difference between memory rules is
how quickly \(\M_\gamma(t)\) grows.

\begin{example}[Full history]
If \(\gamma(t)\equiv1\), then no past round is discounted. Every round keeps
weight \(1\), and therefore
\[
\M_\gamma(t)=t.
\]
The general guarantee becomes
\[
O\!\left(
\sigma\ln m+
\sigma\sqrt{\frac{t\ln m}{n}}
\right).
\]
This is the most history-sensitive notion in this family, because every past
round remains fully relevant.
\end{example}

\begin{example}[Geometric memory]
At the other extreme, suppose \(\gamma(t)\equiv\lambda\) for a fixed
\(\lambda\in(0,1)\). Then a round from \(j\) steps ago has retained weight
\(\lambda^j\), so
\[
\M_\gamma(t)
=
\sum_{j=0}^{t-1}\lambda^{2j}
=
\frac{1-\lambda^{2t}}{1-\lambda^2}
\le
\frac1{1-\lambda^2}.
\]
Thus \(\M_\gamma(t)\) is bounded uniformly in \(t\), and the general guarantee is
time-uniform:
\[
O\!\left(
\sigma\ln m+
\sigma\sqrt{\frac{\ln m}{n(1-\lambda^2)}}
\right).
\]
For instance,
\[
\lambda=\frac12
\quad\Longrightarrow\quad
\M_\gamma(t)\le \frac{4}{3},
\]
whereas
\[
\lambda=0.9
\quad\Longrightarrow\quad
\M_\gamma(t)\le \frac{1}{1-0.9^2}\approx5.27 .
\]
Larger \(\lambda\) means slower forgetting. The fairness test then retains more
of the past, and the bound correspondingly becomes larger.

More generally, if there exist \(t_0\ge1\) and \(\lambda<1\) such that
\(\gamma(t)\le\lambda\) for every \(t\ge t_0\), then
\[
\M_\gamma(t)
\le
\max\left\{
\M_\gamma(t_0),\frac{1}{1-\lambda^2}
\right\}
\qquad (t\ge t_0).
\]
Hence the resulting fairness bound is uniform in time, with a constant that
also depends on the finite transient up to \(t_0\).
\end{example}

The recursion can also be read in reverse. If we want a target squared history
length \(L_t\), then choosing
\[
\gamma(t)^2=\frac{L_{t+1}-1}{L_t}
\]
makes the recursion produce \(\M_\gamma(t)=L_t\), provided the right-hand side
lies in \((0,1]\).

\begin{example}[Polynomial remembered history]
The memory rule can be chosen so that the effective squared history length grows
sublinearly. Fix \(\alpha\in(0,1)\), and suppose we want
\[
\M_\gamma(t)=t^\alpha
\qquad(t\ge1).
\]
Using the recursion, this is obtained by choosing
\[
\gamma(t)^2
=
\frac{(t+1)^\alpha-1}{t^\alpha}
\qquad(t\ge1).
\]
Then the general guarantee becomes
\[
O\!\left(
\sigma\ln m+
\sigma\sqrt{\frac{t^\alpha\ln m}{n}}
\right).
\]
For example, \(\alpha=1/2\) gives \(\M_\gamma(t)=\sqrt t\), and hence
\[
O\!\left(
\sigma\ln m+
\sigma\sqrt{\frac{\sqrt t\,\ln m}{n}}
\right).
\]
This rule retains more and more history, but much more slowly than full history.
\end{example}

\begin{example}[Logarithmic remembered history]
As a slower-growing example, suppose we want
\[
\M_\gamma(t)=1+\ln t
\qquad(t\ge1).
\]
Since \(\M_\gamma(1)=1\), it is enough to choose
\[
\gamma(t)^2
=
\frac{\M_\gamma(t+1)-1}{\M_\gamma(t)}
=
\frac{\ln(t+1)}{1+\ln t}
\qquad(t\ge1).
\]
Then the recursion gives \(\M_\gamma(t)=1+\ln t\), and the general guarantee
becomes
\[
O\!\left(
\sigma\ln m+
\sigma\sqrt{\frac{\ln t\cdot\ln m}{n}}
\right).
\]
Here the amount of retained history grows with time, but only logarithmically.
\end{example}

In summary, discounted memory gives a spectrum of memory rules:
\begin{center}
\renewcommand{\arraystretch}{1.55}
\setlength{\tabcolsep}{2pt}
\begin{tabular}{@{}p{0.18\linewidth}p{0.25\linewidth}p{0.22\linewidth}p{0.26\linewidth}@{}}
\hline
\textbf{Rule}
&
\(\boldsymbol{\gamma(t)^2}\)
&
\(\boldsymbol{\M_\gamma(t)}\)
&
\(\boldsymbol{\text{Extra term}}\)
\\[3pt]
\hline
Full history
&
\(1\)
&
\(t\)
&
\(O\!\left(\sigma\sqrt{t\ln m/n}\right)\)
\\[7pt]
Polynomial
&
\(\displaystyle
\frac{(t+1)^\alpha-1}{t^\alpha}
\)
&
\(t^\alpha\), fixed \(\alpha\in(0,1)\)
&
\(O\!\left(\sigma\sqrt{t^\alpha\ln m/n}\right)\)
\\[9pt]
Logarithmic
&
\(\displaystyle
\frac{\ln(t+1)}{1+\ln t}
\)
&
\(1+\ln t\)
&
\(O\!\left(\sigma\sqrt{\ln t\cdot\ln m/n}\right)\)
\\[9pt]
Geometric
&
\(\lambda^2\), fixed \(\lambda\in(0,1)\)
&
\(\displaystyle
\le\frac{1}{1-\lambda^2}
\)
&
\(O\!\left(\sigma\sqrt{\ln m/(n(1-\lambda^2))}\right)\)
\\[7pt]
\hline
\end{tabular}
\end{center}
The additive \(\sigma\ln m\) term is present in all rows. Slower forgetting gives
a more history-sensitive fairness notion but a larger time-dependent term.
Faster forgetting gives a more local fairness notion but a smaller, and
sometimes time-uniform, guarantee.

For discounted proportional item allocation, the parameters specialize to
\(m=n\) and \(\sigma^2=1\).

\section{Conclusions and Future Work}
\label{sec:conclusion}

We introduced a generic framework for perpetual online fairness and, in the full-history model, a fully online rule that guarantees \(z_q^t \le c_{t}\) with \(c_{t}=\tilde O\!\left(\sigma\sqrt{t/n}\right)\).
We instantiated the framework for online indivisible-goods allocation and online public decisions, obtaining perpetual \(\PROPax{c}\) and \(\EFax{c}\) guarantees.
For classical \(\EFc\), we obtained a separate tight characterization.
Without restricting the possible item values in advance, the exact worst-case parameter at
horizon \(T\) is \(\lceil T/n\rceil\) against an adaptive adversary. When a
finite value set of size \(L\) is known in advance, our threshold
instantiation instead gives
\[
O\!\left(
\ln(nL)+\sqrt{\frac{t\ln(nL)}{n}}
\right).
\]

Taken together, our lower bounds clarify the price of full-history fairness.
When the whole past counts forever, some deterioration with time is unavoidable: even a substantially weaker bounded proportionality cannot generally be maintained with a constant threshold.
Accordingly, the time-growing guarantees in this paper should be read as near-optimal consequences of the fully-online model with full-history, rather than as a weakness of the potential-based analysis.

The lower bound also motivates changing how fairness counts the past. Our discounted-memory result shows that the same potential-rule framework then yields a threshold controlled by the squared history length \(\M_\gamma(t)\). In particular, the time dependence is never worse than the full-history \(\sqrt t\)-type dependence; it can be logarithmic for intermediate memory schedules, and it becomes time-uniform when \(\M_\gamma(t)\) is bounded.


\paragraph{Additional future directions.}
The general viewpoint that we propose suggests several natural extensions and refinements:
\begin{itemize}[leftmargin=1.5em]
\item \emph{Constants and parameter dependence.}
The square-root dependence on \(t\) is tight in the full-history worst case, but there is room to sharpen the dependence on $m=|{Q}|$ and on the number $n$ of reference actions through more refined potentials, adaptive choices of the potential parameter, or additional structure in the local balance conditions.

\item \emph{Stronger fairness objectives.}
We focused on \(c_t\)-fairness, namely the requirement that there are no \(c_t\)-violated requirements.
In applications one may want guarantees that are group-aware, weighted, or
prioritized (e.g., different ${q}$ carry different importance).
One possible approach is to modify the potential, or to duplicate tracked requirements with weights, while preserving the same local-balance proof strategy.

\item \emph{Beyond worst-case adversaries.}
Our guarantees and lower bounds are proved against a fully adaptive adversary. In many settings, inputs are structured or follow a stochastic model, in which one might hope for substantially lower perpetual fairness guarantees. A natural direction is thus to quantify how $c_{t}$ might improve as the adversary is weakened.

\item \emph{Limited recourse.}
While our main model is fully irrevocable, many practical systems allow small
amounts of controlled recourse (e.g., occasional swaps or bounded reassignments).
It would be interesting to extend the framework so that the action set $A^{t+1}$
includes limited recourse moves, and to quantify how even modest recourse improves perpetual fairness guarantees (without harming the worst-case guarantee).

\item \emph{Learning-augmented variants.}
In some settings, the input sequence is partially predictable.
A promising direction is to combine our framework with predictions while preserving the worst-case guarantees.

\end{itemize}

Overall, our results suggest that memory is a central modeling choice in perpetual fairness: the question is not only which fairness requirement to impose, but also how the system should count past unfairness.

\GenAIDisclosureInPaper

\begin{acks}
This research is partly funded by the Israel Science Foundation grant 1092/24.

\textbf{Generative AI usage disclosure.} \GenAIDisclosureText
\end{acks}

\printbibliography

\newpage
\appendix


\section{How our lower bound implies an envy lower bound in the normalized model of \texorpdfstring{\citet{benade2018make}}{Benade et al. (2018)}}
\label{sec:detailed_compare}

We prove that our negative result (\Cref{thm:lb-framework}) implies a matching square-root lower bound on envy in the
online fair division model of \citet{benade2018make}.
Note that \citet{benade2018make} assume that valuations are normalized to $[0,1]$.
Similarly, during the active lower-bound phase, the adversarial inputs used in our hard instance are vectors
\(x^t\in[1/2,1)^n\), and the later padding phase uses the all-zero vector. Hence all item values lie in \([0,1]\).
We first recall the envy definitions used by \citet{benade2018make}.

\paragraph{The normalized envy model of \citet{benade2018make}.}
There are $n$ agents and items arrive online, one per round.
At round $t$ an item arrives with values $(v_1(g^{t}),\ldots,v_n(g^{t}))\in[0,1]^n$
and must be allocated irrevocably to a single agent.
Let $P_i^t$ denote agent $i$'s bundle after $t$ rounds.

\begin{definition}
(a) The \emph{amount of envy agent $i$ has towards agent $j$ at time $t$}, denoted $Envy^{t}_{i,j}$, is defined as
\[
Envy^{t}_{i,j} := \max\{v_i(P^t_j) - v_i(P^t_i),\,0\}.
\]

(b) The \emph{amount of envy at time $t$}, denoted $Envy^{t}$, is defined as
\[
Envy^{t} := \max_{i,j \in [n]} Envy^{t}_{i,j}.
\]
\end{definition}

Now, fix a time $t$ and let $G^t:=P_1^t\cup\cdots\cup P_n^t$ be the set of arrived items.
Recall our definition from Application 1: 
$\Prop_i^t := \frac{1}{n}v_i(G^t)$.

Define the proportionality deficit
\[
d_i^t := \Prop_i^t-v_i(P_i^t).
\]

\begin{lemma}[Envy from proportionality deficit]
\label[lemma]{lem:envy-from-deficit-framework}
For any time $t$ and agent $i$,
\[
\max_{j\in[n]} Envy_{i,j}^t
\ \ge\
d_i^t.
\]
\end{lemma}

\begin{proof}
Since $\frac{1}{n}\sum_{j\in[n]} v_i(P^t_j)=\Prop_i^t$, there exists a $j$ with
$v_i(P^t_j)\ge \Prop_i^t$. Hence
\[
\max_j\{v_i(P_j^t)-v_i(P_i^t)\}
\ge
\Prop_i^t-v_i(P_i^t)
=
d_i^t.
\]
If \(d_i^t\le 0\), then the claim is trivial since envy is nonnegative. If \(d_i^t>0\), the same inequality implies
\[
\max_{j\in[n]} Envy_{i,j}^t\ge d_i^t.
\]
\end{proof}

We can now explain why the hard instance of \Cref{sec:lower_bound} is already contained in the
online fair division model above.
In each active round \(t+1\) of the hard instance, the adversary reveals a vector \(x^{t+1}\in[1/2,1)^n\); after the forced bad prefix is obtained, we may pad with all-zero vectors.
Interpret this as an arriving item $g^{t+1}$ with valuations $v_i(g^{t+1}) := x_i^{t+1}$.
This gives the same proportionality deficits \(d_i^t=\Prop_i^t-v_i(P_i^t)\) as in
\Cref{sec:lower_bound}.

The next theorem transfers the abstract lower bound to envy.
The reason is that a proportionality deficit for some agent implies envy toward at least one other agent, so the hard instance for proportionality also creates large envy.
\begin{theorem}[Envy lower bound in the normalized online fair-division model]
\label{thm:envy-lb-from-framework}
Fix \(n\ge 2\) and a horizon \(T\ge n\).
For every online allocation algorithm, there exists an adaptive valuation sequence
with item values in \([0,1]\) such that
\[
Envy^{T}
\in
\Omega\!\left(\sqrt{\frac{T}{n}}\right).
\]
In particular, for fixed \(n\), this specializes to \(Envy^{T}\in\Omega(\sqrt{T})\).
\end{theorem}

\begin{proof}
By the time-indexed form of \Cref{thm:lb-framework}, there is a universal
constant \(\alpha>0\) such that, for every \(T\ge n\), the adaptive adversary can
force a prefix proportionality deficit larger than
\[
\alpha\sqrt{\frac{T}{n}}.
\]
After the first such prefix, the adversary may set all subsequent item values to
\(0\) for every agent. Hence utilities, proportional shares, and deficits do not
change afterward. Therefore, at time \(T\), there is some agent \(i\) with
\[
d_i^T
=
\Prop_i^T-v_i(P_i^T)
>
\alpha\sqrt{\frac{T}{n}}.
\]
By Lemma~\ref{lem:envy-from-deficit-framework},
\[
Envy^{T}
\ge
d_i^T.
\]
Thus
\[
Envy^{T}
\in
\Omega\!\left(\sqrt{\frac{T}{n}}\right).
\]
\end{proof}

\paragraph{Discussion.}
In the regime \(T\ge n\), the theorem gives
\(
Envy^T\in\Omega\!\left(\sqrt{\frac{T}{n}}\right)
\)
in the normalized online fair-division model of \citet{benade2018make}. Thus, in this regime, our lower bound
reaches the square-root rate corresponding to the endpoint \(r=1\). By comparison, \citet{benade2018make}
proved \(\Omega((T/n)^{r/2})\) for every fixed \(n\), every \(T\), and every \(r<1\).

\section{Comparison with \texorpdfstring{\citet{benade2018make}}{Benade et al. (2018)} deterministic online algorithm}
\label[appendix]{app:benade-intermediate-linear}
\citet{benade2018make} 
present a deterministic online allocation algorithm that guarantees that the envy at the final time $T$ is at most $O(\sqrt{\frac{T}{n}})$, which is analogous to our Corollary~\ref{cor:ef-application} .

However, their analysis relies on fixed global normalization: they assume that all item values are in $[0,1]$. This is equivalent to assuming that the maximum value of an item is known in advance.

One could try to remove this global normalization requirement by re-normalizing online, whenever a new
item exceeds the previously seen maximum value. However, such a renormalization rescales the entire
history and can change the signed envy differences even without any allocation action.
For example, suppose agent~1 (Alice) receives $100$ items of value $1$ that only she wants (so agent~2
values them at $0$). If all these items are allocated to Alice, then after these $100$ rounds the signed
envy differences are exactly
\[
f_{12}=-100,
\qquad
f_{21}=0.
\]
Under \citet{benade2018make} exponential potential (see \eqref{eq:benade-potential-app}),
\[
\phi_{12}(100)=\kappa^{T-100}\exp\!\bigl(s(-100-\lambda)\bigr),
\qquad
\phi_{21}(100)=\kappa^{T-100}\exp\!\bigl(s(0-\lambda)\bigr),
\]
so the total potential at that point is
\[
\phi(100)=\phi_{12}(100)+\phi_{21}(100)
=\kappa^{T-100}\exp(-s\lambda)\bigl(e^{-100s}+1\bigr).
\]
Now suppose that at the beginning of round~101 an item of value $1000$ arrives and the algorithm re-normalizes by the running
maximum. This rescales all past values by a factor $1/1000$, so the same allocation history now
corresponds to
\[
\widetilde f_{12}=-100/1000=-0.1,
\qquad
\widetilde f_{21}=0,
\]
and hence the potential becomes
\[
\widetilde\phi(100)
=\kappa^{T-100}\exp(-s\lambda)\bigl(e^{-0.1s}+1\bigr)
>\phi(100).
\]
Here the index $t$ counts the number of allocated items; the renormalization happens at the beginning of round~101 (before allocating the new item), so the state is still the post-$100$-allocation state and the appropriate notation is $\widetilde\phi(100)$.

In particular, the single term $\phi_{12}(100)$ increases from
$\kappa^{T-100}\exp\!\bigl(s(-100-\lambda)\bigr)$ to
$\kappa^{T-100}\exp\!\bigl(s(-0.1-\lambda)\bigr)$, i.e., by the multiplicative factor
\[
\frac{\kappa^{T-100}e^{s(-0.1-\lambda)}}{\kappa^{T-100}e^{s(-100-\lambda)}}=\exp(99.9\,s).
\]
Thus, the potential used in \citet{benade2018make} can \emph{increase} sharply in a single step purely
because of rescaling, which breaks the monotonicity property that is a key ingredient in their analysis.

Independently of these scale issues, a main motivation for our framework is that \emph{final-time} guarantees for a fixed horizon $T$
do not automatically give \emph{anytime} (prefix-wise) control.
Below we show an explicit two-agent instance in which the deterministic rule of
\citet{benade2018make} (specialized to $n=2$) can have \emph{linear} envy at an intermediate time
$t^*=\lfloor \sqrt{T}\rfloor$.
This does not contradict their $O(\sqrt{T})$ final-time guarantee; rather, it shows that the analysis
does not imply an $O(\sqrt{t})$-type bound that holds for \emph{all} prefixes $t$.

We consider $n=2$ agents and horizon $T$.
At each round $t\in[T]$ 
one item arrives with values $(v_{1t},v_{2t})\in[0,1]^2$.
\citet{benade2018make}'s deterministic algorithm allocates each item to minimize the potential
\begin{equation}\label{eq:benade-potential-app}
\phi(t)\;=\;\sum_{i\neq j}\phi_{ij}(t),
\qquad
\phi_{ij}(t)\;=\;\kappa^{T-t}\exp\!\bigl(s(f_{ij}(t)-\lambda)\bigr),
\end{equation}
where (for $n=2$) $f_{12}(t)$ and $f_{21}(t)$ are the signed envy-differences and
\[
\Envy_{ij}(t)\;=\;\max\{f_{ij}(t),0\}.
\]
The parameters $s$ and $\lambda$ depend explicitly on $T$:
\begin{equation}\label{eq:benade-params-app}
s
=\sqrt{\,2\log\!\Bigl(1+\frac{n\log n}{T}\Bigr)}
=\sqrt{\,2\log\!\Bigl(1+\frac{2\log 2}{T}\Bigr)},
\qquad
\lambda
=10\sqrt{\frac{T\log n}{n}}
=10\sqrt{\frac{T\log 2}{2}}.
\end{equation}
From this, it is already clear that their rule is tuned to a particular  horizon, and thus it is not fully-online in our sense, even though it still gives guarantees at every time step up to the specified horizon. We now prove this formally.

\begin{proposition}[Linear envy up to time $\lfloor\sqrt{T}\rfloor$]\label{thm:benade-linear-envy-app}
There exists a sequence of valuations such that,
for any $T\geq 1$, 
if \citet{benade2018make}'s deterministic online algorithm is executed with horizon $T$, then for every time $t\le t^*:=\lfloor\sqrt{T}\rfloor$, the envy at time $t$ satisfies
\[
\Envy^{t} \;\ge\; 0.1\, t.
\]
\end{proposition}

\begin{proof}
Fix $\rho=0.1$ and let $t^*=\lfloor\sqrt{T}\rfloor$.
Consider the (nonadaptive) valuation sequence
\[
(v_{1t},v_{2t})=
\begin{cases}
(1,\rho), & 1\le t\le t^*,\\
(0,0), & t^*< t\le T.
\end{cases}
\]

We prove that the algorithm allocates every item $1,2,\dots,t^*$ to agent~1; then agent~2 envies agent~1
by exactly $\rho t^*$ at time $t^*$.

\paragraph{Step 1: Two-agent decision rule.}
For $n=2$, minimizing $\phi(t)$ is equivalent to minimizing
\begin{align}
\label{eq:benade-potential-2}
\exp\!\bigl(sf_{12}(t)\bigr)+\exp\!\bigl(sf_{21}(t)\bigr),
\end{align}
since at any round $t$ the factors $\kappa^{T-t}$ and $\exp(-s\lambda)$ are common to both recipient choices.
Thus, to decide who gets the item at round $t$, the algorithm chooses the recipient that minimizes \eqref{eq:benade-potential-2}
after updating the envy levels $(f_{12},f_{21})$ by the chosen allocation.

\paragraph{Step 2: Track \texorpdfstring{$f_{12}(t)$}{f_{12}(t)} and $f_{21}(t)$ if all first $t$ items go to agent 1.}
Assume items $1,2,\dots,t$ have been allocated to agent~1.
Then for each such item denoted by $\ell\in\{1,\dots,t\}$, agent~1's signed envy-difference toward agent~2 decreases by $v_{1\ell}=1$,
while agent~2's signed envy-difference toward agent~1 increases by $v_{2\ell}=\rho$.
Therefore
\begin{equation}\label{eq:f-track-app}
f_{12}(t)=-t,
\qquad
f_{21}(t)=\rho t.
\end{equation}

\paragraph{Step 3: A one-step comparison inequality.}
Consider round $t+1\le t^*$ (so the arriving item has values $(1,\rho)$).

If we give item $t+1$ to agent~1, then the envy-levels change to
\[
f_{12}=-(t+1),
\qquad
f_{21}=\rho(t+1),
\]
so the objective becomes
\begin{equation}\label{eq:S1-app}
S_1(t):=\exp\!\bigl(-s(t+1)\bigr)+\exp\!\bigl(s\rho(t+1)\bigr).
\end{equation}

If we give item $t+1$ to agent~2, then $f_{12}$ increases by $1$ and $f_{21}$ decreases by $\rho$, so
\[
f_{12}=1-t,
\qquad
f_{21}=\rho(t-1),
\]
so the objective becomes
\begin{equation}\label{eq:S2-app}
S_2(t):=\exp\!\bigl(s(1-t)\bigr)+\exp\!\bigl(s\rho(t-1)\bigr).
\end{equation}

We claim that $S_1(t)\le S_2(t)$ whenever
\begin{equation}\label{eq:key-ineq-app}
\exp\!\bigl(s(1+\rho)t\bigr)\ \le\ \frac{\sinh(s)}{\sinh(s\rho)}.
\end{equation}
Indeed,
\begin{align*}
S_1(t)\le S_2(t)
&\Longleftrightarrow
\exp\!\bigl(-s(t+1)\bigr)-\exp\!\bigl(s(1-t)\bigr)
\le
\exp\!\bigl(s\rho(t-1)\bigr)-\exp\!\bigl(s\rho(t+1)\bigr)
\\
&\Longleftrightarrow
-2\exp(-st)\sinh(s)\ \le\ -2\exp(s\rho t)\sinh(s\rho)
\\
&\Longleftrightarrow
\exp(-st)\sinh(s)\ \ge\ \exp(s\rho t)\sinh(s\rho)
\\
&\Longleftrightarrow
\exp\!\bigl(s(1+\rho)t\bigr)\ \le\ \frac{\sinh(s)}{\sinh(s\rho)},
\end{align*}
which is exactly \eqref{eq:key-ineq-app}. Therefore, as long as \eqref{eq:key-ineq-app} holds,
the algorithm allocates item $t+1$ to agent~1.

\paragraph{Step 4: Lower bound the right-hand side up to \texorpdfstring{$t=t^*-1$}{t=t^*-1}.}
Using convexity of $\sinh$ and $\sinh(0)=0$, for any $\rho\in(0,1)$ we have
\[
\sinh(s\rho)\le \rho\sinh(s)
\qquad\Longrightarrow\qquad
\frac{\sinh(s)}{\sinh(s\rho)}\ge \frac{1}{\rho}.
\]
So it is sufficient that
\begin{equation}\label{eq:suff-exp-app}
\exp\!\bigl(s(1+\rho)t\bigr)\le \frac{1}{\rho}
\qquad\Longleftrightarrow\qquad
t\le \frac{\log(1/\rho)}{s(1+\rho)}.
\end{equation}

We will show that for $\rho=0.1$ and all $T \geq 1$ we have
\begin{equation}\label{eq:tminus1-bound-app}
t^*-1=\lfloor\sqrt{T}\rfloor-1\ \le\ \frac{\log(1/\rho)}{s(1+\rho)}.
\end{equation}

First note that for all $x\ge 0$, $\log(1+x)\le x$, hence by \eqref{eq:benade-params-app},
\[
s=\sqrt{\,2\log\!\Bigl(1+\frac{2\log 2}{T}\Bigr)}
\ \le\
\sqrt{\,2\cdot \frac{2\log 2}{T}}
\ =\
2\sqrt{\frac{\log 2}{T}}.
\]
Therefore
\[
\frac{\log(1/\rho)}{s(1+\rho)}
\ \ge\
\frac{\log(1/\rho)}{1+\rho}\cdot \frac{1}{2}\sqrt{\frac{T}{\log 2}}.
\]
With $\rho=0.1$, we have $\log(1/\rho)=\log 10$ and $(1+\rho)=1.1$, and one checks numerically that
\[
\frac{\log 10}{1.1}\cdot \frac{1}{2\sqrt{\log 2}}\ >\ 1.
\]
Therefore for all $T\ge 1$,
\[
\frac{\log(1/\rho)}{s(1+\rho)}
\ \ge\
\sqrt{T}
\ \ge\
t^*-1.
\]
So by induction using Steps~2--3, the algorithm allocates every item $1,2,\dots,t^*$ to agent~1.

\paragraph{Step 5: Compute envy at time \texorpdfstring{$t$}{t}.}
For any time $t\le t^*$, agent~2's value for agent~1's bundle is $\rho t$ and her value for her own bundle is $0$.
Thus $\Envy_{2\to 1}(t)=\rho t$, and hence $\Envy^{t}\ge \rho t$ for all $t\le t^*$.

\end{proof}

\section{Intuition via Lorenz Curves: Why Minimize a Smooth \texorpdfstring{$p$}{p}-Potential?}
\label[appendix]{sec:lorenz-intuition}
This section gives an informal intuition for the design of the $p$-potential rule
\eqref{eq:f-def-general}--\eqref{eq:rvm-choice-general}.
In our framework, we keep track of ``how unfair things currently are,'' we want to accumulate that information over time,
and we want a single score that lets us compare different feasible actions in the next round. Indeed, one may think that using $\max_q z_q^t$ (the worst-off deficit) is appropriate for this task. However, 
$\max_q z_q^t$ does not tell us whether:
\begin{itemize}
  \item unfairness is \emph{concentrated} (there are several large deficits, close to the maximum), or
  \item unfairness is \emph{spread out} (only one deficit is large, and the others are moderate or small).
\end{itemize}
The first situation is more risky, as when several deficits are large, an adversary can create a situation in which the algorithm must increase one of these deficits, and then, for any choice of the algorithm, one of the deficits will grow over the bound. Thus, we must control all the deficits, not only the largest one.

We introduce some notation. Sort the deficits in nonincreasing order,
\[
z_{(1)}^t \ \ge\ z_{(2)}^t \ \ge\ \cdots \ \ge\ z_{(m)}^t,
\]
and define the \emph{tail sums}:
\begin{equation}
\label{eq:tail-sums}
S^{t}(k)\ :=\ \sum_{i=1}^{k} z_{(i)}^t
\qquad (k=0,1,\ldots,m),
\end{equation}
with $S^{t}(0)=0$.
So $S^{t}(k)$ is the total deficit among the $k$ worst-off deficits.
This is the ``upper tail'' of the deficit distribution: it focuses on the largest values of $z_q^t$.

Tail sums are directly related to the fairness threshold $c$ in two ways:
\begin{enumerate}
    \item The worst deficit is the first tail sum.
\[
S^{t}(1)\ =\ z_{(1)}^t\ =\ \max_{q\in {Q}} z_q^t.
\]
Therefore, $c$-fairness is equivalent to $S^{t}(1)\le c$.
\item {Many $c$-violated requirements force a large tail sum.}

Let $k:=|\mathcal{D}^{t}(c)|$.
That is, the $k$-th largest deficit satisfies $z_{(k)}^t>c$, and therefore
\begin{equation}
\label{eq:tail-implies-many}
S^{t}(k)\ =\ \sum_{i=1}^k z_{(i)}^t \ >\ k\,c.
\end{equation}
\end{enumerate}

\subsection*{Lorenz curves for deficit vectors (via tail sums)}

Lorenz curves are a classical way to visualize how mass is distributed across a vector.
Consider \Cref{fig:lorenz-tail-intuition}, which depicts the Lorenz curve of a deficit vector. Note that if \emph{many} deficits exceed $c$, then the tail-sum curve rises above the straight line $kc$.

\begin{figure}[t]
\centering

\pgfmathsetmacro{\cval}{1.5}   
\pgfmathsetmacro{\kzero}{4}   
\pgfmathsetmacro{\kcut}{\kzero+0.5}

\def\sorteddeficits{
  (1,3.4) (2,2.8) (3,2.1) (4,1.9) (5,1.2)
  (6,0.9) (7,0.7) (8,0.4) (9,0.3) (10,0.1)
}
\def\tailsums{
  (0,0) (1,3.4) (2,6.2) (3,8.3) (4,10.2) (5,11.4)
  (6,12.3) (7,13.0) (8,13.4) (9,13.7) (10,13.8)
}

\begin{minipage}{0.49\linewidth}
\centering
\begin{tikzpicture}
\begin{axis}[
  width=\linewidth,
  height=0.72\linewidth,
  axis lines=left,
  xmin=0.5, xmax=10.5,
  ymin=0, ymax=3.8,
  xlabel={rank $i$ (sorted: $z_{(1)}^t \ge \cdots \ge z_{(m)}^t$)},
  ylabel={$z_{(i)}^t$},
  xtick={1,4,7,10},
  ytick={0,1,2,3},
  ticklabel style={font=\scriptsize},
  label style={font=\scriptsize},
  clip=false,
]
  \addplot+[ybar, bar width=4pt] coordinates {\sorteddeficits};

  \addplot[densely dashed] coordinates {(0.5,\cval) (10.5,\cval)};
  \node[font=\scriptsize, anchor=west] at (axis cs:10.55,\cval) {$c$};

  \addplot[dotted] coordinates {(\kcut,0) (\kcut,3.8)};
  \node[font=\scriptsize, anchor=south] at (axis cs:\kcut,3.85) {$k=|\mathcal{D}^{t}(c)|$};

\end{axis}
\end{tikzpicture}

\end{minipage}
\hfill
\begin{minipage}{0.49\linewidth}
\centering
\begin{tikzpicture}
\begin{axis}[
  width=\linewidth,
  height=0.72\linewidth,
  axis lines=left,
  xmin=0, xmax=10,
  ymin=0, ymax=15,
  xlabel={$k$ (number of worst requirements)},
  ylabel={$S^{t}(k)=\sum_{i=1}^k z_{(i)}^t$},
  xtick={0,2,4,6,8,10},
  ytick={0,3,6,9,12,15},
  ticklabel style={font=\scriptsize},
  label style={font=\scriptsize},
  legend style={font=\scriptsize, draw=none, fill=none, at={(0.02,0.98)}, anchor=north west},
  clip=false,
]
  \addplot[thick, mark=*, mark size=1.2pt] coordinates {\tailsums};
  \addlegendentry{$S^{t}(k)$}

  \addplot[densely dashed, domain=0:10, samples=2] {\cval*x};
  \addlegendentry{$k\,c$}

  \addplot[only marks, mark=*, mark size=2.0pt] coordinates {(\kzero,10.2)};
  \draw[dotted] (axis cs:\kzero,0) -- (axis cs:\kzero,10.2);
  \node[font=\scriptsize, anchor=west] at (axis cs:\kzero+0.2,10.2)
    {$S^{t}(k) > kc$};

\end{axis}
\end{tikzpicture}

\end{minipage}

\caption{Lorenz view of a deficit vector (schematic).
\textbf{Left:} sorted deficits with threshold $c$; $k_0=|\mathcal{D}^{t}(c)|$ is the number exceeding $c$.
\textbf{Right:} tail sums $S^{t}(k)$ compared to the line $k\,c$; when $k_0=|\mathcal{D}^{t}(c)|$ we must have
$S^{t}(k_0)>k_0c$ as in \eqref{eq:tail-implies-many}.}
\Description{Two-panel schematic. Left panel shows a bar chart of sorted deficits with a dashed horizontal line at threshold c and a dotted vertical line at k0. Right panel shows the tail-sum curve S^{t}(k) together with the dashed line k c; the point at k0 lies above the line.}
\label{fig:lorenz-tail-intuition}
\end{figure}

Tail sums can help compare two deficit vectors.
Given two deficit vectors $x,y\in\mathbb{R}_{\ge 0}^m$, if we sort $x$ and $y$ in nonincreasing order we can say that
\[
x \ \text{Lorenz-dominates}\ y
\quad\Longleftrightarrow\quad
\sum_{i=1}^k x_{(i)} \ \le\ \sum_{i=1}^k y_{(i)}\ \ \text{for every }k\in[m].
\]
That is, for \emph{every} choice of ``how many worst deficits you look at,'' the total deficit among those worst
deficits is smaller under $x$ than under $y$. 
Clearly, Lorenz dominance can serve as a fairness notion, since it prefers a low worst-off deficit ($k=1$) and low upper-tail deficits ($k=2,3,\dots$) simultaneously.

However, Lorenz dominance is only a \emph{partial} order.
Two tail-sum curves can cross (one vector is better for small $k$, the other is better for medium $k$),
so Lorenz dominance alone does not always tell us which action to prefer.
To define a single-step decision rule, we need to turn this ``curve comparison'' into one scalar score.

\subsection*{From a curve to a number: symmetric convex potentials}

A standard way to turn a whole vector into one number is to apply the same function to each coordinate and add:
\[
\Phi(z)\ :=\ \sum_{q\in {Q}} f(z_q).
\]
This additive form is not just conceptually natural; it is also exactly what matches our local balance conditions.
Those conditions control one-step changes coordinate-by-coordinate (and then average over a few reference actions),
so a sum-of-coordinates potential is the right container: we can bound each coordinate's contribution and then sum over $q$.

Two structural requirements are important in our setting:

\paragraph{Symmetry (treat all coordinates the same).}
A function $F:\mathbb{R}_{\ge 0}^m\to\mathbb{R}$ is \emph{symmetric} if permuting the coordinates does not change its value.
This matches the idea that the tracked requirements are unlabeled: the score should depend on the multiset
$\{z_q^t:q\in {Q}\}$, not on which $q$ happens to carry which deficit.

\paragraph{Convexity (penalize concentration).}
Convexity is a formal way to say ``one huge deficit is worse than two medium deficits of the same total size.''
Concretely, if we move a small amount of deficit from a larger coordinate to a smaller one
(so the total deficit stays the same but becomes more evenly spread), a convex symmetric score should decrease.
This matches the Lorenz-tail philosophy: we want to avoid very steep upper tails.

These ideas are connected by a classical theorem from majorization theory:

\begin{remark}[Lorenz dominance $\Rightarrow$ smaller symmetric convex potentials (informal)]
\label{rem:schur-convex-intuition}
If a deficit vector $x$ Lorenz-dominates $y$, then any function
$F(x)=\sum_i g(x_i)$ with $g$ convex and nondecreasing will satisfy $F(x)\le F(y)$.
So minimizing a convex symmetric potential is a principled way to ``scalarize'' the Lorenz-tail preference.
\emph{(This even applies to $F(x)=\max_i x_i$, which is convex and symmetric, but it is nonsmooth.)}
\end{remark}

So far, this only tells us \emph{what kind} of scalar objective we want.
Next we explain \emph{which} convex symmetric objective is best fits the goal of tail control.

\subsection*{Why high powers control the upper tail}

A particularly simple family of symmetric convex objectives is given by \emph{high power sums}.
Ignoring smoothing for a moment, consider
\[
\sum_{q\in {Q}} (z_q^t)^{2p}.
\]
You can view this as a ``tail-sensitive energy'': when $p$ is large, the largest coordinates dominate the sum.

\paragraph{Tail counts from a one-line inequality.}
There is a direct deterministic bound that explains the tail connection.
For any threshold $c>0$,
\[
|\mathcal{D}^{t}(c)|
\ =\ |\{q:\ z_q^t>c\}|
\ \le\
\frac{\sum_{q} (z_q^t)^{2p}}{c^{2p}}.
\]
This is the same algebra that underlies Markov's inequality in probability: if you pick $q$ uniformly at random from ${Q}$, then $\Pr[z_q^t>c]=|\mathcal{D}^{t}(c)|/m$, and Markov's inequality gives the same type of bound. Thus, if the power sum is small, then there cannot be many coordinates above a large threshold $c$. So high powers are a natural ``proxy'' (i.e., a surrogate objective) for keeping the Lorenz upper tail under control. Later, we bound a \emph{smooth regularized} version of this power sum using the local balance conditions, and then plug that bound into inequalities of this type to control $|\mathcal{D}^{t}(c)|$.

\paragraph{A continuous viewpoint (area under the tail-count curve).}
There is also an identity (often called the \emph{layer-cake representation} or \emph{Cavalieri principle}) saying that for
nonnegative numbers $(z_q)$ and any $p\ge 1$,
\[
\sum_q (z_q)^{2p}
\ =\
\int_0^\infty 2p\,c^{2p-1}\cdot |\{q:\ z_q>c\}|\,dc.
\]
So $\sum_q (z_q)^{2p}$ is literally a weighted area under the function
$c\mapsto|\mathcal{D}^{t}(c)|$:
larger $p$ puts more weight on larger thresholds $c$, i.e., it focuses more on the far upper tail.
(For a reference, see standard analysis texts that cover the layer-cake formula, e.g., Lieb--Loss, \emph{Analysis}.)

This is the bridge from ``Lorenz tails'' to ``high powers.''

\begin{figure}[t]
\centering

\pgfmathsetmacro{\kzero}{4}  
\pgfmathsetmacro{\pval}{4}   
\pgfmathsetmacro{\kcut}{\kzero+0.5}

\def\deficitshares{%
  (1,0.2464) (2,0.2029) (3,0.1522) (4,0.1377) (5,0.0870)
  (6,0.0652) (7,0.0507) (8,0.0290) (9,0.0217) (10,0.0072)
}
\def\cumdeficit{%
  (0,0) (1,0.2464) (2,0.4493) (3,0.6014) (4,0.7391) (5,0.8261)
  (6,0.8913) (7,0.9420) (8,0.9710) (9,0.9928) (10,1.0000)
}

\def\potshares{%
  (1,0.4148) (2,0.2585) (3,0.1344) (4,0.1080) (5,0.0409)
  (6,0.0227) (7,0.0136) (8,0.0044) (9,0.0025) (10,0.0003)
}
\def\cumpot{%
  (0,0) (1,0.4148) (2,0.6732) (3,0.8076) (4,0.9156) (5,0.9565)
  (6,0.9792) (7,0.9928) (8,0.9972) (9,0.9997) (10,1.0000)
}

\begin{tikzpicture}
\begin{groupplot}[
  group style={
    group size=2 by 2,
    group name=lp,
    horizontal sep=0.14\linewidth,
    vertical sep=0.18\linewidth,   
  },
  width=0.40\linewidth,            
  height=0.28\linewidth,           
  axis lines=left,
  ticklabel style={font=\scriptsize},
  label style={font=\scriptsize},
  title style={font=\scriptsize, yshift=-1pt},
  xlabel style={font=\scriptsize, yshift=-1pt},
  ylabel style={font=\scriptsize},
  clip=true,
]

\nextgroupplot[
  title={Raw deficit mass},
  xmin=0.5, xmax=10.5,
  ymin=0, ymax=0.46,
  xtick={1,4,7,10},
  ytick={0,0.2,0.4},
  ylabel={share},
  xticklabels=\empty,
]
  \path[fill=black!6] (axis cs:0.5,0) rectangle (axis cs:\kcut,0.46);
  \addplot+[ybar, bar width=4pt, fill=black!20, draw=black] coordinates {\deficitshares};
  \addplot[dotted] coordinates {(\kcut,0) (\kcut,0.46)};
  \node[font=\scriptsize, anchor=north west, fill=white, inner sep=1pt] at (axis description cs:0.02,0.88) {$z_{(i)}/\sum_j z_{(j)}$};

\nextgroupplot[
  title={Upper-Lorenz share},
  xmin=0, xmax=10,
  ymin=0, ymax=1.05,
  xtick={0,2,4,6,8,10},
  ytick={0,0.5,1.0},
  axis y line*=right,
  yticklabel pos=right,
  ylabel={cum. share},
  xticklabels=\empty,
]
  \path[fill=black!4] (axis cs:0,0) rectangle (axis cs:\kzero,1.05);
  \addplot[thick, mark=*, mark size=1.1pt] coordinates {\cumdeficit};
  \addplot[dotted] coordinates {(\kzero,0) (\kzero,1.05)};
  \addplot[only marks, mark=*, mark size=2.0pt] coordinates {(\kzero,0.7391)};
  \node[font=\scriptsize, fill=white, inner sep=1.2pt, rounded corners=1pt, anchor=south east]
    (Ltxt) at (axis description cs:0.94,0.26)
    {$S(k_0)/S(m)\approx 0.74$};
  \draw[->, thin] (Ltxt.north west) -- (axis cs:\kzero,0.7391);

\nextgroupplot[
  title={After convex lens},
  xmin=0.5, xmax=10.5,
  ymin=0, ymax=0.46,
  xtick={1,4,7,10},
  ytick={0,0.2,0.4},
  ylabel={share},
  xlabel={rank $i$ (sorted)},
]
  \path[fill=black!6] (axis cs:0.5,0) rectangle (axis cs:\kcut,0.46);
  \addplot+[ybar, bar width=4pt, fill=black!20, draw=black] coordinates {\potshares};
  \addplot[dotted] coordinates {(\kcut,0) (\kcut,0.46)};
\node[
  font=\scriptsize,
  anchor=north west,
  fill=white,
  inner sep=1pt
] at (rel axis cs:0.10,0.97)
{$\tilde f(z_{(i)})/\sum_j \tilde f(z_{(j)})$};

\nextgroupplot[
  title={Cumulative potential share},
  xmin=0, xmax=10,
  ymin=0, ymax=1.05,
  xtick={0,2,4,6,8,10},
  ytick={0,0.5,1.0},
  axis y line*=right,
  yticklabel pos=right,
  ylabel={cum. share},
  xlabel={$k$ (worst coords)},
]
  \path[fill=black!4] (axis cs:0,0) rectangle (axis cs:\kzero,1.05);
  \addplot[densely dashed] coordinates {\cumdeficit}; 
  \addplot[thick, mark=*, mark size=1.1pt] coordinates {\cumpot}; 
  \addplot[dotted] coordinates {(\kzero,0) (\kzero,1.05)};
  \addplot[only marks, mark=*, mark size=2.0pt] coordinates {(\kzero,0.9156)};

\node[
  font=\scriptsize,
  anchor=south east,
  fill=white,
  inner sep=1pt,
  align=left
] at (rel axis cs:1.005,0.5)
{dashed: $S(k)/S(m)$\\solid: potential};

  \node[font=\scriptsize, fill=white, inner sep=1.2pt, rounded corners=1pt, anchor=south east]
    (Ptxt) at (axis description cs:0.94,0.26)
    {$\approx 0.92$ of potential mass};
  \draw[->, thin] (Ptxt.north west) -- (axis cs:\kzero,0.9156);

\end{groupplot}

\coordinate (topcenter) at ($(lp c1r1.south)!0.5!(lp c2r1.south)$);
\coordinate (botcenter) at ($(lp c1r2.north)!0.5!(lp c2r2.north)$);
\coordinate (midgap) at ($(topcenter)!0.5!(botcenter)$);

\node[
  font=\scriptsize,
  fill=white,
  inner sep=2pt,
  rounded corners=1pt,
  align=center,
  text width=0.72\linewidth 
] (lens) at (midgap)
{convex lens on deficits: $u \mapsto \tilde f(u)$\\[-1pt]
$\tilde f(u)=f(u)-f(0),\qquad f(u)=(u^2+4p^2\sigma^2)^p,\ \ p=\pval$};

\draw[->, thin] (lens.north) -- (topcenter);
\draw[->, thin] (lens.south) -- (botcenter);

\end{tikzpicture}

\caption{Lorenz tails $\leftrightarrow$ the smooth $p$-potential (schematic).
Top row: deficit mass across ranks and the upper-Lorenz share $S(k)/S(m)$.
Bottom row: after applying the convex magnifying lens $\tilde f(u)=f(u)-f(0)$ with $f(u)=(u^2+4p^2\sigma^2)^p$ (baseline removed only to make ``potential mass'' visible),
the same rank-based views show that the top $k_0$ coordinates carry a much larger fraction of the total mass.
This is the visual meaning of ``high-power smooth potentials are Lorenz-tail sensitive.''}
\Description{Four-panel schematic with a centered lens label between the rows. The lens label is placed in the vertical gap and does not overlap the plots.}
\label{fig:lorenz-vs-potential-lens}
\end{figure}

Figure~\ref{fig:lorenz-vs-potential-lens} shows the scalarization pictorially: the $p$-potential is the same Lorenz-style
mass distribution after a convex magnification that concentrates weight on the worst deficits.

\subsection*{Why we smooth: the regularized potential \texorpdfstring{$f(u)=(u^2+4p^2\sigma^2)^p$}{f(u)=(u^2+4p^2\sigma^2)^p}}

If our only goal were to prefer smaller upper tails, we could try to directly minimize $\sum_q (z_q^t)^{2p}$.
The reason we do not take these options literally is that our analysis is driven by \emph{one-step local balance conditions}.

In round $t{+}1$, the local balance conditions give reference actions
\(\hat a_1^{t+1},\ldots,\hat a_n^{t+1}\) and, for each tracked requirement \(q\),
raw shifts \(\Delta_q^{t+1}(1),\ldots,\Delta_q^{t+1}(n)\) satisfying
\[
z_q^{t+1}(\hat a_k^{t+1})
\le
[z_q^t+\Delta_q^{t+1}(k)]_+
\qquad\text{for all }k\in[n].
\]
These shifts have nonpositive mean and bounded squared size. To turn such
statements into a bound on a scalar objective, we compare the resulting
one-step change in \(f\) to $f(z_q^t)$ using a Taylor expansion, so that the remainder term is controlled by $(\Delta_q^{t+1}(a))^2$.
This requires $f$ to be smooth and to have curvature that stays controlled even when $z_q^t$ is close to $0$.

This is why we add the constant $4p^2\sigma^2$ inside the power, and define \eqref{eq:f-def-general}:
\begin{equation*}
f(u)\ :=\ (u^2+4p^2\sigma^2)^p,
\qquad
\Phi^t\ :=\ \sum_{q\in {Q}} f(z_q^t),
\qquad
\Psi^t\ :=\ (\Phi^t)^{1/p}.
\end{equation*}
That constant plays two intuitive roles:

\paragraph{(a) It keeps tail sensitivity where we need it.}
When $u\gg p$, we have $u^2\gg 4p^2\sigma^2$ and thus
$f(u)\approx (u^2)^p=u^{2p}$.
So large deficits are still penalized like a high power, exactly as in the tail-control discussion above.

\paragraph{(b) It makes the potential gentle and stable near $0$.}
When $u\ll p$, the $+4p^2\sigma^2$ term dominates, so tiny deficits (and tiny fluctuations) do not dominate the score.

\paragraph{(c) It controls curvature under bounded one-step changes (analysis-critical).}
The most technical (but practically important) point is that with the $+4p^2\sigma^2$ baseline we can bound the curvature of $f$
uniformly when we perturb $u$ by a bounded amount (e.g., $u\mapsto u+\xi$ with $|\xi|\le \sigma$).
This lets us say: ``the potential after one round is not much larger than what a first-order prediction would suggest,''
which is exactly what we need when we average over the reference actions in the local balance conditions.

Finally, we take $\Psi^t=(\Phi^t)^{1/p}$ simply as a convenient rescaling:
it is monotone in $\Phi^t$, so minimizing $\Psi$ is the same as minimizing $\Phi$,
but the $1/p$ power makes the expressions in the analysis cleaner.

\subsection*{Why the rule minimizes the \emph{next} potential}

The $p$-potential rule is defined by a one-step lookahead:
for each feasible action $a\in A^{t+1}$ we can compute the \emph{hypothetical next} deficits
$z^{t+1}(a)$ and hence $\Psi^{t+1}(a)$, and we choose
\[
a^{t+1}\in\arg\min_{a\in A^{t+1}} \Psi^{t+1}(a).
\]
This is the most direct way to implement the Lorenz-tail intuition online:
\emph{pick the action whose resulting deficit distribution has the flattest possible upper tail.}

There are three concrete reasons we evaluate at time $t+1$ rather than $t$:
\begin{itemize}
  \item Fairness is judged on the \emph{state after the decision} (the vector that actually results from our choice).
  \item In scale-aware settings, even the normalization used in $z^{t+1}(a)$ can depend on $a$,
        so looking only at $\Psi^t$ can miss exactly the effect that differentiates actions.
  \item Our local balance conditions are one-step statements for round $t{+}1$, and minimizing $\Psi^{t+1}(a)$ lets us compare our
        choice to the reference actions via ``minimum $\le$ average.''%
\end{itemize}

\paragraph{Choosing $p$ is choosing how ``tail-focused'' we are.}
The parameter $p$ is a knob:
small $p$ behaves more like a mean/energy control, while larger $p$ increasingly concentrates attention on the worst deficits.
When $p\in \Theta(\log m)$, the aggregation behaves like a soft-max:
it is strong enough to control all $m$ coordinates, but only incurs a logarithmic price in $m$.
This is particularly useful when $m$ is very large (for example, when the set of tracked requirements contains all envy pairs).

\subsection*{Bounding the Gini coefficient}
Another, classical way to compress Lorenz-curve information into one number is the \emph{Gini coefficient},
defined as (twice) the area between the Lorenz curve and the diagonal.
In our setting this should be interpreted on the \emph{deficit} vector $z^t$:
the whole point of working with $z_q^t$ (rather than raw deficits in heterogeneous units) is that each coordinate has
already been scaled by its own natural tolerance/unit, so the coordinates are commensurate and a Lorenz/Gini summary is meaningful.

Let
\[
T^t \ :=\ \sum_{q\in {Q}} z_q^t \ =\ S^t(m),
\qquad\text{and}\qquad
\bar z^t \ :=\ \frac{T^t}{m}.
\]
Then the Gini coefficient of the deficit vector can be expressed directly from the tail sums \eqref{eq:tail-sums} as
\begin{equation}
\label{eq:gini-tail}
\mathrm{Gini}(z^t)
\ =\
\frac{2}{m\,T^t}\sum_{k=1}^{m}\left(S^{t}(k)-k\,\bar z^t\right)
\ =\
\frac{1}{2m\,T^t}\sum_{q,q'\in {Q}} \bigl|z_q^t-z_{q'}^t\bigr|,
\end{equation}
(with the convention $\mathrm{Gini}(0):=0$ when $T^t=0$).
Thus $\mathrm{Gini}(z^t)$ averages (over $k$) the gap between the upper-tail curve $S^t(k)$ and the equal-share baseline
$k\bar z^t$; it measures how concentrated the deficits are across coordinates.

A caveat is that $\mathrm{Gini}$ is \emph{scale-invariant} (even after normalization):
multiplying all deficits by a common constant does not change $\mathrm{Gini}$.
A scale-sensitive companion is the \emph{Gini mean difference}
\begin{equation}
\label{eq:gmd-def}
\mathrm{GMD}(z^t)
\ :=\ 2\bar z^t\,\mathrm{Gini}(z^t)
\ =\ \frac{1}{m^2}\sum_{q,q'\in {Q}}\bigl|z_q^t-z_{q'}^t\bigr|,
\end{equation}
which still measures inequality but now also reflects the \emph{absolute magnitude} of the deficits.

Finally, $\mathrm{GMD}$ is controlled by high-power (tail-sensitive) energies:
for any $p\ge 1$,
\begin{equation}
\label{eq:gmd-lp}
\mathrm{GMD}(z^t)
\ \le\
2\,m^{-1/(2p)}\Bigl(\sum_{q\in {Q}} (z_q^t)^{2p}\Bigr)^{1/(2p)}.
\end{equation}
So objectives that keep the $2p$-power sum small also keep a Gini-style dispersion measure small.
In particular, recalling the regularized high-power potential \eqref{eq:f-def-general},
$\Phi^t=\sum_{q\in {Q}} \bigl((z_q^t)^2+4p^2\sigma^2\bigr)^p$ and $\Psi^t=(\Phi^t)^{1/p}$,
we have $\sum_{q\in {Q}} (z_q^t)^{2p}\le \Phi^t$ and therefore
\[
\mathrm{GMD}(z^t)\ \le\ 2\,m^{-1/(2p)}(\Phi^t)^{1/(2p)}\ =\ 2\,m^{-1/(2p)}\sqrt{\Psi^t}.
\]
Hence, our $p$-potential algorithm also keeps the Gini mean difference bounded.

\section{Exact computation of optimal policies for the unit-scale proportionality game}
\label[appendix]{app:optimal-policy}

This appendix gives a computational supplement for the online fair division model of
\Cref{subsec:ex-propxc}.  It studies a \emph{unit-scale} (globally normalized) surrogate in which
item values satisfy $v_i(g)\in[0,1]$ for every agent $i$ and arriving item $g$.
Fix an integer slack parameter $c\ge 0$.

\paragraph{The unit-scale surrogate.}
Define agent $i$'s proportional share at time $t$ as $\Prop_i^t := \frac1n v_i(G^t)$, and let
$d_i^t := \Prop_i^t - v_i(P_i^t)$ be her proportionality deficit.
The unit-scale proportionality condition with slack $c$ is:
\begin{equation}
\label{eq:unit-scale-prop}
d_i^t \ \le\ c
\qquad\text{for all } i\in[n]\text{ and all times }t.
\end{equation}
This condition is \emph{necessary} for standard $\PROPc$ under the normalization $v_i(g)\le 1$:
indeed, if $\PROPc$ holds then $d_i^t\le \sum_{g\in B_i^t(c)} v_i(g)\le c$.
It is also necessary for our (stronger) scale-based $\PROPax{c}$, since $U_i^t\le 1$ implies
$d_i^t \le cU_i^t \le c$.
Therefore, \emph{violating} \eqref{eq:unit-scale-prop} certifies violation of $\PROPc$ (and of $\PROPax{c}$).

\paragraph{Goal of this appendix.}
We show how, for small $n$ and $c$, one can compute (exactly) the value of the worst-case online game for
\eqref{eq:unit-scale-prop}: the maximum number of rounds an online allocator can guarantee before an adaptive
adversary can force a violation.  The computation is exponential (hence the name $\mathsf{EXP}$) and is intended
as intuition / verification for small parameters rather than a scalable method.
It can also be interpreted as an exact computation for the surplus dynamics underlying the hard instance of
\Cref{sec:lower_bound} after a linear change of variables.

\subsection*{Shifted scaled surplus dynamics}

As in \Cref{sec:detailed_compare}, define agent $i$'s \emph{scaled proportionality surplus} at time $t$ by
\[
\dsprop_i^t \ :=\ n\,v_i(P_i^t)-v_i(G^t).
\]
Then $d_i^t = \Prop_i^t-v_i(P_i^t) = -\dsprop_i^t/n$, so \eqref{eq:unit-scale-prop} is equivalent to
$\dsprop_i^t \ge -nc$ for all $i$.

Define the \emph{shifted surplus} by
\[
\dsprop_i^t(c)\ :=\ \dsprop_i^t + nc,
\qquad
\vec{\dsprop}^{\,t}(c)\ :=\bigl(\dsprop_1^t(c),\ldots,\dsprop_n^t(c)\bigr).
\]
Then \eqref{eq:unit-scale-prop} holds at time $t$ iff $\dsprop_i^t(c)\ge 0$ for all $i$.
Since $c$ is fixed throughout this appendix, we will slightly abuse notation and write $\dsprop_i^t$ for
$\dsprop_i^t(c)$ and $\vec{\dsprop}^{\,t}$ for $\vec{\dsprop}^{\,t}(c)$.

\paragraph{One-step update.}
In round $t+1$, an item $g^{t+1}$ arrives with values $(v_1(g^{t+1}),\ldots,v_n(g^{t+1}))\in[0,1]^n$,
and the algorithm must irrevocably choose a recipient $a^{t+1}\in[n]$.
For the shifted surplus, the update is:
\begin{equation}
\label{eq:shifted-surplus-update}
\dsprop_i^{t+1} \ =\
\begin{cases}
\dsprop_i^t + (n-1)\,v_i(g^{t+1}), & i=a^{t+1},\\
\dsprop_i^t - v_i(g^{t+1}), & i\neq a^{t+1}.
\end{cases}
\end{equation}
The game ends (the adversary ``wins'') once some coordinate becomes negative, i.e.\ once
$\dsprop_i^t<0$ for some $i$, which implies a violation of \eqref{eq:unit-scale-prop}.

\subsection*{Safe points and domination}

We will reason about states in the extended real space.  Let $\overline{\mathbb{R}}$ be the extended real number
system, obtained from $\mathbb{R}$ by adding $\infty$ and $-\infty$ as actual numbers \cite{aliprantis1998principles}.

\begin{definition}[$k$-safe point]
A point $\vec{x}\in \overline{\mathbb{R}}^n$ is called \emph{$k$-safe} if, whenever the current shifted-surplus
vector equals $\vec{\dsprop}^{\,t}=\vec{x}$, there exists an online allocation rule that guarantees that the
adversary does not win within the next $k$ rounds; i.e.\ after $k$ further rounds we still have
$\dsprop_i^{t+k}\ge 0$ for all $i\in[n]$.
\end{definition}

\begin{definition}[Domination]
A point $\vec{x}=(x_1,\ldots,x_n)\in\overline{\mathbb{R}}^n$ \emph{dominates} another point
$\vec{x}'=(x'_1,\ldots,x'_n)\in\overline{\mathbb{R}}^n$ if $x_i>x'_i$ for all $i\in[n]$.
\end{definition}

\subsection*{A one-step LP and the sets $D^k$}

Fix a point $\vec{x}=(x_1,\ldots,x_n)$ and an index $i\in[n]$.
Consider the following linear program with variables $y$ and $z$:
\begin{equation*}
\begin{array}{ll@{}lll}
\text{maximize}  & \displaystyle y &\\
\text{subject to}& \displaystyle y +(n-1)z &\le x_i\\
                 & \displaystyle y -z &\le x_j , \qquad j\neq i \\
                 & 0 \le z \le 1 \,.
\end{array}
\end{equation*}
Let $Y(x_1,\ldots,x_n,i)$ denote the optimal value of $y$, and let $Z(x_1,\ldots,x_n,i)$ denote the corresponding
optimal value of $z$ (fix an arbitrary choice if there are multiple optima).

Define a sequence of point sets $D^k$ recursively:
\begin{align*}
D^0
&:= \set{(\underbrace{\infty,\ldots,\infty}_{i-1},0,\underbrace{\infty,\ldots,\infty}_{n-i}) \given i\in[n]},
\\[2pt]
D^{k+1}
&:= \set{
\bigl(
Y(x^{1}_1,\ldots,x^{n}_1,1),\ldots,Y(x^{1}_n,\ldots,x^{n}_n,n)
\bigr)
\given x^{1},\ldots,x^{n}\in D^k
}.
\end{align*}

Intuitively, $D^k$ encodes the boundary of states from which an optimal adversary can force a violation within $k$
rounds.

\subsection*{Geometric characterization}
The recursive sets \(D^k\) are useful because they exactly describe the boundary between safe and unsafe states.
The next theorem makes this precise.
\begin{theorem}[Geometric characterization of safe states]
\label{thm:Dk-characterization}
For every $k\ge 0$, a point $\vec{x}'\in \overline{\mathbb{R}}^n$ is $k$-safe if and only if it is \emph{not}
dominated by any point in $D^k$.
\end{theorem}

\begin{proof}
We prove the claim by induction on $k$.

\paragraph{Base case ($k=0$).}
A point $\vec{x}'$ is $0$-safe iff it is already feasible now, i.e.\ $x'_i\ge 0$ for all $i$.
This holds iff $\vec{x}'$ is not dominated by any vector in $D^0$ (since domination by the $i$th vector in $D^0$
means $x'_i<0$).

\paragraph{Inductive step.}
Assume the claim holds for $k$ and prove it for $k+1$.

\par\noindent
\emph{($\Rightarrow$) dominated $\Longrightarrow$ not $(k{+}1)$-safe.}
Assume $\vec{x}'$ is dominated by some $\vec{x}\in D^{k+1}$.
By construction of $D^{k+1}$, there exist points $\vec{x}^{\,1},\ldots,\vec{x}^{\,n}\in D^k$ such that
\[
\vec{x}=
\bigl(
Y(x^{1}_1,\ldots,x^{n}_1,1),\ldots,Y(x^{1}_n,\ldots,x^{n}_n,n)
\bigr).
\]
Consider a time $t$ with $\vec{\dsprop}^{\,t}=\vec{x}'$, so $\dsprop_i^t=x'_i<x_i$ for all $i$.

In round $t+1$, the adversary chooses an item $g^{t+1}$ with values
\[
v_i(g^{t+1}) \ :=\ Z(x^{1}_i,\ldots,x^{n}_i,i)
\qquad\text{for each }i\in[n].
\]
Suppose the algorithm allocates the item to some agent $j$.
Using \eqref{eq:shifted-surplus-update} and writing $z_i:=v_i(g^{t+1})$, we get:
\[
\dsprop_j^{t+1}=x'_j+(n-1)z_j
\quad\text{and}\quad
\dsprop_i^{t+1}=x'_i-z_i\ \ (i\neq j).
\]
By $x'_i<x_i$ and by the LP constraints defining $Y(\cdot,i)$ and $Z(\cdot,i)$, we have
\[
x'_j+(n-1)z_j < x_j+(n-1)z_j = Y(x^{1}_j,\ldots,x^{n}_j,j)+(n-1)Z(x^{1}_j,\ldots,x^{n}_j,j)\le x^{\,j}_j,
\]
and for $i\neq j$,
\[
x'_i-z_i < x_i-z_i = Y(x^{1}_i,\ldots,x^{n}_i,i)-Z(x^{1}_i,\ldots,x^{n}_i,i)\le x^{\,j}_i.
\]
Hence $\vec{\dsprop}^{\,t+1}$ is dominated by $\vec{x}^{\,j}\in D^k$.
By the induction hypothesis, $\vec{\dsprop}^{\,t+1}$ is not $k$-safe, so $\vec{x}'$ is not $(k+1)$-safe.

\par\noindent
\emph{($\Leftarrow$) not $(k{+}1)$-safe $\Longrightarrow$ dominated.}
Suppose $\vec{x}'$ is not $(k+1)$-safe.
Then when $\vec{\dsprop}^{\,t}=\vec{x}'$, there exists an adversarial item $g^{t+1}$ (with some values
$z_i:=v_i(g^{t+1})\in[0,1]$) such that \emph{for every} allocation choice $j\in[n]$, the resulting state
$\vec{\dsprop}^{\,t+1}$ is not $k$-safe.
By the induction hypothesis, for each $j\in[n]$ there exists a point $\vec{x}^{\,j}\in D^k$ that dominates
the state obtained when the item is allocated to $j$.

Fix an index $i\in[n]$ and consider the coordinate-wise inequalities implied by domination.
If the item is allocated to agent $i$, then (by \eqref{eq:shifted-surplus-update})
\[
x'_i+(n-1)z_i=\dsprop_i^{t+1}<x^{\,i}_i.
\]
If the item is allocated to some $j\neq i$, then $\dsprop_i^{t+1}=x'_i-z_i<x^{\,j}_i$.
Therefore the pair $(y,z)=(x'_i,z_i)$ satisfies the LP constraints for index $i$ with RHS
$(x^{1}_i,\ldots,x^{n}_i)$:
\[
y+(n-1)z \le x^{\,i}_i,\qquad y-z \le x^{\,j}_i\ (j\neq i),\qquad 0\le z\le 1.
\]
Since $Y(x^{1}_i,\ldots,x^{n}_i,i)$ maximizes $y$ over these constraints, we get
\[
x'_i=y \ <\ Y(x^{1}_i,\ldots,x^{n}_i,i).
\]
Doing this for all $i\in[n]$ shows that $\vec{x}'$ is dominated by the point
\[
\bigl(
Y(x^{1}_1,\ldots,x^{n}_1,1),\ldots,Y(x^{1}_n,\ldots,x^{n}_n,n)
\bigr)\in D^{k+1},
\]
as required.
\end{proof}

\subsection*{The \textsc{AUX} procedure and the $\mathsf{EXP}$ rule}

\paragraph{The \textsc{AUX} procedure.}
Given a state $\vec{x}$, \textsc{AUX} computes the smallest $k$ such that $\vec{x}$ is not $k$-safe,
by explicitly constructing $D^0,D^1,\ldots$ until some point dominates $\vec{x}$.

\begin{algorithm}[ht]
\caption{\textsc{AUX}$(\vec{x})$: minimum rounds until a forced violation from state $\vec{x}$}
\label{alg:aux-prop}
\KwIn{$\vec{x}=(x_1,\ldots,x_n)\in\overline{\mathbb{R}}^n$}
\KwOut{The minimum $k$ such that $\vec{x}$ is not $k$-safe}
$k\gets 0$\;
$prePoints \gets D^0$\;
\While{$\forall\,\vec{p}\in prePoints:\ \neg(\vec{p}\text{ dominates }\vec{x})$}{
    $nextPoints \gets \emptyset$\;
    \ForEach{$(\vec{p}^{\,1},\ldots,\vec{p}^{\,n})\in (prePoints)^n$}{
        $\vec{q}\gets
        \bigl(
        Y(p^{1}_1,\ldots,p^{n}_1,1),\ldots,Y(p^{1}_n,\ldots,p^{n}_n,n)
        \bigr)$\;
        $nextPoints \gets nextPoints \cup \{\vec{q}\}$\;
    }
    $prePoints \gets nextPoints$\;
    $k\gets k+1$\;
}
\Return $k$\;
\end{algorithm}

\begin{theorem}[Correctness of \textsc{AUX}]
\label{thm:aux-correct}
If the current shifted-surplus vector is $\vec{\dsprop}^{\,t}=\vec{x}$, then \Cref{alg:aux-prop} returns
the minimum number of further rounds $k$ such that $\vec{x}$ is not $k$-safe (equivalently: such that an
optimal adversary can force a violation within $k$ rounds).
\end{theorem}

\begin{proof}
\Cref{alg:aux-prop} constructs the sets $D^k$ in increasing order of $k$ and stops at the first $k$ for which
some point in $D^k$ dominates $\vec{x}$. By \Cref{thm:Dk-characterization}, this first such $k$ is exactly the
minimum $k$ for which $\vec{x}$ is not $k$-safe.
\end{proof}

\paragraph{The $\mathsf{EXP}$ policy.}
Given an arriving item, $\mathsf{EXP}$ evaluates each possible recipient choice $i\in[n]$ by the
\textsc{AUX}-value of the \emph{post-round} state, and allocates to maximize that value.

\begin{algorithm}[t]
\caption{$\mathsf{EXP}$: optimal policy for maximizing the time to violation (unit-scale surrogate)}
\label{alg:exp-prop}
\KwIn{Current shifted-surplus vector $\vec{\dsprop}^{\,t}$; arriving item values $(v_1(g^{t+1}),\ldots,v_n(g^{t+1}))\in[0,1]^n$}
\KwOut{Recipient $a^{t+1}\in[n]$}
\ForEach{$i\in[n]$}{
    compute the hypothetical next state $\vec{\dsprop}^{\,t+1,i}$ if $g^{t+1}$ is given to $i$ using
    \eqref{eq:shifted-surplus-update}\;
    $\tau_i \gets \textsc{AUX}(\vec{\dsprop}^{\,t+1,i})$\;
}
$a^{t+1}\in \argmax_{i\in[n]} \tau_i$\;
\Return $a^{t+1}$\;
\end{algorithm}

\paragraph{Complexity note.}
The construction of $D^{k+1}$ iterates over all $n$-tuples (with repetition) of points from $D^k$, and for each
such tuple solves $n$ linear programs. Thus the runtime grows extremely quickly with $k$ (and with $n$), but is
practical for small $n$ and moderate $k$.

\subsection*{Two-agent illustration}

For intuition, consider $n=2$.  The figure below illustrates (schematically) how the $(\dsprop_1^t,\dsprop_2^t)$
plane can be partitioned into regions according to the number of rounds needed for an optimal adversary to force a
violation.

\begin{figure}[ht]
    \centering
\scalebox{1.10}{
\begin{tikzpicture}[node/.style={font=\small}]
    \draw[->] (-0.5,0) -- (4.2,0) node[right] {$\dsprop_1^t$};
    \draw[->] (0,-0.5) -- (0,4.2) node[above] {$\dsprop_2^t$};

    \fill[blue!30] (0,0) -- (0,1) -- (1,1) -- (1,0) -- cycle;
    \node at (0.5,0.5) {1};

    \fill[red!30] (1,0) -- (2,0) -- (2,0.5) -- (1,0.5) -- cycle;
    \node at (1.5,0.25) {2};

    \fill[red!30] (0,1) -- (0,2) -- (0.5,2) -- (0.5,1) -- cycle;
    \node[node] at (0.25,1.5) {2};

    \fill[green!30] (2,0) -- (3,0) -- (3,0.25) -- (2,0.25) -- cycle;
    \node[node, font=\tiny] at (2.5,0.125) {3};

    \fill[green!30] (0,2) -- (0,3) -- (0.25,3) -- (0.25,2) -- cycle;
    \node[node, font=\tiny] at (0.125,2.5) {3};

    \fill[green!30] (1,1) -- (1.5,1) -- (1.5,0.5) -- (1,0.5) -- cycle;
    \node at (1.25,0.75) {3};

    \fill[green!30] (0.5,1.5) -- (1,1.5) -- (1,1) -- (0.5,1) -- cycle;
    \node at (0.75,1.25) {3};

    \fill[green!30] (1,1) -- (1,1.25) -- (1.25,1.25) -- (1.25,1) -- cycle;
    \node[node, font=\tiny] at (1.125,1.125) {3};

    \node[below left, font=\tiny] at (0,0) {$(0,0)$};
    \node[left, font=\tiny] at (0,1) {$(0,1)$};
    \node[below, font=\tiny] at (1,0) {$(1,0)$};
    \node[below, font=\tiny] at (2,0) {$(2,0)$};
    \node[left, font=\tiny] at (0,2) {$(0,2)$};
\end{tikzpicture}}
\caption{Schematic partition: each rectangle is labeled by an upper bound on the number of rounds needed for an
optimal adversary to force a violation of \eqref{eq:unit-scale-prop} (for fixed $c$, after shifting).}
\Description{A two-dimensional plot with axes labeled delta one at time t and delta two at time t. Several colored rectangles in the nonnegative quadrant are labeled 1, 2, and 3, indicating the number of rounds until a violation from states in each region.}
\label{fig:aux-schematic-n2}
\end{figure}

\paragraph{A concrete computed bound (example).}
When $n=2$ and $c=1$, the initial shifted-surplus state is
$\vec{\dsprop}^{\,0}=(nc,nc)=(2,2)$.
Our implementation of \Cref{alg:aux-prop} (enumerating $D^k$ for $k=0,1,\ldots$) finds that some point in $D^{10}$
dominates $(2,2)$, while no point in $D^9$ does.  Hence the minimum $k$ such that $(2,2)$ is not $k$-safe is $k=10$:
an optimal adversary can force a violation of \eqref{eq:unit-scale-prop} within $10$ rounds, and this is tight for
this surrogate game.

\begin{figure}[ht]
   \centering
   \includegraphics[width=0.40\textwidth]{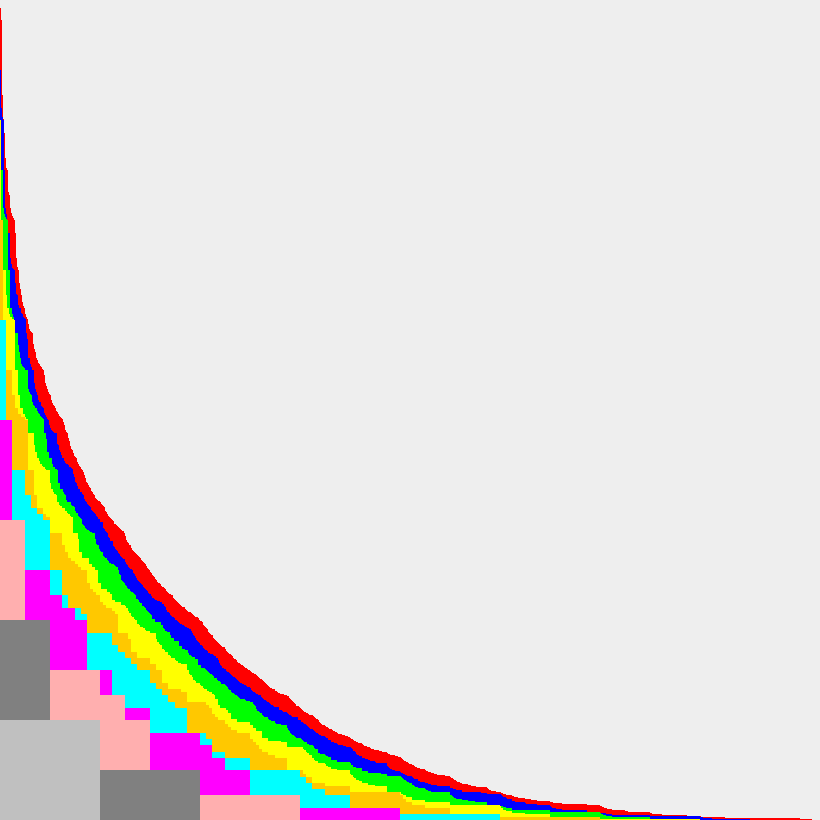}
  \caption{Empirical visualization of the dominated regions induced by $D^k$ for $k=0,\ldots,10$ when $n=2$.}
    \Description{A region plot for the two-agent state space. Each color corresponds to a value of k from 0 to 10,
  indicating states dominated by some point in D to the power k.}
  \label{fig:Dk-n2}
\end{figure}

\clearpage
\section{An unrestricted lower bound for classical
\texorpdfstring{$\EFc$}{EFc}}
\label{app:classical-efc-lower-bound}

A restriction specified in advance on the possible item values is necessary for obtaining a sublinear
classical \(\EFc\) guarantee against an adaptive adversary. Indeed, even when
all values lie in \((0,1]\), the exact unrestricted worst-case guarantee is
linear.

\begin{theorem}[Lower bound for classical
\texorpdfstring{$\EFc$}{EFc}]
\label{thm:classical-efc-linear-lb}
For every \(n\ge2\) and \(T\ge1\), every online allocation algorithm,
deterministic or randomized, can be forced by a adaptive
adversary so that all revealed values satisfy
\[
0<v_i(g^t)\le1
\qquad (i\in[n],\ t\in[T]),
\]
and the allocation at time \(T\) is not \(\EFc\) for any integer
\[
c<\left\lceil\frac{T}{n}\right\rceil.
\]
The lower bound holds even if the horizon \(T\) is known to the algorithm.
Conversely, round-robin guarantees
\(\EFparam{\lceil t/n\rceil}\) after every prefix \(t\). Hence the exact
unrestricted worst-case parameter at time \(T\) is
\(\lceil T/n\rceil\).
\end{theorem}

\begin{proof}
We first prove the lower bound. For each agent \(i\), maintain two numbers
\(L_i^t\le U_i^t\). Initially, let
\[
L_i^0:=(T+1)^{-2^T},
\qquad
U_i^0:=1.
\]
In round \(t+1\), before the recipient is chosen, reveal the item
\(g^{t+1}\) at the geometric midpoint\footnote{Geometric-mean thresholds over
multiplicative ranges are classical in online search; see
\citet{elyaniv1992competitive}. Here, a separate interval is maintained for
each agent and refined according to the allocation decision.}:
\begin{equation}
\label{eq:efc-geometric-midpoint}
v_i(g^{t+1}):=\sqrt{L_i^tU_i^t}
\qquad (i\in[n]).
\end{equation}
After the algorithm chooses \(a^{t+1}\in[n]\), update
\begin{equation}
\label{eq:efc-geometric-update}
(L_i^{t+1},U_i^{t+1})
:=
\begin{cases}
\bigl(v_i(g^{t+1}),U_i^t\bigr), & a^{t+1}=i,\\[1mm]
\bigl(L_i^t,v_i(g^{t+1})\bigr), & a^{t+1}\neq i.
\end{cases}
\end{equation}
The values in \eqref{eq:efc-geometric-midpoint} depend only on the realized
history through time \(t\). Moreover,
\[
0<L_i^t\le v_i(g^{t+1})\le U_i^t\le1,
\]
and hence all item values lie in \((0,1]\).

We claim that, for every \(t\in\{0,\ldots,T\}\) and every agent \(i\),
and for every \(s\in\{1,\ldots,t\}\),
\begin{align}
g^s\in P_i^t
&\quad\Longrightarrow\quad
v_i(g^s)\le L_i^t,
\label{eq:efc-own-items-low}\\
g^s\in G^t\setminus P_i^t
&\quad\Longrightarrow\quad
v_i(g^s)\ge U_i^t,
\label{eq:efc-missed-items-high}\\
\frac{U_i^t}{L_i^t}
&=(T+1)^{2^{T-t}}.
\label{eq:efc-ratio-invariant}
\end{align}
The claim follows by induction on \(t\). It is immediate at \(t=0\). Suppose
it holds at time \(t\), and write
\[
x:=v_i(g^{t+1})=\sqrt{L_i^tU_i^t}.
\]

If \(a^{t+1}=i\), then
\[
L_i^{t+1}=x,
\qquad
U_i^{t+1}=U_i^t.
\]
All earlier received items remain at most \(L_i^t\le x\), the new received
item has value \(x\), and all missed items remain at least \(U_i^t\).

If \(a^{t+1}\neq i\), then
\[
L_i^{t+1}=L_i^t,
\qquad
U_i^{t+1}=x.
\]
All received items remain at most \(L_i^t\), the new missed item has value
\(x\), and all earlier missed items remain at least \(U_i^t\ge x\).

Thus \eqref{eq:efc-own-items-low} and
\eqref{eq:efc-missed-items-high} are preserved. In either case,
\[
\frac{U_i^{t+1}}{L_i^{t+1}}
=
\sqrt{\frac{U_i^t}{L_i^t}},
\]
which proves \eqref{eq:efc-ratio-invariant} and completes the induction.

At time \(T\), \eqref{eq:efc-ratio-invariant} gives
\[
\frac{U_i^T}{L_i^T}=T+1.
\]
Therefore, by
\eqref{eq:efc-own-items-low}--\eqref{eq:efc-missed-items-high}, for every
received item \(g\in P_i^T\) and every missed item
\(h\in G^T\setminus P_i^T\),
\begin{equation}
\label{eq:efc-final-separation}
v_i(h)\ge (T+1)v_i(g).
\end{equation}

Let
\[
j:=\min\operatorname*{arg\,max}_{a\in[n]} |P_a^T|,
\qquad \text{and} \qquad
m:=|P_j^T|.
\]
By the pigeonhole principle,
\[
m\ge\left\lceil\frac{T}{n}\right\rceil.
\]
Choose any \(i\neq j\), and let
\[
H:=\min_{h\in P_j^T}v_i(h).
\]
This minimum is well defined and positive because
\(m\ge \lceil T/n\rceil\ge1\) and all revealed values are positive.
Since every item in \(P_j^T\) is missed by \(i\),
\eqref{eq:efc-final-separation} implies that every \(g\in P_i^T\) satisfies
\[
v_i(g)\le\frac{H}{T+1}.
\]
Consequently,
\[
v_i(P_i^T)
\le
|P_i^T|\frac{H}{T+1}
\le
\frac{TH}{T+1}
<
H.
\]

Now let \(B\subseteq P_j^T\) satisfy \(|B|\le m-1\). At least one item of
\(P_j^T\) remains after removing \(B\), and every item in \(P_j^T\) has
\(v_i\)-value at least \(H\). Therefore,
\[
v_i(P_j^T\setminus B)
\ge H
>
v_i(P_i^T).
\]
Thus agent \(i\) still envies agent \(j\) after removing any \(m-1\) or fewer
items from \(P_j^T\). Hence the allocation is not \(\EFc\) for any
\(c<m\), and therefore not for any
\[
c<\left\lceil\frac{T}{n}\right\rceil.
\]

The construction defines a single adversarial policy, and the
argument applies to every realized sequence of recipients. Hence the lower
bound holds for every realization of a randomized algorithm's internal
randomness.

For the matching upper bound, allocate items in round-robin order. At every
time \(t\), each bundle contains at most
\[
\left\lceil\frac{t}{n}\right\rceil
\]
items. For every ordered pair \(i\neq j\), removing all of \(P_j^t\) removes
at most \(\lceil t/n\rceil\) items and leaves
\[
v_i(P_j^t\setminus P_j^t)
=
v_i(\emptyset)
=
0
\le
v_i(P_i^t).
\]
Thus round robin is \(\EFparam{\lceil t/n\rceil}\) after every prefix,
which proves tightness.
\end{proof}

\begin{remark}[Relation to the known-value guarantee]
\label{rem:classical-efc-value-universe}
The lower bound does not contradict
\Cref{cor:efc-classical-threshold}. The realized sequence contains only
finitely many values, but these values are selected adaptively from the
allocation history. To specify in advance a single value set containing every
value that the adversary may reveal under every possible allocation history,
one must include all values represented by the construction's decision tree.
For any fixed agent, a straightforward induction shows that the set of all
midpoint values that may be revealed during rounds \(1,\ldots,T\), over all
possible histories, is
\[
\left\{(T+1)^{-q}:q=1,\ldots,2^T-1\right\}.
\]
These values are pairwise distinct. Hence any value set specified in advance
that contains every possible revealed value must satisfy
\[
L\ge 2^T-1.
\]
Consequently, \(\log L=\Theta(T)\), and the guarantee in
\Cref{cor:efc-classical-threshold} becomes linear in \(T\), consistently with
\Cref{thm:classical-efc-linear-lb}.
\end{remark}


\end{document}